%% file: mainArxiv.tex
\documentclass{article}

 \usepackage{geometry}
 \usepackage{fullpage}

\def\showauthornotes{1}
\def\showdraftbox{0}

\input{macros.tex}

\usepackage{subfig}
\usepackage{hyperref}
\usepackage{pdfsync}
\usepackage{wrapfig}

\newcommand{\eps}{\varepsilon}

\newcommand\blfootnote[1]{%
  \begingroup
  \renewcommand\thefootnote{}\footnote{#1}%
  \addtocounter{footnote}{-1}%
  \endgroup
}

\usepackage[
  backend=biber,
  backref=true,
  backrefstyle=none,
  date=year,
  doi=false,
  giveninits=true,
  hyperref=true,
  maxbibnames=10,
  sortcites=false,
  style=alphabetic,
  url=false, 
]{biblatex}
\title{Fast Algorithms for \texorpdfstring{$\ell_p$}{TEXT}-Regression}

\author{%
 Deeksha Adil\footnote{Supported by a Post Graduate Schlarship (PGSD) by NSERC (Natural Sciences and Engineering Research Council of Canada) and Sushant Sachdeva's NSERC Discovery Grant.}\\
  University of Toronto\\
  \texttt{deeksha@cs.toronto.edu}
   \and
Rasmus Kyng\footnote{The research leading to these results has received funding from the grant “Algorithms and complexity for high-accuracy flows and convex optimization” (no. 200021 204787) of the Swiss National Science Foundation.} \\
ETH Zurich \\
  \texttt{kyng@inf.ethz.ch}\\ 
    \and
  Richard Peng\footnote{This material is based upon work supported by the National Science Foundation under Grant No. 1846218, and by an Natural Sciences and Engineering Research Council of Canada Discovery Grant. Part this work was done while the author was at the Georgia Institute of Technology.} \\
  University of Waterloo\\
  \texttt{y5peng@uwaterloo.ca} 
  \and
  Sushant Sachdeva\footnote{Sushant Sachdeva’s research is supported by an NSERC (Natural Sciences and Engineering Research Council of
Canada) Discovery Grant.} \\
  University of Toronto \\
  \texttt{sachdeva@cs.toronto.edu}
}

\begin{document}

\maketitle

\begin{abstract}
  The $\ell_p$-norm regression
  problem is a classic problem in optimization with wide ranging
  applications in machine learning and theoretical computer
  science. The goal is to compute
  $\xx^{\star} =\arg\min_{\AA\xx=\bb}\|\xx\|_p^p$, where
  $\xx^{\star}\in \mathbb{R}^n,\AA\in \mathbb{R}^{d\times n},\bb \in
  \mathbb{R}^d$ and $d\leq n$. Efficient {\it high-accuracy}
  algorithms for the problem have been challenging both in theory and
  practice and the state of the art algorithms require
  $poly(p)\cdot n^{\frac{1}{2}-\frac{1}{p}}$ linear system solves for
  $p\geq 2$. In this paper, we provide new algorithms for
  $\ell_p$-regression (and a more general formulation of the problem)
  that obtain a {\it high-accuracy} solution in
  $O(p n^{\nfrac{(p-2)}{(3p-2)}})$
  linear system solves. We further propose a new {\it inverse
    maintenance} procedure that speeds-up our algorithm to
  $\Otil(n^{\omega})$ total runtime, where $O(n^{\omega})$ denotes the
  running time for multiplying $n \times n$ matrices. Additionally, we
  give the first {\it Iteratively Reweighted Least Squares (IRLS)}
  algorithm that is guaranteed to converge to an optimum in a few
  iterations. Our IRLS algorithm has shown exceptional practical
  performance, beating the currently available implementations in
  MATLAB/CVX by 10-50x.
\end{abstract}

\input{Chapters/Introduction}
\input{Chapters/IterRefinement}
\input{Chapters/General_Oracle}
\input{Chapters/pNorm2qNorm}

\input{Chapters/InverseMaintenance}
\input{Chapters/IRLS}

\printbibliography

\appendix
\input{Chapters/Appendix-l2Prob}
\input{Chapters/Appendix-EnergyLemma}

\end{document}

%% file: macros.tex
 \usepackage{hyperref}
\usepackage{amstext, amsmath,amsthm, bbm,xspace,nicefrac,graphicx, url,bbold,xcolor}
\usepackage{algorithm} 
\usepackage[noend]{algpseudocode}
\usepackage{thmtools}
\usepackage{thm-restate}

\newcommand{\defeq}{\stackrel{\textup{def}}{=}}

\newtheorem{theorem}{Theorem}[section]

\newtheorem{lemma}[theorem]{Lemma}
\newtheorem{corollary}[theorem]{Corollary}

\newtheorem{definition}[theorem]{Definition}

\newcommand{\nfrac}[2]{\nicefrac{#1}{#2}}
\def\abs#1{\left| #1 \right|}
\newcommand{\norm}[1]{\ensuremath{\left\lVert #1 \right\rVert}}

\newcommand{\ceil}[1]{\left\lceil\, {#1}\,\right\rceil}

\newcommand\rea{\mathbb R}

\newcommand{\marginlabel}[1]%
{\mbox{}\marginpar{\it{\raggedleft\hspace{0pt}#1}}}
\newcommand\poly{\text{poly}}  

\definecolor{Mygray}{gray}{0.8}

 \ifcsname ifcommentflag\endcsname\else
  \expandafter\let\csname ifcommentflag\expandafter\endcsname
                  \csname iffalse\endcsname
\fi

\ifnum\showauthornotes=1
\newcommand{\Authornote}[2]{{\sf\small\color{red}{[#1: #2]}}}
\newcommand{\Authoredit}[2]{{\sf\small\color{red}{[#1]}\color{blue}{#2}}}
\newcommand{\Authorcomment}[2]{{\sf \small\color{gray}{[#1: #2]}}}
\newcommand{\Authorfnote}[2]{\footnote{\color{red}{#1: #2}}}
\newcommand{\Authorfixme}[1]{\Authornote{#1}{\textbf{??}}}
\newcommand{\Authormarginmark}[1]{\marginpar{\textcolor{red}{\fbox{
#1:!}}}}
\else
\newcommand{\Authornote}[2]{}
\newcommand{\Authoredit}[2]{}
\newcommand{\Authorcomment}[2]{}
\newcommand{\Authorfnote}[2]{}
\newcommand{\Authorfixme}[1]{}
\newcommand{\Authormarginmark}[1]{}
\fi

\def\iff{\Leftrightarrow}

\newlength{\pgmtab}  
\let\originalleft\left
\let\originalright\right
\renewcommand{\left}{\mathopen{}\mathclose\bgroup\originalleft}
  \renewcommand{\right}{\aftergroup\egroup\originalright}

\def\defeq{\stackrel{\mathrm{def}}{=}}
\def\setof#1{\left\{#1  \right\}}

\def\ceil#1{\left\lceil #1 \right\rceil}

\def\union{\cup}

\def\aa{\pmb{\mathit{a}}}
\newcommand\bb{\boldsymbol{\mathit{b}}}
\newcommand\cc{\boldsymbol{\mathit{c}}}
\newcommand\dd{\boldsymbol{\mathit{d}}}

\newcommand\ff{\boldsymbol{\mathit{f}}}

\renewcommand\gg{\boldsymbol{\mathit{g}}}

\newcommand\rr{\boldsymbol{\mathit{r}}}
\renewcommand\ss{\boldsymbol{\mathit{s}}}
\def\tt{\boldsymbol{\mathit{t}}}

\newcommand\vvs{\boldsymbol{\mathit{v}}}
\newcommand\ww{\boldsymbol{\mathit{w}}}
\newcommand\yy{\boldsymbol{\mathit{y}}}

\newcommand\xx{\boldsymbol{\mathit{x}}}

\newcommand\xxtil{\widetilde{\boldsymbol{\mathit{x}}}}

\renewcommand\AA{\boldsymbol{\mathit{A}}}
\newcommand\BB{\boldsymbol{\mathit{B}}}
\newcommand\CC{\boldsymbol{\mathit{C}}}

\newcommand\II{\boldsymbol{\mathit{I}}}

\newcommand\NN{\boldsymbol{\mathit{N}}}
\newcommand\MM{\boldsymbol{\mathit{M}}}

\newcommand\RR{\boldsymbol{\mathit{R}}}
\renewcommand\SS{\boldsymbol{\mathit{S}}}
\newcommand\UU{\boldsymbol{\mathit{U}}}
\newcommand\WW{\boldsymbol{\mathit{W}}}
\newcommand\VV{\boldsymbol{\mathit{V}}}
\newcommand\XX{\boldsymbol{\mathit{X}}}

\newcommand\ZZ{\boldsymbol{\mathit{Z}}}

\newcommand\rrhat{\boldsymbol{\widehat{\mathit{r}}}}

\newcommand\Otil{\widetilde{O}}

\newenvironment{tight_enumerate}{
\begin{enumerate}
 \setlength{\itemsep}{2pt}
 \setlength{\parskip}{1pt}
}{\end{enumerate}}
\newenvironment{tight_itemize}{
\begin{itemize}
 \setlength{\itemsep}{2pt}
 \setlength{\parskip}{1pt}
}{\end{itemize}}

\newcommand\Dtil{{\widetilde{{\Delta}}}}
\newcommand\Dbar{{\bar{{\Delta}}}}
\newcommand\Dopt{{{{\Delta^{\star}}}}}
\newcommand{\res}{\boldsymbol{res}}
\newcommand\Mtil{{\widetilde{\boldsymbol{\mathit{M}}}}}

 {
	\begin{enumerate}}{\end{enumerate}}

\ifnum\showdraftbox=1

\else

\fi

%% file: Chapters/Introduction.tex
\section{Introduction}

\blfootnote{Preliminary
 versions of the results in this paper have appeared as conference
 publications~\cite{AdilKPS19, AdilPS19, AdilS20, AdilBKS21}. This
 paper unifies and simplifies results from the preliminary versions. }Linear regression in $\ell_p$-norm seeks to compute a vector
$\xx^{\star}\in \mathbb{R}^n$ such that,
 \[
 \xx^{\star} = \arg\min_{\AA\xx=\bb}\|\xx\|^p_p,
 \]
 where $\AA\in \mathbb{R}^{d\times n}, \bb \in \mathbb{R}^{d}$, $d\leq
 n$.
 This is a classic convex optimization problem that captures several
 well-studied questions including least squares regression ($p =2$)
 which is equivalent to solving a system of linear equations, and
 linear programming ($p=\infty$).
 The $\ell_p$-norm regression problem for $p>1$ has found
 use across a wide range of applications in machine learning and
 theoretical computer science including low rank matrix approximation
 \cite{chierichetti2017algorithms}, sparse recovery
 \cite{candes2005decoding}, graph based semi-supervised learning
 \cite{alamgir2011phase,calder2019consistency,rios2019, KyngRSS15}, data
 clustering and learning problems
 \cite{elmoataz2015p,elmoataz2017game,hafiene2018nonlocal}. 
 In this paper, we focus on solving the $\ell_p$-norm regression
 problem for $p\geq 2.$ The exact solution to the $\ell_p$-norm
 regression problem for $p \neq 1,2,\infty,$ may not even be
 expressible using rationals.
 Thus, the goal is often relaxed to finding an $\eps$-approximate solution to
 the problem, i.e., find $\hat{\xx}$ such that $\AA\hat{\xx} = \bb$
 and,
 \[
   \|\hat{\xx}\|_p^p \leq (1+\epsilon)\|\xx^{\star}\|_p^p,
 \]
 for some small $\eps > 0.$
 Furthermore, several applications such as graph based semi-supervised
 learning require that $\hat{\xx}$ is close to $\xx^{\star}$
 coordinate-wise and not just in objective value -- necessitating a
 \textit{high-accuracy} solution with
 $\epsilon \approx \frac{1}{\poly(n)}$.
 In order to find such high-accuracy solutions efficiently, we require
 an algorithm with runtime dependence on $\epsilon$ being
 $\poly\left( \log \frac{1}{\epsilon} \right)$ rather than
 $\poly\left( \frac{1}{\epsilon} \right).$

Fast, high-accuracy algorithms for $\ell_p$-regression are
challenging both in theory and practice, due to the lack of smoothness
and strong convexity of the objective.
The Interior Point Method framework by \textcite{nesterov1994interior}
can be used to compute a high-accuracy solution for all
$p\in [1,\infty]$ in $\Otil(\sqrt{n})$\footnote{$\Otil$ hides
  constants, $p$ dependencies, $\log \frac{1}{\epsilon},$ and $\log n$
  factors unless explicitly mentioned} iterations, with each iteration
requiring solving an $n \times n$ system of linear equations. This was
the most efficient algorithm for $\ell_p$-regression until 2018.
In 2018, \textcite{bubeck2018homotopy} showed that $\Omega(\sqrt{n})$
iterations are necessary for the interior point framework and proposed
a new homotopy-based approach that could compute a high-accuracy
solution in $\Otil(n^{\left|\frac{1}{2}-\frac{1}{p}\right|})$
linear system solves for all $p \in (1,\infty)$. Their algorithms
improve over the interior point method by $n^{\Omega(1)}$ factors for
values of $p$ bounded away from $1$ and $\infty.$ However, for $p$
approaching $1$ or $\infty$, the number of linear system solves
required by their algorithm approaches $\sqrt{n},$ the same as
required by interior point methods. Finding an algorithm for
$\ell_p$-regression requiring $o(n^{1/2})$ linear system solves has
been a long standing open problem.

Among practical implementations for the $\ell_p$-norm regression
problem, the {\it Iteratively Reweighted Least Squares (IRLS)} methods
stand out due to their simplicity, and have been studied since
1961~\cite{lawson1961}.
For some range of values for $p,$ IRLS converges rapidly.
However, the method is guaranteed to converge only for $p \in (1.5,3)$
and diverges even for small values of $p,$ e.g.  $p = 3.5$
\cite{rios2019}.
Over the years, several empirical modifications
of the algorithm have been used for various applications in practice
(refer to \cite{burrus2012} for a full survey). However, an IRLS
algorithm that is guaranteed to converge to the optimum in a few
iterations for all values of $p,$ has again been a long standing
challenge.

\subsection{Our Contributions}

In this paper, we present the first algorithm for the $\ell_p$-regression problem that finds a high-accuracy solution in at most $O(p n^{1/3})=o(n^{1/2})$ linear system solves, which has been a long sought-after goal in optimization. Our algorithm builds on a new {\it iterative refinement} framework for $\ell_p$-norm objectives that allows us to find a high-accuracy solution using low-accuracy solutions to a subproblem. The iterative refinement framework allows for the subproblems to be solved to an $n^{o(1)}$-approximation and this has been useful in several follow up works on graph optimization (see Section~\ref{sec:RelatedWorks}). We further propose a new {\it inverse maintenance} framework and show how to speed up our algorithm to solve the $\ell_p$-norm problem to a high-accuracy in total time $\Otil(n^{\omega})$. Finally, we give the first IRLS algorithm that provably converges to a high-accuracy solution in a few iterations. 

Preliminary versions of the results presented in this paper have
appeared in previous conference publications by \textcite{AdilKPS19,AdilPS19, AdilS20, AdilBKS21}.
In this paper, we present our results for a more general formulation of the $\ell_p$-regression problem,
\begin{equation}\label{eq:Main}
\min_{\AA\xx = \bb}\quad \ff(\xx) = \dd^{\top}\xx + \|\MM\xx\|_2^2 + \|\NN\xx\|_p^p
\end{equation}
for matrices $\AA\in \mathbb{R}^{d\times n}, \MM \in \mathbb{R}^{m_1\times n}, \NN \in \mathbb{R}^{m_2\times n},$ $m_1,m_2 \geq n, d\leq n$. Let $m = \max\{m_1,m_2\}$ and, $\dd \perp \{ker(\MM) \cap ker(\NN) \cap ker(\AA)\}$, $\bb \in im(\AA)$ so that the above problem has a bounded solution.  Our first result is a fast, high-accuracy algorithm for Problem~\eqref{eq:Main}.
\begin{restatable}{theorem}{MainThm}\label{thm:MainThm}
Let $\epsilon >0 $ and $p \geq 2$. There is an algorithm that starting from $\xx^{(0)}$ satisfying $\AA\xx^{(0)}=\bb$, finds an $\epsilon$-approximate solution to Problem \eqref{eq:Main} in $O\left(p m^{\frac{p-2}{3p-2}}\log \frac{\ff(\xx^{(0)})-\ff(\xx^{\star})}{\epsilon }\right)$ calls to a linear system solver.
\end{restatable}

As a corollary, for the $\ell_p$-norm regression problem, i.e.,
$\dd = \MM = 0$ and $\NN = \II$, our algorithm converges in
$O(p n^{\frac{p-2}{3p-2}}\log \frac{n}{\epsilon})$ calls to a linear
system solver. This is the first algorithm that converges to a high
accuracy solution at an asymptotic rate of convergence
$\Otil(n^{1/3}) = o(n^{1/2})$ for all $p\in [2,\infty)$, and
  thus faster than all previously known algorithms by at least a
  factor of $n^{\Omega(1)}$. As a result, we answer the long standing
  problem in optimization of whether such a rate of convergence could
  be achieved.

Our next result shows how to speed up our algorithms
and solve Problem \eqref{eq:Main} in time $\Otil(m^{\omega})$ (or
$\Otil(n^{\omega})$ for $\ell_p$-regression), where $\omega \approx 2.37$ and $O(n^{\omega})$ is
the current time required for multiplying two $n \times n$ matrices. This is
almost as fast as solving a system of linear equations. We achieve
this guarantee via a new \textit{inverse maintenance} procedure for
$\ell_p$-regression and prove the following result.
\begin{theorem}\label{thm:InverseIntro}
If $\AA,\MM,\NN$ are explicitly given, matrices with polynomially bounded condition numbers,
and $p \geq 2$, there is an algorithm for Problem \eqref{eq:Main} that can
be implemented to run in total time $\Otil(m^{\omega})$. 
\end{theorem}
 Our inverse maintenance algorithm is presented in Section \ref{sec:Inv}, where we also give a more fine grained dependence on the parameters $m_1,m_2,n$ and $p$ in the rate of convergence (Theorem \ref{thm:InverseMaint}). Our algorithms and techniques for $\ell_p$-regression have motivated a line of work in graph optimization and the study of accelerated width reduced methods which we describe in detail in Section \ref{sec:RelatedWorks}.

Our next contribution is towards the IRLS approach.
 For the $\ell_p$-regression problem i.e. $\dd=\MM=0$ in
 \eqref{eq:Main}, we give an IRLS algorithm that globally converges to
 the optimum in at most
 $O\left(p^3m^{\frac{p-2}{2(p-1)}}\log\frac{m}{\epsilon}\right)$
 linear system solves for all $p\geq 2$ (Section \ref{sec:IRLS}). This
 is the first IRLS algorithm that is guaranteed to converge to the
 optimum for all values of $p\geq 2$, with a quantitative bound on the
 runtime. Our IRLS algorithm has proven to be very fast
   and robust in practice and is faster than existing implementations
   in MATLAB/CVX by 10-50x. These speed-ups are demonstrated in
   experiments performed in \cite{AdilPS19} and we present these
   results along with our algorithm in Section~\ref{sec:IRLS}.

\begin{restatable}{theorem}{IRLSMain}\label{thm:IRLS-Main}
Let $p \geq 2$. Algorithm \ref{alg:IRLS} returns $\xx$ such that
$\AA\xx = \bb$ and $\|\NN\xx\|_p^p  \leq (1 + \epsilon)
\|\NN\xx^{\star}\|_p^p$, in at most $O\left(p^3
  m^{\frac{(p-2)}{2(p-1)}}\log \left(\frac{m}{\epsilon
    }\right)\right)$ calls to a linear system solver. 
\end{restatable}

  The analysis of our IRLS algorithm fits into the overall framework of this paper. Such an algorithm first appeared in the conference paper by \textcite{AdilPS19}, where they also ran some experiments to demonstrate the performance of their IRLS algorithm in practice. We include some of their experimental results to show that the rate of convergence in practice is even better than the theoretical bounds.

\subsection{Technical Overview}

\paragraph{Overall $\log \frac{1}{\epsilon}$ Convergence} Our algorithm follows an overall {\it iterative refinement} approach
for $p\geq 2$, which implies $\ff(\xx+\delta)-\ff(\xx)$ can be upper
bounded by the function
$\res_p = \gg^{\top}\delta + \|\RR\delta\|_2^2 + \|\NN\delta\|_p^p$,
and lower bounded by a similar function. Here, the vector $\gg$ and
matrix $\RR$ depend on $\xx,$ and the matrix $\NN$ is as defined in
Problem \eqref{eq:Main}. We prove that if we can solve
$\min_{\AA\delta=0}\res_p(\delta)$ to a $\kappa$-approximation,
$O(p \kappa \log
  \nfrac{(\ff(\xx^{(0)})-\ff(\xx^{\star}))}{\eps})$
such solves (iterations) suffice to obtain an $\eps$-approximate solution to Problem \eqref{eq:Main} (Theorem
\ref{thm:IterativeRefinement}). We call this problem the {\it Residual
  Problem} and this process {\it Iterative Refinement for
  $\ell_p$-norms}.

\paragraph{Solving the Residual Problem} We next perform a binary search on the linear term of the residual problem and reduce it to solving $O(\log p)$ problems of the form, $\min_{\AA\delta = \cc}  \|\RR\delta\|_2^2 + \|\NN\delta\|_p^p$ (Lemma \ref{lem:binary}). In order to solve these new problems, we use a multiplicative weight update routine that returns a constant approximate solution in $O(p m^{\nfrac{(p-2)}{(3p-2)}})$ calls to a linear system solver (Theorem \ref{cor:ResidualDecision}). We can thus find a constant approximate solution to the residual problem in $O(p m^{\nfrac{(p-2)}{(3p-2)}}\log p)$ calls to a linear system solver (Corollary \ref{cor:ResApprox}). Combined with iterative refinement, we obtain an algorithm that converges in $O\left(p^2 m^{\frac{p-2}{3p-2}}\log p\log \frac{\ff(\xx^{(0)})-\ff(\xx^{\star})}{\epsilon}\right) \leq \Otil\left(p^2 m^{1/3}\log \frac{1}{\epsilon}\right)$ linear system solves. 

\paragraph{Improving $p$ Dependence}Furthermore, we prove that for any $q\neq p$, given a $p$-norm
residual problem, we can construct a corresponding $q$-norm residual
problem such that $\beta$-approximate solution to the $q$-norm
residual problem roughly gives a
$O(\beta^2) m^{\abs{\frac{1}{p}-\frac{1}{q}}}$
 approximate solution to the $p$-norm residual problem (Theorem
\ref{thm:p2q}). As a consequence, if $p$ is large, i.e.
$p \geq \log m$, a constant approximate solution to the corresponding
$\log m$-norm residual problem will give an
$O(m^{\frac{1}{\log m}}) \leq O(1)$-approximate solution to the
$p$-norm residual problem in at most
$O(\log m \cdot m^{\frac{\log m-2}{3\log m-2}})\leq
\Otil(m^{\frac{p-2}{3p-2}})$ calls to a linear system
solver. Combining this with the algorithm described in the previous
paragraph, we obtain our final guarantees as described in Theorem
\ref{thm:MainThm}.

\paragraph{$\ell_p$-Regression in Matrix Multiplication Time}
We next
describe how to obtain the guarantees of Theorem
\ref{thm:InverseIntro}. While solving the residual problem, the
algorithm solves a system of linear equations at every iteration. The
key observation for obtaining improved running times is that the
weights determining these linear systems change slowly. Thus, we can
maintain a spectral approximation to the linear system via a sequence
of \emph{lazy} low-rank updates. The Sherman-Morrison-Woodbury formula
then allows us to update the inverse quickly. We can use the spectral
approximation as a preconditioner for solving the linear system
quickly at each iteration. Thus, we obtain a speed-up since the linear
systems do not need to be solved from scratch at each iteration,
giving Theorem~\ref{thm:InverseIntro}.

\paragraph{Good Starting Solution}For $\ell_p$-norm objectives, i.e., $\min_{\AA\xx=\bb} \|\NN\xx\|_p^p$, we further show how to find a starting solution $\xx^{(0)}$ such that $\|\NN\xx^{(0)}\|_p^p \leq O(m)\|\NN\xx^{\star}\|_p^p$ . The key idea is that for any $k$, a constant approximate solution to the $k$-norm problem is an $O(m)$-approximate solution to the $2k$-norm problem (Lemma \ref{lem:homotopy}). This inspires a homotopy approach, where we first solve an $\ell_2$ norm problem followed by $\ell_{2^2},\ell_{2^3},\cdots, \ell_{2^{\ceil{\log p}}}$-norm problems to constant approximations. We can thus obtain the required starting solution in at most $O\left(p m^{\frac{p-2}{3p-2}} \log m \log^2 p\right) $ calls to a linear system solver.

\paragraph{IRLS Algorithm} For the IRLS algorithm, given the residual problem at an iteration, we show how to construct a weighted least squares problem, the solution of which is an $O\left(p^2 m^{\frac{p-2}{2(p-1)}}\right)$-approximate solution to the residual problem (Lemma \ref{lem:residualIRLS}). This result along with the overall iterative refinement culminates in our IRLS algorithm where we directly solve these weighted least squares problems in every iteration.

\subsection{Related Works}
\label{sec:RelatedWorks}

 \paragraph{$\ell_p$-Regression}  Until 2018, the fastest
 high-accuracy algorithms for $\ell_p$-regression, including the \textcite{nesterov1994interior} Interior Point Method framework and \textcite{bubeck2018homotopy} homotopy method, asymptotically required $\approx
 O(\sqrt{n})$ linear system solves. The first algorithm for $\ell_p$-regression to beat the $\sqrt{n}$ iteration bound was the algorithm by \textcite{AdilKPS19}, which was faster than all known algorithms and asymptotically required at most $\approx O(p^{O(p)}n^{1/3})$ iterations , for all $p >1$. Concurrently \textcite{bullins2018fast} used tools from convex optimization to give an algorithm for $p=4$ which matches the rates of \textcite{AdilKPS19} up to logarithmic factors. Subsequent works have improved the $p$ dependence \cite{AdilS20,AdilBKS21} and proposed alternate methods for obtaining matching rates (upto logarithmic and $p$ factors) \cite{CJJJLST}. A recent work by \textcite{jambulapati2021improved} shows how to solve $\ell_p$-regression in $\approx n + poly(p) \cdot d^{\frac{p-2}{3p-2}}$ iterations where $d$ is the smaller dimension of the constraint matrix $\AA$.

 \paragraph{Width Reduced MWU Algorithms} Width reduction is a technique that has been used repeatedly in multiplicative weight update algorithms to speed up rates of convergence from $m^{1/2}$ to $m^{1/3}$, where $m$ is the size of the input. This technique was first seen in the work of \textcite{CKMST}, in the context of the maximum flow problem where for a graph with $n$ vertices and $m$ edges to improve the iteration complexity from $\Otil(m^{1/2})$ to $\Otil(m^{1/3})$. A similar improvement was further seen in algorithms for $\ell_1, \ell_{\infty}$-regression by \textcite{chin2013runtime,ene2019improved}, $\ell_p$-regression ($p \geq 2$) \textcite{AdilKPS19} and, algorithms for matrix scaling \cite{allen2017much}. In a recent work \textcite{AdilBS21} extend this technique to improve iteration complexities for all {\it quasi-self-concordant} objectives which includes soft-max and logistic regression among others.

 \paragraph{Inverse Maintenance}
 Inverse Maintenance is a technique used to speed up algorithms and
 was first introduced by \textcite{vaidya1990solving} in the context of
 minimum cost and multicommodity flows and has further been used for
 interior point methods \textcite{LeeS14}, \textcite{lee2015faster}. In
 2019, \textcite{AdilKPS19} developed a method for $\ell_p$-regression
 that utilized the idea of reusing inverses due to controllable rates
 of change of underlying variables.

 \paragraph{IRLS Algorithms} Iteratively Reweighted Least Squares Algorithms are simple to implement and have thus been used in a wide range of applications including sparse signal reconstruction \cite{gorodnitsky1997}, compressive sensing \cite{chartrand2008iteratively} and Chebyshev approximation in FIR filter design \cite{barreto1994sub}. Refer to \textcite{burrus2012} for a full survey. The works by \textcite{osborne1985} and \textcite{karlovitz1970} show convergence in the limit and with certain assumptions on the starting solution. For $\ell_1$-regression, \textcite{Straszak2016OnAN,straszak2016natural,straszak2016irls} show quantitative convergence bounds. In 2019, \textcite{AdilPS19} give the first IRLS algorithm with quantitative bounds that is guaranteed to converge with no conditions on the starting point. Their algorithm also works well in practice as suggested by their experiments.

\paragraph{Follow-up Work in Graph Optimization}

The $\ell_p$-norm flow problem, which asks to minimize the $\ell_p$-norm of a flow vector while satisfying certain demand constraints, is modeled via the $\ell_p$-regression problem. The maximum flow problem is the special case of $p=\infty$. For graphs with $n$ vertices and $m$ edges, the $\ell_p$-norm regression algorithm of \textcite{AdilKPS19} when combined with fast laplacian solvers, directly gives an $\approx \Otil(p^{O(p)}m^{4/3})$ time algorithm for the $\ell_p$-norm flow problem. Building on their work, specifically the iterative refinement framework, which allows to solve these problems to a high-accuracy while only requiring an $m^{o(1)}$-aproximate solution to an $\ell_p$-norm subproblem, \textcite{kyngPSW19} give an algorithm for unweighted graphs that runs in time $\exp(p^{3/2})m^{1 + \frac{7}{\sqrt{p-1}}+o(1)}$. We note that their algorithm runs in time $m^{1+o(1)}$ for $p = \sqrt{\log m}$. Further works including \textcite{AdilBKS21} also utilize the iterative refinement guarantees to give an algorithm with runtime $p(m^{1+o(1)} + n^{4/3+o(1)})$ for weighted $\ell_p$-norm flow problems by designing new sparsification algorithms that preserve $\ell_p$-norm objectives of the subproblem to an $m^{o(1)}$-approximation. For the maximum flow problem, \textcite{AdilS20} give an $m^{1+o(1)}\epsilon^{-1}$ time algorithm for the approximate maximum flow problem on unweighted graphs. \textcite{kathuria2020unit} build on these works further and give an algorithm that computes maximum $s$-$t$ flow problem where each edge has integer capacities at most $U$, in time $m^{4/3+o(1)}U^{1/3}$. In a recent breakthrough result by \textcite{chen2022maximum}, the authors give an algorithm for the maximum flow problem and the $\ell_p$-norm flow problem that runs in almost linear time, $m^{1+o(1)}$.

\subsection{Organization of Paper}
Section \ref{sec:IR} describes the overall iterative refinement framework, first for $p \geq 2$, and then for $p \in (1,2)$. In the end, we show how to find good starting solutions for pure $\ell_p$-norm objectives for ${}p \geq 2$. Section \ref{chap:MWU} describes the width reduced multiplicative weight update routine used to solve the residual problem. In Section \ref{sec:p2q} we show how to solve $p$-norm residual problems using $q$-norm residual problems and give our overall algorithm (Algorithm \ref{alg:CompleteLp}). Section \ref{sec:Inv} contains our new inverse maintenance algorithm that allows us to solve $\ell_p$-regression almost as fast as linear regression. Finally in Section \ref{sec:IRLS} we give an IRLS algorithm and present some experimental results from \textcite{AdilPS19}.


%% file: Chapters/IterRefinement.tex

\section{Iterative Refinement for \texorpdfstring{$\ell_p$}{TEXT}-norms}
\label{sec:IR}

Recall that we would like to find a high-accuracy solution for the problem, 
\[
\min_{\AA\xx = \bb}\quad \ff(\xx) = \dd^{\top}\xx + \|\MM\xx\|_2^2 + \|\NN\xx\|_p^p
\]
for matrices $\AA\in \mathbb{R}^{d\times n}, \MM \in \mathbb{R}^{m_1\times n}, \NN \in \mathbb{R}^{m_2\times n},$ $m_1,m_2 \geq n, d\leq n$.

A common approach in smooth, convex optimization is upper bounding the function using a first order Taylor expansion plus a quadratic function (smoothness), and minimizing this bound repeatedly to converge to the optimum. Additionally, when the function has a similar quadratic lower bound (strong convexity) it can be shown that minimizing this upper bound $O\left(\log \frac{1}{\epsilon}\right)$\footnote{hiding problem dependent parameters} times is sufficient to converge to an $\epsilon$-approximate solution. The $\ell_p$-norm function satisfies no such quadratic upper bound since it has a very steep growth, or lower bound since it is too flat around $0$. In this section we show that we can instead upper and lower bound the $\ell_p$ function for $p \geq 2$ by a second order Taylor expansion plus an $\ell_p^p$ term. We show that it is sufficient to minimize such a bound to a $\kappa$-approximation $O\left(p \kappa \log \frac{1}{\epsilon}\right)$ times. Such an iterative refinement method was previously only known for $p = 2$, and we thus call this algorithm {\it Iterative Refinement for $\ell_p$-norms}. In further sections, we show different ways to minimize this upper bound approximately to obtain fast algorithms.

For $p\in (1,2)$, we use a smoothed function which is quadratic in a small range around $0$ and grows as $\ell_p^p$ otherwise. We use this function to give upper and lower bounds and a similar iterative refinement scheme.

We further show how to obtain a good starting solution for Problem \eqref{eq:Main} in the special case when the vector $\dd$ and matrix $\MM$ are zero, i.e., the objective function is only the $\ell_p$-norm function.

These sections are based on the results and proofs from \textcite{AdilKPS19,AdilPS19,AdilS20,AdilBKS21}.

\subsection{Iterative Refinement}\label{sec:IterRef}

We will prove that the following algorithm can be used to obtain a high-accuracy solution, i.e., $\log \frac{1}{\epsilon}$ rate of convergence for $\ell_p$-regression.
\begin{algorithm}[H]
\caption{Iterative Refinement}
\label{alg:IterRef}
 \begin{algorithmic}[1]
 \Procedure{\textsc{Main-Solver}}{$\AA, \MM,\NN, \dd,\bb,p,\epsilon$}
\State $\xx \leftarrow \xx^{(0)}$
\State $\nu\leftarrow$ Bound on $ \ff(\xx^{(0)})-\ff(\xx^{\star})$ \Comment{If $\ff(\xx^{\star})\geq 0$, then $\nu \leftarrow \ff(\xx^{(0)})$}
\While{$\nu >\epsilon$}
\State $\Dtil\leftarrow$ {\sc ResidualSolver($\xx,\MM,\NN,\AA,\dd,\bb,\nu,p$)}
\If{$\res_p(\Dtil) \geq \frac{\nu}{32p\kappa}$}
\State $\xx \leftarrow \xx - \frac{\Dtil}{p}$ \label{alg:line:update-x}
\Else
\State $\nu \leftarrow \frac{\nu}{2}$\label{alg:line:update-nu}
\EndIf
\EndWhile
\State \Return $\xx$
 \EndProcedure 
 \end{algorithmic}
\end{algorithm}
Specifically, we will prove,
\begin{restatable}{theorem}{Iterative}\label{thm:IterativeRefinement}
Let $p \geq 2$, and $\kappa \geq 1$. Let the initial solution $\xx^{(0)}$ satisfy $\AA\xx^{(0)} = \bb$. Algorithm \ref{alg:IterRef} returns an $\epsilon$-approximate solution $\xx$ of Problem \eqref{eq:Main} in at most $O\left(p \kappa \log \left(\frac{\ff(\xx^{(0)})-\ff(\xx^{\star}) }{\epsilon }\right)\right)$ calls to a $\kappa$-approximate solver for the residual problem (Definition \ref{def:residual}).
\end{restatable}

Before we prove the above result, we will define some of the terms used in the above statement.

\subsubsection{Preliminaries}

\begin{definition}[$\epsilon$-Approximate Solution]
Let $\xx^{\star}$ denote the optimizer of Problem \eqref{eq:Main}. We say $\xxtil$ is an $\epsilon$-approximate solution to \eqref{eq:Main} if $\AA\xxtil = \bb$ and 
\[
\ff(\xxtil)\leq \ff(\xx^{\star}) + \epsilon.
\]
\end{definition}
\begin{definition}[Residual Problem]\label{def:residual}
  For any $p \geq 2$, we define the residual problem $res_p(\Delta)$,
   for \eqref{eq:Main} at a feasible $\xx$ as,
   \[
     \max_{\AA\Delta = 0} \quad \res_p(\Delta) \defeq \gg^{\top}\Delta -
     \Delta^{\top}\RR\Delta - \|\NN\Delta\|_p^p, \text{    where,}
   \] 
   \[\gg = \frac{1}{p}\dd + \frac{2}{p}\MM^{\top} \MM\xx +
     \NN^{\top}Diag(|\NN\xx|^{p-2}) \NN\xx \quad \text{and} \quad\RR =\frac{2}{p^2}\MM^{\top}\MM+ 2\NN^{\top} Diag(|\NN\xx|^{p-2})\NN.\]

\end{definition}

\begin{definition}[Approximation to Residual Problem]\label{def:res-approx}
Let $p\geq 2$ and $\Dopt$ be the optimum of the residual problem. $\Dtil$ is a $\kappa$-approximation to the residual problem if $\AA\Dtil = 0$ and,
\[
res_p(\Dtil) \geq \frac{1}{\kappa} res_p(\Dopt).
\]
\end{definition}

\subsubsection{Bounding Change in Objective}
In order to prove our result, we first show that we can upper and lower bound the change in our $\ell_p$-objective by a linear term plus a quadratic term plus an $\ell_p$-norm term.
\begin{lemma}\label{lem:precondition}
For any $\xx,\Delta$ and $p \geq 2$, we have for vectors $\rr,\gg$ defined coordinate wise as $\rr =|\xx|^{p-2}$ and $\gg = p |\xx|^{p-2}\xx$,
\[
\frac{p}{8} \sum_i \rr_i \Delta_i^2 + \frac{1}{2^{p+1}}\norm{\Delta}_p^p \leq \norm{\xx+\Delta}^p_p - \norm{\xx}_p^p - \gg^{\top}\Delta \leq 2 p^2 \sum_i \rr_i \Delta_i^2 + p^p \norm{\Delta}_p^p.
\]
\end{lemma}

\begin{proof}
To show this, we show that the above holds for all coordinates. For a single coordinate, the above expression is equivalent to proving,
\[
\frac{p}{8} |x|^{p-2} \Delta^2 + \frac{1}{2^{p+1}}\abs{\Delta}^p \leq \abs{\xx+\Delta}^p - \abs{\xx}^p - p\abs{x}^{p-1}sgn(x)\Delta \leq 2 p^2 |x|^{p-2} \Delta^2  + p^p \abs{\Delta}^p.
\]
Let $\Delta = \alpha x$. Since the above clearly holds for $x=0$, it remains to show for all $\alpha$,
\[
\frac{p}{8}\alpha^2 + \frac{1}{2^{p+1}}\abs{\alpha}^p \leq \abs{1+\alpha}^p - 1 - p\alpha \leq 2 p^2 \alpha^2  + p^p \abs{\alpha}^p.
\]
\begin{enumerate}
\item $\alpha \geq 1$:\\
\noindent In this case, $1+\alpha \leq 2 \alpha \leq p\cdot \alpha$. So, $ \abs{1+\alpha}^p \leq p^p \abs{\alpha}^p$ and the right inequality directly holds. To show the other side, let 
\[
h(\alpha) = (1+\alpha)^p - 1 - p\alpha - \frac{p}{8} \alpha^2 - \frac{1}{2^{p+1}}{\alpha}^p.
\]
We have,
\[
h'(\alpha) = p(1+\alpha)^{p-1}  - p- \frac{p}{4} \alpha - \frac{p}{2^{p+1}}{\alpha}^{p-1}
\]
and 
\[
h''(\alpha) = p(p-1)(1+\alpha)^{p-2}  - \frac{p}{4}  - \frac{p(p-1)}{2^{p+1}}{\alpha}^{p-2} \geq 0.
\]
Since $h''(\alpha) \geq 0$, $h'(\alpha) \geq h'(1) \geq 0$. So $h$ is an increasing function in $\alpha$ and $h(\alpha) \geq h(1) \geq 0$.

\item $\alpha \leq -1$:\\
Now, $\abs{1+\alpha} \leq 1+\abs{\alpha} \leq p\cdot \abs{\alpha}$, and $2\alpha^2 p^2 - \abs{\alpha} p \geq 0$. As a result,
\[
\abs{1+\alpha}^p \leq - \abs{\alpha} p +2\alpha^2 p^2  + p^p\cdot \abs{\alpha}^p
\]
which gives the right inequality. Consider,
\[
h(\alpha) = |1+\alpha|^p - 1 - p\alpha - \frac{p}{8} \alpha^2 - \frac{1}{2^{p+1}}|\alpha|^p.
\]
\[
h'(\alpha) = - p|1+\alpha|^{p-1}  - p - \frac{p}{4} \alpha + p\frac{1}{2^{p+1}}|\alpha|^{p-1}.
\]
Let $\beta = -\alpha$. The above expression now becomes,
\[
- p (\beta - 1)^{p-1} - p + \frac{p}{4} \beta + p\frac{1}{2^{p+1}}\beta^{p-1}.
\]
We know that $\beta \geq 1$. When $\beta \geq 2$, $\frac{\beta}{2} \leq \beta - 1$ and $\frac{\beta}{2} \leq \left(\frac{\beta}{2}\right)^{p-1}$. This gives us,
\[
\frac{p}{4} \beta + p\frac{1}{2^{p+1}}\beta^{p-1} \leq \frac{p}{2} \left(\frac{\beta}{2}\right)^{p-1} +  \frac{p}{2} \left(\frac{\beta}{2}\right)^{p-1} \leq p (\beta - 1)^{p-1}
\]
giving us $h'(\alpha) \leq 0$ for $\alpha \leq -2$. When $\beta \leq 2$, $\frac{\beta}{2} \geq \left(\frac{\beta}{2}\right)^{p-1}$ and $\frac{\beta}{2}  \leq 1$.
\[
\frac{p}{4} \beta + p\frac{1}{2^{p+1}}\beta^{p-1} \leq \frac{p}{2}\cdot  \frac{\beta}{2} + \frac{p}{2}\cdot  \frac{\beta}{2} \leq p
\]
giving us $h'(\alpha) \leq 0$ for $ -2 \leq \alpha \leq -1$. Therefore, $h'(\alpha) \leq 0$ giving us, $h(\alpha) \geq h(-1) \geq 0$, thus giving the left inequality.

\item $\abs{\alpha} \leq 1$:\\
Let $s(\alpha) =   1 + p\alpha + 2p^2 \alpha^2 +  p^p \abs{\alpha}^p - (1+\alpha)^p.$ Now,
\[
s'(\alpha)  = p + 4 p^2 \alpha +  p^{p+1} \abs{\alpha}^{p-1} sgn(\alpha) - p (1+\alpha)^{p-1}.
\]
When $\alpha \leq 0$, we have,
\[
s'(\alpha)  = p + 4 p^2 \alpha - p^{p+1} \abs{\alpha}^{p-1} - p (1+\alpha)^{p-1}.
\]
and 
\[
s''(\alpha)  = 4 p^2 +  p^{p+1}(p-1) \abs{\alpha}^{p-2} - p (p-1)(1+\alpha)^{p-1} \geq 2 p^2 +  p^{p+1}(p-1) \abs{\alpha}^{p-2} - p (p-1) \geq 0 .
\]
So $s'$ is an increasing function of $\alpha$ which gives us, $s'(\alpha) \leq s'(0) = 0$. Therefore $s$ is a decreasing function, and the minimum is at $0$ which is $0$. This gives us our required inequality for $\alpha \leq 0$.
When $\alpha \geq \frac{1}{p-1}$, $1 + \alpha \leq p \cdot \alpha$ and $s'(\alpha) \geq 0$. We are left with the range $0 \leq \alpha \leq \frac{1}{p-1}$. Again, we have,
\begin{align*}
s''(\alpha)  & = 4 p^2 +  p^{p+1}(p-1) \abs{\alpha}^{p-2} - p (p-1)(1+\alpha)^{p-1} \\
& \geq 4 p^2 +  p^{p+1}(p-1) \abs{\alpha}^{p-2} - p (p-1) (1+\frac{1}{p-1})^{p-1}\\
&\geq 4 p^2 +  p^{p+1}(p-1) \abs{\alpha}^{p-2} - p (p-1) e, \text{When $p$ gets large the last term approaches $e$}\\
& \geq 0.
\end{align*}
Therefore, $s'$ is an increasing function, $s'(\alpha) \geq s'(0) = 0$. This implies $s$ is an increasing function, giving, $s(\alpha) \geq s(0)=0$ as required. 

To show the other direction,
\[
h(\alpha) = (1+\alpha)^p - 1 - p\alpha - \frac{p}{8} \alpha^2 - \frac{1}{2^{p+1}}\abs{\alpha}^p \geq (1+\alpha)^p - 1 - p\alpha - \frac{p}{8} \alpha^2 - \frac{p}{8}{\alpha}^2 = (1+\alpha)^p - 1 - p\alpha - \frac{p}{4} \alpha^2.
\]
Now, since $p \geq 2$,
\begin{align*}
&\left((1+\alpha)^{p-2} - 1 \right)sgn(\alpha) \geq 0\\
\Rightarrow &\left((1+\alpha)^{p-1} - 1 - \alpha \right)sgn(\alpha) \geq 0\\
\Rightarrow &\left(p(1+\alpha)^{p-1} - p - \frac{p}{2}\alpha \right)sgn(\alpha) \geq 0
\end{align*}
We thus have, $h'(\alpha) \geq 0$ when $\alpha$ is positive and $h'(\alpha) \leq 0$ when $\alpha$ is negative. The minimum of $h$ is at $0$ which is $0$. This concludes the proof of this case.
\end{enumerate}
\end{proof}

\subsubsection{Proof of Iterative Refinement}

In this section we will prove our main result. We start by proving the following lemma which relates the objective of the residual problem defined in the preliminaries to the change in objective value when $\xx$ is updated by $\Delta/p$.

\begin{lemma}
\label{lem:RelateResidualOpt}
For any $\xx,\Delta$ and $p \geq 2$ and $\lambda = 16p$,
\[
\res_p(\Delta) \leq  \ff(\xx)  -\ff\left(\xx-\frac{\Delta}{p}\right),
\]
and 
\[
\ff(\xx)  -\ff\left(\xx-\lambda \frac{\Delta}{p}\right) \leq \lambda \cdot \res_p(\Delta).
\]
\end{lemma}
\begin{proof}
We note,
\begin{align*}
\ff\left(\xx-\frac{\Delta}{p}\right)  = & \dd^{\top}\left(\xx-\frac{\Delta}{p}\right) + \norm{\MM\left(\xx-\frac{\Delta}{p}\right)}_2^2 + \norm{\NN\left(\xx-\frac{\Delta}{p}\right)}_p^p\\
 = & \dd^{\top}\xx + \|\MM\xx\|_2^2 + \norm{\NN\left(\xx-\frac{\Delta}{p}\right)}_p^p - \frac{1}{p}\dd^{\top}\Delta - \frac{2}{p} \xx^{\top}\MM^{\top}\MM\Delta + \frac{1}{p^2}\|\MM\Delta\|_2^2\\
 \leq & \dd^{\top}\xx + \|\MM\xx\|_2^2 + \|\NN\xx\|_p^p - p|\NN\xx|^{p-2}(\NN\xx)^{\top}\frac{\NN\Delta}{p} + 2p^2 \frac{(\NN\Delta)^{\top}}{p}(\NN\xx)^{p-2}\frac{\NN\Delta}{p}  \\
& + p^p \norm{\frac{\NN\Delta}{p}}_p^p - \frac{1}{p}\dd^{\top}\Delta - \frac{2}{p} \xx^{\top}\MM^{\top}\MM\Delta + \frac{1}{p^2}\|\MM\Delta\|_2^2\\
& \text{(From right inequality of Lemma \ref{lem:precondition})}\\
& = \ff(\xx) -  \left(\frac{1}{p}\dd + \frac{2}{p}\MM^{\top} \MM\xx +
     \NN^{\top}|\NN\xx|^{p-2} \NN\xx\right)^{\top} \Delta \\
& - \Delta^{\top}\left(\frac{2}{p^2}\MM^{\top}\MM+ 2\NN^{\top} Diag(|\NN\xx|^{p-2})\NN\right)\Delta + \|\NN\Delta\|_p^p\\
& = \ff(\xx) - \res_p(\Delta), \text{ From Definition \ref{def:residual}.}
\end{align*}
 Let $\gg$ and $\RR$ be as defined in Definition \ref{def:residual}. We now use a similar calculation and the left inequality of Lemma \ref{lem:precondition} to get,
 \[
 \ff\left(\xx - \lambda \frac{\Delta}{p}\right) \geq \ff(\xx) - \lambda \gg^{\top}\Delta - \frac{\lambda^2}{16p} \Delta^{\top}\RR\Delta - \frac{\lambda^p}{p^p 2^{p+1}}.
 \]
 For $\lambda = 16p$,
 \begin{align*}
 \ff(\xx) - \lambda \gg^{\top}\Delta - \frac{\lambda^2}{16p} \Delta^{\top}\RR\Delta - \frac{\lambda^p}{p^p 2^{p+1}} & \geq \ff(\xx) - \lambda \left(\gg^{\top}\Delta - \frac{\lambda}{16p} \Delta^{\top}\RR\Delta - \frac{\lambda^{p-1}}{p^p 2^{p+1}}\right)\\
 & \geq \ff(\xx) - \lambda \res_p(\Delta),
 \end{align*}
 thus concluding the proof of the lemma.
\end{proof}
We now track the value of $\ff(\xx^{(t)}) - \ff(\xx^{\star})$ with a parameter $\nu$. We will first show that, if we have a $\kappa$ approximate solver for the residual problem, we can either take a step to obtain $\xx^{(t+1)}$ such that 
\begin{equation}\label{eq:progressStep}
\ff(\xx^{(t+1)}) - \ff(\xx^{\star}) \leq \left(1-\frac{1}{32 p\kappa}\right)\left(\ff(\xx^{(t)}) - \ff(\xx^{\star})\right),
\end{equation}
or we need to reduce the value of $\nu$ by a factor of $2$ since $\ff(\xx^{(t)}) - \ff(\xx^{\star})$ is less than $\nu/2$.

\begin{restatable}{lemma}{BinarySearch}\label{lem:invariant}
Consider an iterate $t$. Let $\res_p$ denote the residual problem at $\xx^{(t)}$ and $\nu$ be as defined in Algorithm \ref{alg:IterRef}. Let $\Dtil$ denote the solution returned by a $\kappa$-approximate solver to the residual problem. Then, 
\begin{enumerate}
  \item either $\ff(\xx^{(t)}) - \ff(\xx^{\star}) \leq \nu $ and, $\xx^{(t+1)} = \xx^{(t)} - \frac{\Dtil}{p}$ satisfies \eqref{eq:progressStep},
  \item or, $\ff(\xx^{(t)}) - \ff(\xx^{\star}) \leq \frac{\nu}{2}$ and Line \ref{alg:line:update-nu} in the algorithm is executed.\label{option2}
\end{enumerate}
\end{restatable}
\begin{proof}
We will first prove that $\ff(\xx^{(t)}) - \ff(\xx^{\star}) \leq \nu$ by induction. For $t = 0$, $\ff(\xx^{(0)}) - \ff(\xx^{\star}) \leq \nu $ by definition. Now, let us assume this is true for iteration $t$. Note that, if the algorithm updates $\xx$ in line \ref{alg:line:update-x}, since $\ff(\xx^{(t+1)}) \leq \ff(\xx^{(t)})$ (solution of the residual problem is always non-negative), the relation holds for $t+1$. Otherwise, the algorithm reduces $\nu$ to $\nu/2$ and $\res_p(\Dtil) < \frac{\nu}{32p\kappa}$. For $\Dbar$ such that $\xx^{\star}= \xx^{(t)} - \lambda\frac{\Dbar}{p}$, and from Lemma \ref{lem:RelateResidualOpt},
\[
\ff(\xx^{(t)}) -\ff(\xx^{\star}) = \ff(\xx^{(t)}) - \ff\left(\xx^{(t)} - \lambda\frac{\Dbar}{p}\right) \leq \lambda \res_p(\Dbar) \leq \lambda \res_p(\Dopt).
\]
Since $\Dtil$ is a $\kappa$-approximate solution to the residual problem,
\[
\lambda \res_p(\Dopt) \leq \lambda \kappa \res_p(\Dtil) < 16p \kappa\frac{\nu}{32p\kappa} \leq \frac{\nu}{2}.
\]
We have thus shown that $\ff(\xx^{(t)}) - \ff(\xx^{\star}) \leq \nu$ for all iterates $t$ and whenever Line \ref{alg:line:update-nu} of the algorithm is executed, \ref{option2} from the lemma statement holds. It remains to prove that if $\res_p(\Dtil) \geq \frac{\nu}{32p\kappa}$, then $\xx^{(t+1)} = \xx^{(t)} - \frac{\Dtil}{p}$ satisfies \eqref{eq:progressStep}. Since, $\ff(\xx^{(t)}) - \ff(\xx^{\star}) \leq \nu$,
\[
\res_p(\Dtil) \geq \frac{\nu}{32 p\kappa} \geq  \frac{1}{32p\kappa}\left(\ff(\xx^{(t)}) - \ff(\xx^{\star})\right).
\]
Now, from Lemma \ref{lem:RelateResidualOpt},
\begin{align*}
\ff\left(\xx^{(t+1)} \right) - \ff(\xx^{\star}) &\leq \ff(\xx^{(t)}) - \res_p(\Dtil) - \ff(\xx^{\star})\\
& \leq \left(\ff(\xx^{(t)}) - \ff(\xx^{\star})\right) -  \frac{1}{32p \kappa}\left(\ff(\xx^{(t)}) - \ff(\xx^{\star})\right)\\
& = \left(1 -  \frac{1}{32p \kappa}\right) \left(\ff(\xx^{(t)}) - \ff(\xx^{\star})\right).
\end{align*}
\end{proof}

\begin{restatable}{corollary}{End}\label{lem:end}
The vector $\xx$ returned by Algorithm \ref{alg:IterRef} is an $\epsilon$-approximate solution to Problem \eqref{eq:Main}.
\end{restatable}
\begin{proof}
Our starting solution $\xx^{(0)}$ satisfies $\AA\xx^{(0)} = \bb$ and the solutions $\Dtil$ of the residual problem added in each iteration satisfy $\AA\Dtil = 0$. Therefore, $\AA\xx = \bb$. For the second part, note that we always have $\ff(\xx^{(t)}) - \ff(\xx^{\star}) \leq \nu$. When we stop, $\nu \leq \epsilon$. Thus,
\[
\ff(\xx^{(t)})- \ff(\xx^{\star}) \leq \epsilon.
\]

\end{proof}
We are now ready to prove our main result.
\Iterative*
\begin{proof}
From Corollary \ref{lem:end}, the solution returned by the algorithm is as required. We next need to bound the runtime. From Lemma \ref{lem:invariant}, the algorithm, either reduces $\nu$ or Equation \eqref{eq:progressStep} holds. The number of times we can reduce $\nu$ is bounded by $\log \frac{\ff(\xx^{(0)})-\ff(\xx^{\star})}{\epsilon }$. The number of times Equation \eqref{eq:progressStep} holds can be bounded as follows,
\[
\frac{\epsilon }{2} \leq \ff\left(\xx^{(t+1)} \right) - \ff(\xx^{\star}) \leq \left(1 -  \frac{1}{32p \kappa}\right)^t \left(\ff(\xx^{(0)}) - \ff(\xx^{\star})\right).
\]
Therefore, the total number of iterations $T$ is bounded as $T \leq 32p \kappa \log \left(\frac{\ff(\xx^{(0)})-\ff(\xx^{\star})}{\epsilon }\right)$.
\end{proof}

\subsection{Starting Solution and Homotopy for pure \texorpdfstring{$\ell_p$}{TEXT} Objectives}
\label{sec:homotopy}

In this section, we consider the case where $\ff(\xx) = \|\NN\xx\|_p^p$, i.e., $\dd = 0$ and $\MM = 0$. 
\begin{equation}\label{eq:lpObj}
\min_{\AA\xx = \bb} \|\NN\xx\|_p^p
\end{equation}
For such cases, we show how to find a good starting solution. We note that we can solve the following problem since it is equivalent to solving a system of linear equations,
\[
\min_{\AA\xx = \bb} \|\NN\xx\|_2^2.
\]
Refer to Appendix \ref{chap:l2solve} for details on how the above is equivalent to solving a system of linear equations.

We next consider a homotopy on $p$. Specifically, we want to find a starting solution for the $\ell_p$-norm problem by first solving an $\ell_2$-norm problem, followed by $\ell_{2^2}, \ell_{2^3}, ..., \ell_{2^{\lfloor \log p - 1\rfloor}}$-norm problems to a constant approximation. The following lemma relates these solutions.

\begin{lemma}\label{lem:homotopy}
Let $\xx_k^{\star}$ denote the optimum of the $k$-norm and $\xx_{2k}^{\star}$ the optimum of the $2k$-norm problem \eqref{eq:lpObj}. Let $\xxtil$ be an $O(1)$-approximate solution to the $k$-norm problem. The following relation holds,
\[
\norm{\xx^{\star}_{2k}}_{2k}^{2k} \leq \norm{\xxtil}_{2k}^{2k} \leq O(m) \norm{\xx^{\star}_{2k}}_{2k}^{2k}.
\]
In other words, $\xxtil$ is a $O(m)$-approximate solution to the $2k$-norm problem.
\end{lemma}
\begin{proof}
The left side follows from optimality of $\xx^{\star}_{2k}$. For the other side, we have the following relation,
\[
 \norm{\xxtil}_{2k}^{2k} \leq \norm{\xxtil}_{k}^{2k} \leq O(1)\norm{\xx^{\star}_{k}}_{k}^{2k} \leq O(1)\norm{\xx^{\star}_{2k}}_{k}^{2k} \leq O(1)m^{2k\left(\frac{1}{k}-\frac{1}{2k}\right)}\norm{\xx^{\star}_{2k}}_{2k}^{2k} = O(m)\norm{\xx^{\star}_{2k}}_{2k}^{2k}.
\]
\end{proof}

Consider the following procedure to obtain a starting point $\xx^{(0)}$ for the $\ell_p$-norm problem.

\begin{algorithm}[H]
\caption{Homotopy on $p$ for Starting Solution}
\label{alg:homotopy}
 \begin{algorithmic}[1]
 \Procedure{\textsc{StartSolution}}{$\AA, \NN, \bb,p$}
\State $\xx^{(0)} \leftarrow 0, k \leftarrow 2$
\While{$k \leq 2^{\lfloor \log p - 1\rfloor}$}
\State $\xx^{(0)} \leftarrow$ {\sc Main-Solver} $(\AA,0,\NN,0,\bb,k,1)$\Comment{2-approximate solution to the $k$-norm Problem}
\State $k \leftarrow 2k$
\EndWhile
\State \Return $\xx^{(0)}$
 \EndProcedure 
 \end{algorithmic}
\end{algorithm}

\begin{lemma}
Let $\xx^{(0)}$ be as returned by Algorithm \ref{alg:homotopy}. Suppose there exists an oracle that solves the residual problem for any norm $\ell_k$, i.e., $\res_k$ to a $\kappa_k$-approximation in time $T(k,\kappa_k)$. We can then compute $\xx^{(0)}$ which is a $O(m)$-approximation to the $\ell_p$-norm problem, in time at most 
\[
O\left(p \log m \right) \sum_{k = 2^i,i = 2}^{i = \lfloor \log p - 1\rfloor} \kappa_k T(k,\kappa_k).
\]
\end{lemma}
\begin{proof}
For any $k$, we have an $O(1)$-approximation solution to the $k/2$-norm solution. From Lemma \ref{lem:homotopy}, this is a $O(m)$-approximate solution to the $k$-norm problem. We now have from Theorem \ref{thm:IterativeRefinement}, that we require $O\left(k\kappa_k T(k,\kappa_k) \log m\right)$ time to solve the $k$-norm problem to a constant approximation. Summing over all $k$, we have total runtime,
\[
T = \sum_{k = 2^i,i = 2}^{i = \lfloor \log p - 1\rfloor} O(k \kappa_k T(k,\kappa_k) \log m ) \leq O\left(p \log m \right) \sum_{k = 2^i,i = 2}^{i = \lfloor \log p - 1\rfloor} \kappa_k T(k,\kappa_k).
\]
\end{proof}
In later sections, we will describe an oracle that will have $\kappa_k = O(1)$ for all values of $k$ and $T(k,\kappa_k)$ depends on $k$ linearly.

\subsection{Iterative Refinement for \texorpdfstring{$p \in (1,2)$}{TEXT}}

We will consider the following pure $\ell_p$ problem here, 
\begin{equation}\label{eq:ProblemSmall}
\min_{\AA\xx = \bb} \|\NN\xx\|_p,
\end{equation}
where $p \in (1,2)$. In the previous sections we saw an iterative refinement framework that worked for $p\geq 2$. In this section, we will show a similar iterative refinement for $p \in (1,2)$. In particular, we will prove the following result from \cite{AdilKPS19}.

\begin{restatable}{theorem}{IterRefSmall}\label{thm:IterRefSmall}
Let $p \in (1,2)$, and $\kappa \geq 1$. Given an initial solution $\xx^{(0)}$ satisfying $\AA\xx^{(0)} = \bb$, we can find $\xxtil$ such that $\AA\xxtil = \bb$ and $\|\NN\xxtil\|_p^p \leq (1+\epsilon) \|\xx^{\star}\|_p^p$ in $O\left(\left(\frac{p}{p-1}\right)^{\frac{1}{p-1}} \kappa \log\frac{m}{\epsilon}\right)$ calls to a $\kappa$-approximate solver to the residual problem (Definition \ref{res:small}).
\end{restatable}

The key idea in the algorithm for $p\geq 2$ was an upper and lower bound on the function that was an $\ell_2^2 + \ell_p^p$-norm term (Lemma \ref{lem:precondition}). Such a bound does not hold when $p<2$, however, we will show that a smoothed $\ell_p$-norm function can be used for providing such bounds. Specifically, we use the following smoothed $\ell_p$-norm function defined in \cite{bubeck2018homotopy}.

\begin{definition}(Smoothed $\ell_p$ Function.)
Let $p \in (1,2)$, and $x\in \mathbb{R},t \geq 0$ . We define,
\[
\gamma_p(t,x) = \begin{cases}
\frac{p}{2}t^{p-2}x^2 & \text{ if }|x| \leq t,\\
|x|^p - \left(1 - \frac{p}{2}\right)t^p & \text{ otherwise.}
\end{cases}
\]
For any vector $\xx$ and $\tt \geq 0$, we define $\gamma_p(\tt,\xx) = \sum_i \gamma_p(\tt_i,\xx_i)$.
\end{definition}
We define the following residual problem for this section.
\begin{definition}\label{res:small}
For $p \in (1,2)$, we define the residual problem at any feasible $\xx$ to be,
\[
\max_{\AA\Delta = 0} \res_p(\Delta) \defeq \gg^{\top}\Delta - 2^p \gamma_p(|\NN\xx|,\NN\Delta),
\]
where $\gg = p \NN^{\top}|\NN\xx|^{p-2}\NN\xx$.
\end{definition}

We will follow a similar structure as Section \ref{sec:IterRef}. We begin by proving analogues of Lemma \ref{lem:RelateResidualOpt} and Lemma \ref{lem:precondition}. 
\begin{lemma}\label{lem:precondition-small}
Let $p\in(1,2)$. For any $x$ and $\Delta$,
\[
|x|^p + p|x|^{p-2}x\Delta + \frac{p-1}{p2^p} \gamma_p(|x|,\Delta) \leq|x+\Delta|^p \leq |x|^p + p|x|^{p-2}x\Delta + 2^p \gamma_p(|x|,\Delta)
\]

\begin{proof} We first show the following inequality holds
for $|\alpha| \leq1 $
\begin{equation}\label{smallalpha}
1 + \alpha p + \frac{(p-1)}{4} \alpha^2 \leq (1+\alpha)^p \leq 1 + \alpha p + p2^{p-1} \alpha^2.
\end{equation}

Let us first show the left inequality, i.e. $1 + \alpha p + \frac{p-1}{4} \alpha^2 \leq (1+\alpha)^p $. Define the following function,
\begin{equation*}
h(\alpha) = (1+\alpha)^p - 1 - \alpha p - \frac{p-1}{4} \alpha^2.
\end{equation*}
When $\alpha = 1,-1$, $h(\alpha) \geq 0$. The derivative of $h$ with respect to $\alpha$ is, $h'(\alpha) = p(1+ \alpha)^{p-1} - p- \frac{(p-1)}{2} \alpha $. Next let us see what happens when $\abs{\alpha} <1$. 

\[
h''(\alpha) = p(p-1)(1+\alpha)^{p-2} - \frac{p-1}{2} = (p-1) \left(\frac{p}{(1+\alpha)^{2-p}} - \frac{1}{2}\right) \geq 0
\]
This implies that $h'(\alpha)$ is an increasing function of $\alpha$ and $\alpha_0$ for which $h'(\alpha_0) = 0$ is where $h$ attains its minimum value. The only point where $h'$ is 0 is $\alpha_0 = 0$. This implies $h(\alpha) \geq h(0) = 0$. This concludes the proof of the left inequality. For the right inequality, define:
\begin{equation*}
s(\alpha) = 1 + \alpha p + p2^{p-1} \alpha^2 - (1+\alpha)^p.
\end{equation*}
Note that $s(0) = 0$ and $s(1),s(-1) \geq 0$. We have, 
\[
s'(\alpha) = p + p2^{p} \alpha - p(1+\alpha)^{p-1},
\]
and
\[ 
(1+\alpha)^{p-1}sign(\alpha) \leq (1+\alpha)sign(\alpha).
\]
Using this, we get, $s'(\alpha) sign(\alpha) \geq p|\alpha|  (2^{p} - 1) \geq 0$ which says $s'(\alpha)$ is positive for $\alpha$ positive and negative for $\alpha$ negative. Thus the minima of $s$ is at 0 which is $0$. So $s(\alpha) \geq0$.

Before we prove the lemma, we will prove the following inequality for $\beta\geq 1$,
\begin{equation}\label{beta}
 (\beta -1 )^{p-1} +1 \geq  \frac{1}{2^p} \beta^{p-1}.
\end{equation}

$(\beta -1 ) \geq \frac{\beta}{2}$ for $\beta \geq 2$. So the claim clearly holds for $\beta \geq 2$ since $(\beta -1 )^{p-1} \geq \left(\frac{\beta}{2}\right)^{p-1}$. When $1 \leq \beta \leq 2$, $1 \geq \frac{\beta}{2}$, so the claim holds since, $1 \geq \left(\frac{\beta}{2}\right)^{p-1}$

\noindent We now prove the lemma.

  Let $\Delta = \alpha x$. The term $p |x|^{p-1} sign(x) \cdot \alpha x = \alpha p |x|^{p-1}|x| = \alpha p |x|^p$. Let us first look at the case when $|\alpha| \leq 1$. We want to show, 
\begin{align*}
& |x|^p + \alpha p |x|^p + c\frac{p}{2}|x|^{p-2}|\alpha x|^2  \leq |x+\alpha x|^p \leq |x|^p + \alpha p |x|^p + C\frac{p}{2}|x|^{p-2}|\alpha x|^2 \\
& \iff (1+ \alpha p) + c\frac{p}{2}\alpha^2  \leq (1+\alpha)^p \leq (1 + \alpha p)+ C\frac{p}{2}\alpha^2.
\end{align*}
This follows from Equation~\eqref{smallalpha} and the facts $\frac{cp}{2} \leq \frac{p-1}{4}$ and $\frac{Cp}{2} \geq p2^{p-1}$ . We next look at the case when $|\alpha| \geq 1$.  Now, $\gamma_p (|f|,\Delta) = |\Delta|^p + (\frac{p}{2} - 1)|f|^p$. We need to show

\begin{equation*}
|x|^p (1 + \alpha p)+ \frac{|x|^p(p-1)}{p2^p}( |\alpha|^p + \frac{p}{2} - 1) \leq |x|^p|1 + \alpha|^p \leq  |x|^p (1 + \alpha p)+2^p |x|^p( |\alpha|^p + \frac{p}{2} - 1).
\end{equation*}

When $|x| = 0$ it is trivially true. When $|x| \neq 0$, let 
\begin{equation*}
h(\alpha) = |1+\alpha|^p - (1+\alpha p) - \frac{(p-1)}{p2^p}(|\alpha|^p + \frac{p}{2} - 1).
\end{equation*}
Now, taking the derivative with respect to $\alpha$ we get,
\begin{equation*}
h'(\alpha) = p \left( |1+\alpha|^{p-1}sign(\alpha) - 1 - \frac{(p-1)}{p2^p} |\alpha|^{p-1}sign(\alpha)\right).
\end{equation*}
We use the mean value theorem to get for $|\alpha| \geq 1$,
\begin{align*}
(1+\alpha)^{p-1} - 1 &= (p-1) \alpha (1+z)^{p-2}, z \in (0,\alpha)\\
& \geq (p-1)\alpha (2\alpha)^{p-2} \\
& \geq \frac{p-1}{2} \alpha^{p-1}
\end{align*}
which implies $h'(\alpha) \geq 0$ in this range as well. When $\alpha \leq -1$ it follows from Equation~\eqref{beta} that $h'(\alpha) \leq0$. So the function $h$ is increasing for $\alpha \geq 1$ and decreasing for $\alpha \leq -1$. The minimum value of $h$ is $min \{h(1), h(-1) \} \geq 0$.  It follows that $h(\alpha) \geq 0$ which gives us the left inequality. The other side requires proving,
\begin{equation*}
|1+\alpha|^p \leq 1+\alpha p + 2^p (|\alpha|^p + \frac{p}{2} -1).
\end{equation*}
Define:
\begin{equation*}
s(\alpha) = 1+\alpha p + 2^p (|\alpha|^p + \frac{p}{2} -1) - |1+\alpha|^p.
\end{equation*}
The derivative $s'(\alpha) = p + \left(p 2^p |\alpha|^{p-1} - p|1+\alpha|^{p-1} \right)sign(\alpha)$ is non negative for $\alpha \geq 1$ and non positive for $\alpha \leq -1$.  The minimum value taken by $s$ is $\min\{s(1),s(-1)\} $ which is non negative. This gives us the right inequality.

\end{proof}
\end{lemma}

\begin{lemma}\label{lem:RelateResidualObj-Small}
Let $p \in (1,2)$ and $\lambda^{p-1} = \frac{p4^p}{p-1}$. Then for any $\Delta$,
\[
\res_p(\Delta) \leq \|\NN\xx\|_p^p - \|\NN(\xx-\Delta)\|_p^p,
\]
and
\[
\|\NN\xx\|_p^p -\|\NN(\xx-\lambda\Delta)\|_p^p  \leq \lambda\res_p(\Delta).
\]
\end{lemma}
\begin{proof}
Applying Lemma \ref{lem:precondition-small} to all coordinates,
\[
 - \gg^{\top}\Delta + \frac{p-1}{p2^p} \gamma_p(|\NN\xx|,\NN\Delta) \leq \|\NN(\xx-\Delta)\|_p^p - \|\NN\xx\|_p^p \leq - \gg^{\top}\Delta + 2^p \gamma_p(|\NN\xx|,\NN\Delta).
\]
From the definition of the residual problem and the above equation, the first inequality of our lemma directly follows. To see the other inequality, from the above equation,
\begin{align*}
\|\NN\xx\|_p^p - \|\NN(\xx-\lambda\Delta)\|_p^p  & \leq   \lambda \gg^{\top}\Delta - \frac{p-1}{p2^p} \gamma_p(|\NN\xx|,\lambda \NN\Delta)\\
& \leq \lambda \left(\gg^{\top}\Delta - \lambda^{p-1}\frac{p-1}{p2^p} \gamma_p(|\NN\xx|, \NN\Delta)\right)\\
& = \lambda \cdot \res_p(\Delta).
\end{align*}
Here, we are using the following property of $\gamma_p$,
\[
\gamma_p(t,\lambda\Delta) \geq \min\{\lambda^2,\lambda^p\}\gamma_p (t,\Delta).
\]
\end{proof}

Lemma \ref{lem:RelateResidualObj-Small} is similar to Lemma \ref{lem:RelateResidualOpt}, and we can follow the proof of Theorem \ref{thm:IterativeRefinement} to obtain Theorem \ref{thm:IterRefSmall}.

%% file: Chapters/General_Oracle.tex

\renewcommand{\union}{\cup}

\section{Fast Multiplicative Weight Update Algorithm for \texorpdfstring{$\ell_p$}{TEXT}-norms}
\label{chap:MWU}
In this section, we will show how to solve the residual problem for $p\geq 2$ as defined in the previous section (Definition \ref{def:residual}), to a constant approximation. The core of our approach is a multiplicative weight update routine with {\it width reduction} that is used to speed up the algorithm. For problem instances of size $m$, this routine returns a constant approximate solution in at most $O(m^{1/3})$ calls to a linear system solver. Such a width reduced multiplicative weight update algorithm was first seen in the context of the maximum flow problem and $\ell_{\infty}$-regression in works by \textcite{CKMST,chin2013runtime}

The first instance of such a width reduced multiplicative weight update algorithm for $\ell_p$-regression appeared in the work of \textcite{AdilKPS19}. In a further work, the authors improved the dependence on $p$ in the runtime \cite{AdilBKS21}. The following sections are based on the improved algorithm from \textcite{AdilBKS21}.

\subsection{Algorithm for \texorpdfstring{$\ell_p$}{TEXT}-norm Regression}

Recall that our residual problem for $p\geq 2$ is defined as:
\[
     \max_{\AA\Delta = 0} \quad \res_p(\Delta) \defeq \gg^{\top}\Delta -
     \Delta^{\top}\RR\Delta - \|\NN\Delta\|_p^p,
\] 
for some vector $\gg$ and matrices $\RR$ and $\NN$. Also recall that in Algorithm \ref{alg:IterRef}, we used a parameter $\nu$, which was used to track the value of $\ff(\xx^{(t)}) - \ff(\xx^{\star})$ at any iteration $t$. We will now use this parameter $\nu$ to do a binary search on the linear term in $\res_p$ and reduce the residual problem to,
\begin{align}\label{eq:BinaryProblem}
\begin{aligned}
\min_{\Delta} & \quad\Delta^{\top}\RR\Delta + \|\NN\Delta\|_p^p\\
s.t. & \quad \AA\Delta = 0\\
& \quad \gg^{\top}\Delta = c,
\end{aligned}
\end{align}
for some constant $c$. Further, we will use our multiplicative weight update solver to solve problems of this kind to a constant approximation. We start by proving the binary search results.

\subsubsection{Binary Search}
\label{sec:BinarySearch}

We first note that, if $\nu$ at iteration $t$ is such that $\ff(\xx^{(t)})-\ff(\xx^{\star}) \in (\nu/2,\nu] $, then from Lemma \ref{lem:RelateResidualOpt}, the residual at $\xx^{(t)}$ has optimum value, $\res_p(\Dopt) \in (\frac{\nu}{32p},\nu]$. We now consider a parameter $\zeta$ that has value between $\frac{\nu}{16p}$ and $\nu$ such that $\res_p(\Dopt) \in (\frac{\zeta}{2},\zeta]$. We have the following lemma that relates the optimum of problem of the type \eqref{eq:BinaryProblem} with $\zeta$.

\begin{restatable}{lemma}{binary}\label{lem:binary}
Let $\zeta$ be such that the residual problem satisfies $\res_p(\Dopt) \in (\frac{\zeta}{2},\zeta]$. The following problem has optimum at most $2\zeta$.
\begin{align}\label{eq:BinarySearch}
\begin{aligned}
\min_{\Delta} & \quad\Delta^{\top}\RR\Delta + \|\NN\Delta\|_p^p\\
s.t. & \quad \AA\Delta = 0\\
& \quad \gg^{\top}\Delta = \frac{\zeta}{2}.
\end{aligned}
\end{align}
Further, let $\Dtil$ be a solution to the above problem such that $\Dtil^{\top}\RR\Dtil \leq a^2 \zeta$ and $\|\NN\Dtil\|_p^p \leq a^p\zeta$ for some $a>1$. Then $\frac{\Dtil}{5a^2}$ is a $100a^2$-approximation to the residual problem.
\end{restatable}
\begin{proof}
We have assumed that,
\[
\res(\Dopt) = \gg^{\top}\Dopt -\Dopt^{\top}\RR\Dopt - \norm{\NN\Dopt}_p^p \in  \left(\frac{\zeta}{2},\zeta\right].
\]
Since the last $2$ terms are strictly non-positive, we must have, $\gg^{\top}\Dopt  \geq \frac{\zeta}{2}.$
Since $\Dopt$ is the optimum and satisfies $\AA\Dopt = 0$, 
\[
\frac{d}{d\lambda}\left(\gg^{\top}\lambda \Dopt -  \lambda^2\Dopt^{\top}\RR\Dopt - \lambda^p \norm{\NN\Dopt}_p^p \right)_{\lambda = 1} =0.
\]
Thus,
\[
\gg^{\top} \Dopt -   \Dopt^{\top}\RR\Dopt  - \norm{\NN\Dopt}_p^p  =  \Dopt^{\top}\RR\Dopt  + (p-1) \norm{\NN\Dopt}_p^p.
\]
Since $p \ge 2,$ we get the following
\[
\Dopt^{\top}\RR\Dopt  +  \norm{\NN\Dopt}_p^p \leq \gg^{\top} \Dopt -   \Dopt^{\top}\RR\Dopt  -  \norm{\NN\Dopt}_p^p \leq \zeta.
\]
Now, we know that, $\gg^{\top}\Dopt  \geq \frac{\zeta}{2}$ and $ \gg^{\top} \Dopt -  \Dopt^{\top}\RR\Dopt  -  \norm{\NN\Dopt}_p^p \leq  \zeta$. This gives, 
\[
\frac{\zeta}{2}\leq \gg^{\top} \Dopt \leq  \Dopt^{\top}\RR\Dopt +  \norm{\NN\Dopt}_p^p + \zeta \leq 2\zeta.
\]
Now, let $\Dtil$ be as described in the lemma. We have,
\begin{align*}
\res_p\left(\frac{\Dtil}{5a^2}\right) & = \frac{1}{5a^2}\gg^{\top}\Dtil - \frac{\zeta}{25a^2} - \frac{\zeta}{5^p a^p}\\
& \geq \frac{\zeta}{10a^2} - \frac{2\zeta}{25a^2}\\
& \geq \frac{\zeta}{50a^2}\geq \frac{1}{100a^2}\res_p(\Dopt)
\end{align*}
\end{proof}

\begin{algorithm}[H]
\caption{Algorithm for Solving the Residual Problem}
\label{alg:ResSol}
 \begin{algorithmic}[1]
 \Procedure{\textsc{ResidualSolver}}{$\xx,\MM,\NN,\AA,\dd,\bb,\nu,p$}
\State $\zeta \leftarrow\nu$
\State $(\gg,\RR,\NN) \leftarrow \res_p$\Comment{Create residual problem at $\xx$}
\While{$\zeta > \frac{\nu}{32p}$}
\State $\Dtil_{\zeta} \leftarrow$ {\sc MWU-Solver}$\left([\AA,\gg^{\top}], \RR^{1/2},\NN, [0,\frac{\zeta}{2}]^{\top},\zeta,p\right)$\Comment{Algorithm \ref{alg:FasterOracleAlgorithm}}
\State $\zeta \leftarrow \frac{\zeta}{2}$
\EndWhile
\State \Return $\arg\min_{\Dtil_{\zeta}} \ff\left(\xx - \frac{\Dtil_{\zeta}}{p}\right)$
 \EndProcedure 
 \end{algorithmic}
\end{algorithm}

\subsubsection{Width-Reduced Approximate Solver}
 \label{sec:AKPSOracle}

We are now finally ready to solve problems of the type \eqref{eq:BinaryProblem}. In this section, we will give an algorithm to solve the following problem,

  \begin{align}\label{eq:Un-ScaledProblem}
 \min_{\Delta}  & \quad \Delta^{\top}\MM^{\top}\MM\Delta + \|\NN\Delta\|_p^p \\
 & \text{s.t.} \quad \AA\Delta = \cc. \nonumber
  \end{align}
Here $\AA \in \mathbb{R}^{d \times n},\NN \in \mathbb{R}^{m_1 \times n},\MM\in \mathbb{R}^{m_2 \times n}$, and vector $\cc \in \rea^{d}$. Our approach involves a multiplicative weight update method with a {\it width reduction} step which allows us to solve these problems faster. 

\subsubsection{Slow Multiplicative Weight Update Solver}
 We first give an informal analysis of the multiplicative
  weight update method without width reduction. We will show that this method converges in $\approx m_1^{\frac{p-2}{2(p-1)}}\leq  m_1^{1/2}$ iterations. For simplicity, we let
  $\MM =0$ in Problem~\eqref{eq:Un-ScaledProblem} and assume without
  loss of generality that the optimum $\Dopt$ satisfies, $\|\NN\Dopt\|_p \leq 1$. Consider the following MWU algorithm for parameter $\alpha$ that we will set later:
\begin{enumerate}
  \item $\ww^{(0)} = 1, \xx^{(0)}=0, T = \alpha^{-1}m^{1/p}$
  \item for $t =1,\cdots, T$: \\
  $ \Delta^{(t)} = \arg\min_{\AA\Delta=\cc}\sum_i(\ww^{(t-1)}_i)^{p-2}(\NN\Delta)_i^2, \quad \ww^{(t)} = \ww^{(t-1)} + \alpha|\NN\Delta^{(t)}|,\quad \xx^{(t)} = \xx^{(t-1)} + \Delta^{(t)}$ 
  \item Return $\xxtil = \xx/T$
  \end{enumerate}
We claim that the above algorithm returns $\xxtil$ such that $\|\NN\xxtil\|_p^p \leq O_p(1)$, i.e., a constant approximate solution to the residual problem, in $ \approx m_1^{1/2}$ iterations. We will bound the value of the returned solution, $\|\NN\xxtil\|_p^p$ by looking at how $\|\ww^{(t)}\|_p^p$ grows with $t$. From Lemma~\ref{lem:precondition},
\begin{multline*}
\|\ww^{(t-1)}+\alpha\NN\Delta^{(t)}\|_p^p \leq \|\ww^{(t-1)}\|_p^p + \alpha p \sum_i (\ww^{(t-1)}_i)^{p-1}(\NN\Delta^{(t)})_i \\+ 2p^2 \alpha^2 \sum_i (\ww_i^{(t-1)})^{p-2}(\NN\Delta^{(t)})^2_i + \alpha^p p^p \|\NN\Delta^{(t)}\|_p^p.
\end{multline*}
Observe that the third term on the right hand side is exactly the
objective of the quadratic problem minimized to obtain
$\Delta^{(t)}$. Using that $\Delta^{(t)}$ must achieve a lower
objective than $\Dopt$, i.e., $\sum_i
(\ww_i^{(t-1)})^{p-2}(\NN\Delta^{(t)})^2_i \leq \sum_i
(\ww_i^{(t-1)})^{p-2}(\NN\Dopt)^2_i $ along with Hölder's inequality
and $\|\NN\Dopt\|_p\leq 1$, we can bound this term by
$\|\ww^{(t-1)}\|_p^{p-2}$. We can further bound the second term in
right hand side of the above inequality by the third term using
Hölder's inequality (refer to Proof of
Lemma~\ref{lem:ReduceWidthGammaPotential} for details). These bounds
give, 
\[
\|\ww^{(t)}\|_p^p \leq \|\ww^{(t-1)}\|_p^p + \alpha p \|\ww^{(t-1)}\|^{p-1}_p + 2\alpha^2p^2 \|\ww^{(t-1)}\|_p^{p-2} + \alpha^p p^p\|\NN\Delta^{(t)}\|_p^p.
\]
Observe that the growth of $\|\ww^{(t)}\|_p^p$ is controlled by $\|\NN\Delta^{(t)}\|^p_p$. We next see how large this quantity can be. Assume that, $\|\ww^{(t)}\|_p \leq 3m_1^{1/p}$ for all $t$ (one may verify in the end that this holds for all $t \leq T$). Since $(\ww^{(t-1)}_i)^{p-2} \geq (\ww^{(0)}_i)^{p-2} =1$,
\[
\|\NN\Delta^{(t)}\|_2^2 \leq \sum_i (\ww_i^{(t-1)})^{p-2}(\NN\Delta^{(t)})^2_i \substack{(a)\\\leq} \|\ww^{(t-1)}\|_p^{p-2} \leq 3^{p-2}m_1^{(p-2)/p},
\]
where we used Hölder's inequality in $(a)$.
This implies, $\|\NN\Delta^{(t)}\|_p^p \leq 3^{(p-2)p/2} m_1^{(p-2)/2} $. Now, for $\alpha \approx m_1^{-\frac{p^2-4p+2}{2p(p-1)}}$, $\alpha^p p^p \|\NN\Delta^{(t)}\|_p^p \leq \alpha p m_1^{\frac{p-1}{p}} \leq \alpha p \|\ww^{(t-1)}\|_p^{p-1}$ and,
\[
\|\ww^{(t)}\|_p^p \leq \|\ww^{(t-1)}\|_p^p + \alpha p \|\ww^{(t-1)}\|^{p-1}_p + 2\alpha^2p^2 + \alpha p \|\ww^{(t-1)}\|_p^{p-1} \leq\left(\|\ww^{(t-1)}\|_p + 2\alpha \right)^p.
\]
We can thus prove that,
\[
\|\NN\xx^{(T)}\|_p^p \leq \frac{1}{m_1}\|\ww^{(T)}\|_p^p \leq \left(\|\ww^{(0)})\|_p + 2\alpha T\right)^p = \frac{1}{m_1}\left(m_1^{1/p} + 2m_1^{1/p}\right)^p = 3^p,
\]
as required. The total number of iterations is $T = \alpha^{-1}m_1^{1/p} \approx m_1^{\frac{p-2}{2(p-1)}}$.

To obtain the improved rates of convergence via width reduction, our
algorithm uses a hard threshold on $\|\NN\Delta^{(t)}\|_p^p$ and performs a {\it width reduction step} whenever $\|\NN\Delta^{(t)}\|_p^p$ is larger than the threshold. The
analysis now requires to additionally track how $\|\ww\|_p$ changes with a width reduction step. Our analysis also tracks the value of an additional potential
$\Psi = \min_{\AA\Delta = \cc}\sum_i \ww_i\Delta_i^2$. The interplay
of these two potentials and balancing out their changes with respect
to primal updates and width reduction steps give the improved rates of
convergence.

\subsubsection{Fast, Width-Reduced MWU Solver} In the previous section, we showed that a multiplicative weight update algorithm without width-reduction obtains a rate of convergence $\approx m_1^{1/2}$. In this section we will show how width-reduction allows for a faster $\approx m_1^{1/3}$ rate of convergence. We now present the faster width-reduced algorithm. We will prove the
following result.
\begin{restatable}{theorem}{AKPSAlgo}
\label{cor:ResidualDecision}
Let $p\geq 2$. Consider an instance of Problem \eqref{eq:Un-ScaledProblem} described by matrices $\AA \in \mathbb{R}^{d \times n},\NN \in \mathbb{R}^{m_1 \times n},\MM\in \mathbb{R}^{m_2 \times n}$, and vector $\cc \in \rea^{d}$. If the optimum of this problem is at most $\zeta$, Procedure {\sc Residual-Solver} (Algorithm~\ref{alg:FasterOracleAlgorithm}) returns an $\xx$ such that
$\AA \xx = \cc,$ and $\xx^{\top}\MM^{\top}\MM\xx \leq O(1)\zeta$ and
$\|\NN\xx\|_p^p \leq O(3^p)\zeta$. The algorithm makes ${O} \left(pm_1^{\frac{p-2}{(3p-2)}}\right)$ calls to a linear system solver.
\end{restatable}
The algorithm and analyses of this chapter are based on \cite{AdilKPS19} and \cite{AdilBKS21}.

 In every iteration of the algorithm, we solve a weighted linear system. The solution returned is used to update the current iterate if it has a small $\ell_p$ norm. Otherwise, we do not update the solution, but update the weights corresponding to the coordinates with large value by a constant factor. This step is refered to as the ``width reduction step". The analysis is based on a potential function argument for specially defined potentials.

The following is the oracle used in the algorithm, i.e., the linear system we need to solve. We show in the Appendix \ref{chap:l2solve} how to implement the oracle using a linear system solver.
\begin{algorithm}[H]
\caption{Oracle}
\label{alg:oracle}
 \begin{algorithmic}[1]
 \Procedure{\textsc{Oracle}}{$\AA, \MM,\NN, \cc, \ww,\zeta$}
\State $\rr_e \leftarrow \ww_e^{p-2}$
\State $\Mtil \leftarrow \zeta^{-\frac{p-2}{2p}}\MM$
\label{algline:resistance}
\State Compute,
\[
  \Delta = \arg\min_{\AA \Delta' =  \cc} \quad   m_1^{\frac{p-2}{p}}{\Delta'}^{\top}\Mtil^{\top}\Mtil\Delta' + \frac{1}{3^{p-2}} \sum_e \rr_e \left(\NN\Delta^{'}\right)^2_e
\]
 \State\Return $ \Delta$
 \EndProcedure 
 \end{algorithmic}
\end{algorithm}
We now have the following multiplicative weight update algorithm given in Algorithm \ref{alg:FasterOracleAlgorithm}.
\begin{algorithm}
\caption{Width Reduced MWU Algorithm}
\label{alg:FasterOracleAlgorithm}
 \begin{algorithmic}[1]
 \Procedure{MWU-Solver}{$\AA, \MM,\NN, \cc,\zeta,p$}
 \State $\ww^{(0,0)}_e \leftarrow 1$
 \State $\xx \leftarrow 0$
 \State $\rho \leftarrow m_1^{\frac{(p^2-4p+2)}{p(3p-2)}}$\Comment{width parameter}
 \State $\beta \leftarrow 3^{p-1}\cdot m_1^{\frac{p-2}{3p-2}}$\Comment{resistance threshold}
 \State $\alpha \leftarrow 3^{-\frac{p-1}{p}}\cdot p^{-1} m_1^{-\frac{p^2-5p+2}{p(3p-2)}}$\Comment{step size}
 \State $\tau \leftarrow 3^p\cdot m_1^{\frac{(p-1)(p-2)}{(3p-2)}}$\Comment{$\ell_p$ threshold}
\State $T \leftarrow \alpha^{-1} m_1^{1/p} = 3^{\frac{p-1}{p}} \left( p m^{\frac{p-2}{3p-2}}\right)$ 
\State $i \leftarrow 0, k \leftarrow 0$
\While{$i < T$} 
\State $\Delta \leftarrow \textsc{Oracle}(\AA, \MM,\NN, \cc, \ww^{(i,k)},\zeta)$
\State $\rr \leftarrow \left(\ww^{(i,k)}\right)^{p-2}$
\If{$\norm{\NN\Delta}_{p}^p \leq \tau \zeta$}\Comment{primal step}   \label{algline:CheckWidth}
\State $\ww^{(i+1,k)} \leftarrow \ww^{(i,k)} + \alpha \frac{\abs{\NN\Delta}}{\zeta^{1/p}}$ 
\State $\xx \leftarrow \xx +  \Delta$
\State $ i \leftarrow i+1$ 
\Else
\State For all coordinates $e$ with $|\NN\Delta|_e \geq \rho \zeta^{\frac{1}{p}}$ and $\rr_e \leq \beta$\Comment{width reduction step}\label{lin:WidthReduceEdge}
\State $\quad \quad \quad \ww_e^{(i,k+1)} \leftarrow 2^{\frac{1}{p-2}}\ww_e$  
\State $ \quad \quad \quad k \leftarrow k+1$
\EndIf
\EndWhile
 \State\Return $\frac{\xx}{T}$
 \EndProcedure 
\end{algorithmic}
\end{algorithm}

\paragraph*{Notation}
We will use $\Dopt$ to denote the optimum of \eqref{eq:Un-ScaledProblem}. Since we assume that the optimum value of \eqref{eq:Un-ScaledProblem} is at most $\zeta$, 
\begin{equation}
\Dopt^{\top}\MM^{\top}\MM\Dopt \leq \zeta \quad \text{ and } \quad \norm{\NN\Delta^*}_p^p \leq \zeta
\end{equation}

\subsubsection{Analysis of Algorithm \ref{alg:FasterOracleAlgorithm}}

Our analysis is based on tracking the following two potential functions. We will show how these potentials change with a primal step (Line \ref{algline:CheckWidth}) and a width reduction step (\ref{lin:WidthReduceEdge}) in the algorithm. The proofs of these lemmas appear later in the section.

\[
  \Phi\left(\ww^{\left( i \right)} \right) \defeq \norm{\ww}_p^p
\]
\[
\Psi(\rr)\defeq \min_{\Delta: \AA \Delta = \cc} m_1^{\frac{p-2}{p}}{\Delta}^{\top}\Mtil^{\top}\Mtil\Delta + \frac{1}{3^{p-2}}\sum_e \rr_e \left(\NN\Delta\right)^2_e.
\]
Finally, to prove our runtime bound, we will first show that if the total number of width reduction steps $K$ is not too large, then $\Phi$ is bounded. We then prove that the number of width reduction steps cannot be too large by using the relation between $\Phi$ and $\Psi$ and their respective changes throughout the algorithm.

We now begin our analysis. The next two lemmas show how our potentials change with every iteration of the algorithm.

\begin{restatable}{lemma}{ReduceWidthGammaPotential}
  \label{lem:ReduceWidthGammaPotential}
  After $i$ primal steps, and $k$ width-reduction steps,
  provided $p^p \alpha^p \tau \leq p \alpha m_1^{\frac{p-1}{p}}$,
  the potential $\Phi$ is bounded as follows:
  \begin{align*}
 \Phi\left(\ww^{(i,k)}\right) \leq \left(2\alpha i + m_1^{\nfrac{1}{p}} \right)^{p} \left(1+ \frac{2^{\frac{p}{p-2}}}{\rho^2 m_1^{2/p} \beta^{-\frac{2}{p-2}}} \right)^k.
 \end{align*}
\end{restatable}

\begin{restatable}{lemma}{ReduceWidthElectricalPotential}
\label{lem:ReduceWidthElectricalPotential}
After $i$ primal steps and $k$ width reduction steps, if,
\begin{enumerate}
\item \label{eq:parametersEnsuringElectricalPotentialGrowth1} $\tau^{2/p}\zeta^{2/p} \geq 4 \cdot 3^{p-2} \frac{\Psi(\rr)}{\beta}$, and
\item \label{eq:parametersEnsuringElectricalPotentialGrowth2} $\tau \zeta^{2/p}\geq 2 \cdot 3^{p-2}\Psi(\rr) \rho^{p-2}$,
\end{enumerate}
then,
\[
 {\Psi\left({\rr^{(i,k+1)}}\right)} \geq {\Psi\left({\rr^{(0,0)}}\right)} + \frac{k}{4} \cdot \tau^{2/p} \zeta^{2/p}.
\]
\end{restatable}

The next lemma gives a lower bound on the energy in the beginning and an upper bound on the energy at each step.

\begin{restatable}{lemma}{ElectricalPotentialStartFinishBounds}
  \label{lem:ElectricalPotentialStartFinishBounds}
Let $i$ denote the number of primal steps and $k$ the number of width reduction steps. For any $i,k \geq 0$, we have,
\[
\Psi\left({\rr^{(i,k)}}\right) \le   \zeta^{2/p}\left(m_1^{\frac{p-2}{p}} + \frac{1}{3^{p-2}}\Phi(i,k)^{\frac{p-2}{p}}\right).
\]

\end{restatable}

\subsubsection{Proof of Theorem \ref{cor:ResidualDecision}}\label{sec:ProofMWUFast}
\begin{proof}
Let $\frac{\xx}{T}$ be the solution returned by Algorithm \ref{alg:FasterOracleAlgorithm}. We first note that this satisfies the linear constraint required. We next bound the objective value at $\frac{\xx}{T}$, i.e., $\frac{1}{T^2}\xx^{\top}\MM^{\top}\MM\xx$ and $\frac{1}{T^p}\|\NN\xx\|_p^p$.

Suppose the algorithm terminates in $T = \alpha^{-1} m_1^{1/p}$ primal steps and $K \leq 2^{-\frac{p}{p-2}} \rho^2 m_1^{2/p} \beta^{-\frac{2}{p-2}}$ width reduction steps. We next note that our parameter values $\alpha$ and $\tau$ are such that $p^p \alpha^p \tau \leq p \alpha m_1^{\frac{p-1}{p}}$. We can now apply Lemma \ref{lem:ReduceWidthGammaPotential} to get,
\[
\Phi\left(\ww^{(T,K)}\right) \le 3^p m_1 e^{1} = e \cdot 3^p m_1
\]
We next observe from the weight and $\xx$ update steps in our algorithm that, $ \zeta^{1/p}m_1^{-1/p} \ww^{(T,K)} \geq |\NN\xx|$. Thus,
\begin{align*}
\frac{1}{T}\|\NN\xx\|_p^p \leq \frac{\zeta}{m_1} \norm{\ww^{(T,K)}}_p^p  =\frac{\zeta}{m_1} \Phi\left(\ww^{(T,K)}\right) \leq e\cdot 3^p\zeta.
\end{align*}
We next bound the quadratic term. Let $\Dtil^{(t)}$ denote the solution returned by the oracle in iteration $t$. Since $\Phi \leq e\cdot 3^p m_1$ for all iterations, we always have from Lemma \ref{lem:ElectricalPotentialStartFinishBounds} that, $\Psi(\rr) \leq 4 m_1^{\frac{p-2}{p}}\zeta^{2/p}$. We will first bound $\left(\Dtil^{(t)}\right)^{\top}\MM^{\top}\MM\Dtil^{(t)}$ for every $t$.
\[
\left(\Dtil^{(t)}\right)^{\top}\MM^{\top}\MM\Dtil^{(t)} = \zeta^{\frac{p-2}{p}}\left(\Dtil^{(t)}\right)^{\top}\Mtil^{\top}\Mtil\Dtil^{(t)} \leq \zeta^{\frac{p-2}{p}} m_1^{-\frac{p-2}{p}} \Psi(\rr) \leq 4\zeta.
\]
Now from convexity of $\|\xx\|_2^2$, we get
\begin{align*}
\norm{\MM\frac{\xx}{T}}_2^2 \leq \frac{1}{T^2} \cdot T \sum_t \|\MM\Dtil^{(t)}\|_2^2 \leq  4\zeta.
\end{align*}
We have shown that if the number of width reduction steps is bounded by $K$ then our algorithm returns the required solution. We will next prove that we cannot have more than $K$ width reduction steps.

Suppose to the contrary, the algorithm takes a width reduction step starting from step $(i,k)$ where $i < T$ and  $k = 2^{-\frac{p}{p-2}} \rho^2 m_1^{2/p}\beta^{-\frac{2}{p-2}}$. Since the conditions for Lemma~\ref{lem:ReduceWidthGammaPotential} hold for all preceding steps, we must have $\Phi\left(\ww^{(i,k)}\right) \leq e\cdot 3^p m_1$ which combined with Lemma \ref{lem:ElectricalPotentialStartFinishBounds} implies $\Psi \leq 4 m_1^{\frac{p-2}{p}} \zeta^{2/p}$. Using this bound on $\Psi$, we note that our parameter values satisfy the conditions of Lemma \ref{lem:ReduceWidthElectricalPotential}. From lemma \ref{lem:ReduceWidthElectricalPotential},
\[
 {\Psi\left({\rr^{(i,k+1)}}\right)} \geq {\Psi\left({\rr^{(0,0)}}\right)}  +  \frac{1}{4}\tau^{2/p}\zeta^{2/p}k.
\]
Since our parameter choices ensure $\tau^{2/p}k > \frac{1}{4}m_1$,
\[
 {\Psi\left({\rr^{(i,k+1)}}\right)} - {\Psi\left({\rr^{(0,0)}}\right)}   >  \frac{m_1}{16}\zeta^{2/p}.
\]
Since $\Phi\left(\ww^{(i,k)}\right) \leq O(3^p) m_1$ and $\Psi \geq 0$, from Lemma \ref{lem:ElectricalPotentialStartFinishBounds},
\[
\Psi\left({\rr^{(i,k+1)}}\right) - \Psi\left({\rr^{(0,0)}}\right)
\leq 4 m_1^{\frac{p-2}{p}} \zeta^{2/p},
\]
which is a contradiction. We can thus conclude that we can never have more than $K = 2^{\frac{-p}{p-2}}\rho^2 m_1^{2/p} \beta^{-\frac{2}{p-2}}$ width reduction steps, thus concluding the correctness of the returned solution. We next bound the number of oracle calls required. The total number of iterations is at most,
\[
T + K \leq \alpha^{-1}m_1^{1/p} +  2^{-p/(p-2)}\rho^2 m_1^{2/p} \beta^{-\frac{2}{p-2}} \leq O \left(p m_1^{\frac{p-2}{3p-2}}\right).
\]
\end{proof}


\subsubsection{Proof of Lemma \ref{lem:ReduceWidthGammaPotential}}

We first prove a simple lemma about the solution $\Dtil$ returned by the oracle, that we will use in our proof.

\begin{restatable}{lemma}{Oracle}
\label{lem:Oracle}
Let $p\geq 2$. For any $\ww$, let $\Dtil$ be the solution returned by Algorithm \ref{alg:oracle}. Then,
\[
\sum_e (\NN\Dtil)_e^2 \leq \sum_e \rr_e(\NN\Dtil)_e^2 \leq  \zeta^{\frac{2}{p}} \norm{\ww}^{p-2}
\]
\end{restatable} 

\begin{proof}
Since $\Dtil$ is the solution returned by Algorithm \ref{alg:oracle}, and $\Dopt$ satisfies the constraints of the oracle, we have,
\begin{align*}
  \sum_e \rr_e (\NN\Dtil)_e^2 \leq \sum_e \rr_e (\NN\Delta^*)_e^2 & = \sum_e \ww_e^{p-2} (\NN\Delta^*)_e^2 \leq   \zeta^{2/p }\norm{\ww}^{p-2}_p.
\end{align*}
In the last inequality we use,
\begin{align*}
\sum_e \ww_e (\NN\Dopt)_e^2 & \leq \left(\sum_e (\NN\Dopt)_e^{2\cdot \frac{p}{2}} \right)^{2/p}  \left(\sum_e \abs{\ww_e}^{(p-2)\cdot \frac{p}{p-2}}\right)^{(p-2)/p}\\
& = \|\NN\Dopt\|_p^2 \|\ww\|_p^{(p-2)/p} \\
& \leq \zeta^{2/p}\|\ww\|_p^{(p-2)/p}, \text{since $\norm{\NN\Delta^*}^p_p \leq \zeta$ }.
 \end{align*}
 Finally, using $\rr_e \ge  1,$ we have
$\sum_e (\NN \Delta)_e^{2} \le \sum_e \rr_e (\NN\Delta)_e^{2},$ concluding the proof.
\end{proof}

\ReduceWidthGammaPotential*

\begin{proof}
  We prove this claim by induction. Initially, $i = k = 0,$
  and $\Phi\left(\ww^{(0,0)}\right) = m_1,$ and thus, the claim holds trivially. Assume that the claim holds for some $i,k \ge 0.$
We will use $\Phi$ as an abbreviated notation for $\Phi\left(\ww^{(i,k)}\right)$ below.
\paragraph*{Primal Step.}
For brevity, we use $\ww$ to denote $\ww^{(i,k)}$. If the next step is a \emph{primal} step, 

\begin{align*}
\Phi\left(\ww^{( i+1,k)} \right) = &\norm{ \ww^{(i,k)} + \alpha \frac{\abs{\NN\Dtil}}{\zeta^{1/p}}}_p^p \\
\leq & \norm{\ww}_p^p + \zeta^{-1/p} \alpha \abs{(\NN\Dtil)}^{\top}\abs{\nabla \|\ww\|_p^p}  +  2p^2\alpha^2 \zeta^{-2/p} \sum_e |\ww_e|^{p-2} \abs{\NN\Dtil}_e^2 + \alpha^p p^p \zeta^{-1}\|\NN\Dtil\|_p^p\\
&   \quad \text{by Lemma \ref{lem:precondition}}
\end{align*}
We next bound $\abs{(\NN\Dtil)}^{\top} \abs{\nabla \|\ww\|_p^p}$ by $
\zeta^{1/p} p \norm{\ww}_p^{p-1}.$
Using Cauchy Schwarz's inequality,
\begin{align*}
 \left(\sum_e \abs{\NN\Dtil}_e \abs{\nabla_e \|\ww\|_p^p}\right)^2  = & p^2\left(\sum_e  \abs{\NN\Dtil}_e  |\ww_e|^{p-2}\abs{\ww_e}\right)^2\\
 \leq & p^2 \left(\sum_{e} |\ww_e|^{p-2} \ww_e^2\right) \left(\sum_e |\ww_e|^{p-2}(\NN\Dtil)_e^2\right) \\
  = & p^2 \|\ww\|_p^p\sum_e \rr_e (\NN\Dtil)_e^2\\
  \leq & p^2 \|\ww\|_p^{2p-2}\zeta^{2/p}, \text{ From Lemma \ref{lem:Oracle}.}
\end{align*}
We thus have,
\begin{align*}
 \sum_e \abs{\NN\Dtil}_e\abs{\nabla_e \|\ww\|_p^p} &\leq  p \|\ww\|_p^{p-1}\zeta^{1/p}.
\end{align*}
Using the above bound, we now have,
\begin{align*}
\Phi\left(\ww^{( i+1,k)} \right)\leq &  \norm{\ww}_p^p +  p\alpha \norm{\ww}_p^{p-1} +  2p^2 \alpha^2  \norm{\ww}_p^{p-2} + p^p \alpha^p  \|\NN\Dtil\|_p^p\\
\leq &  \norm{\ww}_p^p +  p\alpha \norm{\ww}_p^{p-1} +  2p^2 \alpha^2  \norm{\ww}_p^{p-2} + p \alpha m_1^{\frac{p-1}{p}},\\
      & \quad \text{(since $p^p \alpha^p \tau \leq p\alpha m_1^{\frac{p-1}{p}}$)}
\end{align*}
Recall $\|\ww\|_p^p = \Phi(\ww).$ Since $\Phi \geq m_1$, we have,
    
\[
\Phi\left(\ww^{( i+1,k)} \right) \leq  \Phi(\ww) +  p \alpha \Phi(\ww)^{\frac{p-1}{p}} +  2 p^2 \alpha^{2}  \Phi(\ww)^{\frac{p-2}{p}} +  p \alpha \Phi(\ww)^{\frac{p-1}{p}} \leq (\Phi(\ww)^{1/p} + 2  \alpha)^{p}.
\]
From the inductive assumption, we have
\begin{align*}
\Phi(\ww) & \le \left({2\alpha  i} + m_1^{\nfrac{1}{p}} \right)^{p}\left(1+ \frac{2^{\frac{p}{p-2}}}{\rho^2 m_1^{2/p} \beta^{-\frac{2}{p-2}}} \right)^k.
\end{align*}
Thus,
\[
\Phi(i+1,k)  \le (\Phi(\ww)^{1/p} + 2 \alpha )^{p} \le \left({2\alpha (i+1)} + m_1^{\nfrac{1}{p}}\right)^{p} \left(1+ \frac{2^{\frac{p}{p-2}}}{\rho^2 m_1^{2/p} \beta^{-\frac{2}{p-2}}} \right)^k
\]
proving the inductive claim.

 \paragraph*{Width Reduction Step.}
 Let $\Dtil$ be the solution returned by the oracle and $H$ denote the set of indices $j$ such that $|\NN\Dtil|_j \geq \rho \zeta^{1/p}$ and $\rr_j \leq \beta$, i.e., the set of indices on which the algorithm performs width reduction.
 We have the following:
\[
\sum_{j \in H} \rr_j \leq \rho^{-2}\zeta^{-2/p} \sum_{j \in H} \rr_e (\NN\Delta)_j^2 \leq\rho^{-2} \zeta^{-2/p} \sum_{j} \rr_j (\NN\Delta)_e^2 \leq \rho^{-2}  \|\ww\|_p^{p-2}\leq \rho^{-2}\Phi^{\frac{p-2}{p}},
\]
where we use Lemma \ref{lem:Oracle} for the second last inequality. Also,
\begin{align*}
\Phi(\ww^{(i,k+1)}) & \le \Phi + \sum_{j \in H} \abs{ \ww_j^{k+1}}^p \leq \Phi + 2^{\frac{p}{p-2}} \sum_{j \in H}  |\ww_j|^{p} \leq \Phi + 2^{\frac{p}{p-2}} \sum_j \rr_j^{\frac{p}{p-2}} \\
& \leq  \Phi + 2^{\frac{p}{p-2}} \left( \sum_{j \in H} \rr_j \right)\left( \max_{j \in H} \rr_j \right) ^{\frac{p}{p-2}-1}   \leq \Phi + 2^{\frac{p}{p-2}}  \rho^{-2}\Phi^{\frac{p-2}{p}}\beta^{\frac{2}{p-2}}.
\end{align*}
Again, since $\Phi(\ww) \geq m_1$,
 \[
 \Phi(\ww^{(i,k+1)})  \le \Phi \left( 1+ 2^{\frac{p}{p-2}} \rho^{-2} m_1^{-\frac{2}{p}}\beta^{\frac{2}{p-2}} \right) \le \left(2\alpha  i + m_1^{\nfrac{1}{p}}\right)^{p} \left(1+ \frac{2^{\frac{p}{p-2}}}{\rho^2 m_1^{2/p} \beta^{-\frac{2}{p-2}}} \right)^k
\]
    proving
    the inductive claim.
\end{proof}

\subsubsection{Proof of Lemma \ref{lem:ReduceWidthElectricalPotential}}

\ReduceWidthElectricalPotential*
\begin{proof}
It will be helpful for our analysis to split the index set into three
disjoint parts:
\begin{itemize}
\item $S =  \setof{e : \abs{\NN\Delta_e} \leq \rho \zeta^{1/p} }$
\item $H = \setof{e : \abs{\NN\Delta_e} > \rho \zeta^{1/p} \text{ and } \rr_e \leq \beta }$
\item $B = \setof{e : \abs{\NN\Delta_e} > \rho \zeta^{1/p} \text{ and } \rr_e > \beta }$.
\end{itemize}
Firstly, we note 
\begin{align*}
\sum_{e \in S} \abs{\NN\Delta}_e^{p} \leq \rho^{p-2} \zeta^{\frac{p-2}{p}} \sum_{e \in S} \abs{\NN\Delta}_e^{2} 
 \leq \rho^{p-2} \zeta^{\frac{p-2}{p}} \sum_{e \in S} \rr_e \abs{\NN\Delta}_e^{2}  \leq \rho^{p-2} \zeta^{\frac{p-2}{p}} 3^{p-2}\Psi(\rr).
\end{align*}
hence, using Assumption~\ref{eq:parametersEnsuringElectricalPotentialGrowth2}
\begin{align*}
  \sum_{e \in H \union B} \abs{\NN \Delta}_e^p
  \geq \sum_{e} \abs{\NN\Delta}_e^p - \sum_{e \in S}
  \abs{\NN\Delta}_e^{p}
  \geq \tau \zeta - \rho^{p-2}\zeta^{\frac{p-2}{p}} 3^{p-2} \Psi(\rr)
  \geq \frac{1}{2} \tau \zeta.
\end{align*}
This means, 
\[
\sum_{e \in H \union B}(\NN \Delta)_e^2 \geq  \left(\sum_{e \in H \union B} |\NN\Delta|_e^p\right)^{2/p} \geq \frac{ \tau^{2/p} \zeta^{2/p}}{2}.
\]
Secondly we note that, 
\[
\sum_{e \in B} (\NN\Delta)_e^2 \leq \beta^{-1} \sum_{e \in B} \rr_e (\NN\Delta)_e^2 \leq \beta^{-1} 3^{p-2} \Psi(\rr).
\] 
So then, using Assumption~\ref{eq:parametersEnsuringElectricalPotentialGrowth1},
\begin{align*}
\sum_{e \in H } (\NN\Delta)_e^2 = \sum_{e \in H \union B } (\NN\Delta)_e^2 - \sum_{e \in B}(\NN \Delta)_e^2  \geq 
\frac{ \tau^{2/p} \zeta^{2/p}}{2} -\beta^{-1} 3^{p-2}\Psi(\rr) \geq  \frac{\tau^{2/p}\zeta^{2/p}}{4}.
\end{align*}
As $\rr_e \geq 1$, this implies $\sum_{e \in H } \rr_e (\NN\Delta)_e^2\geq  \frac{\tau^{2/p} \zeta^{2/p}}{4}$. We note that in a width reduction step, the resistances change by a factor of 2. Thus, combining our last two observations, and applying Lemma~\ref{lem:ckmst:res-increase}, we get
\[
{\Psi\left({\rr^{(i,k+1)}}\right)} \geq {\Psi\left({\rr^{(i,k)}}\right)}+  \frac{1}{4} \tau^{2/p} \zeta^{2/p}.
\]
Finally, for the ``primal step'' case, we use the trivial bound from Lemma~\ref{lem:ckmst:res-increase}, ignoring the second term, 
\[
{\Psi\left({\rr^{(i,k+1)}}\right)} \geq {\Psi\left({\rr^{(i,k)}}\right)}. 
\]
\end{proof}

\subsubsection{Proof of Lemma \ref{lem:ElectricalPotentialStartFinishBounds}}
\ElectricalPotentialStartFinishBounds*
\begin{proof}

Lemma~\ref{lem:Oracle} implies that,
  \begin{align*}
{\Psi\left({\rr^{(i,k)}}\right) } & =\zeta^{-(p-2)/p}m_1^{\frac{p-2}{p}} \Dtil^{\top}\MM^{\top}\MM\Dtil + \frac{1}{3^{p-2}}\sum_e \rr_e (\NN\Dtil)_e^2 \\
& \le \zeta^{-(p-2)/p}m_1^{\frac{p-2}{p}} \Dopt^{\top}\MM^{\top}\MM\Dopt + \frac{1}{3^{p-2}}\sum_e \rr_e (\NN\Dopt)_e^2\\
& \le  \zeta^{2/p} m_1^{\frac{p-2}{p}} + \zeta^{2/p}\frac{1}{3^{p-2}} \norm{\ww}_p^{p-2} \\
& \le \zeta^{2/p}m_1^{\frac{p-2}{p}} + \zeta^{2/p}\frac{1}{3^{p-2}}\Phi(i,k)^{\frac{p-2}{p}}.
  \end{align*}
\end{proof}

\subsection{Complete Algorithm for \texorpdfstring{$\ell_p$}{TEXT}-Regression}

Recall our problem, \eqref{eq:Main},
\[
\min_{\AA\xx = \bb}\quad \ff(\xx) = \dd^{\top}\xx + \|\MM\xx\|_2^2 + \|\NN\xx\|_p^p.
\]
We will now use all the tools and algorithms described so far to give a complete algorithm for the above problem. We will assume we have a starting solution $\xx^{(0)}$ satisfying $\AA\xx^{(0)} =\bb$ and for purely $\ell_p$ objectives, we will use the homotopy analysis from Section \ref{sec:homotopy}.

Our overall algorithm reduces the problem to solving the residual problem (Definition \ref{def:residual}) approximately. In Sections \ref{sec:BinarySearch} and \ref{sec:AKPSOracle}, we give an algorithm to solve the residual problem by first doing a binary search on the linear term and then applying a multiplicative weight update routine to minimize these problems. We have the following result which follows from Lemma \ref{lem:binary} and Theorem \ref{cor:ResidualDecision}.
\begin{corollary}\label{cor:ResApprox}
Consider the residual problem at iteration $t$ of Algorithm \ref{alg:IterRef}. Algorithm \ref{alg:ResSol} using Algorithm \ref{alg:FasterOracleAlgorithm} as a subroutine finds a $O(1)$-approximate solution to the corresponding residual problem in $O\left(pm^{\frac{p-2}{3p-2}} \log p\right)$ calls to a linear system solver.
\end{corollary}
\begin{proof}
Let $\nu$ be such that $\ff(\xx^{(t)})-\ff(\xx^{\star}) \in (\nu/2,\nu]$. Refer to Lemma \ref{lem:invariant} to see that this is the case in which we use the solution of the residual problem. Now, from Lemma \ref{lem:RelateResidualOpt} we know that the optimum of the residual problem satisfies $\res_p(\Dopt) \in (\nu/32p,\nu]$. Since we vary $\zeta$ to take all such values in the range $(\nu/16p,\nu]$ for one such $\zeta$ we must have $\res_p(\Dopt) \in (\zeta/2,\zeta].$ For such a $\zeta$, consider problem \eqref{eq:BinarySearch}. Using Algorithm \ref{alg:FasterOracleAlgorithm} for this problem, from Theorem \ref{cor:ResidualDecision} we are guaranteed to find a solution $\Dtil$ such that $\Dtil^{\top}\RR\Dtil \leq O(1) \zeta$ and $\|\NN\Dtil\|_p^p \leq O(3^p)\zeta$. Now from Lemma \ref{lem:binary}, we note that $\Dtil$ is an $O(1)$-approximate solution to the residual problem. Since Algorithm \ref{alg:FasterOracleAlgorithm} requires $O\left(pm^{\frac{p-2}{3p-2}} \right)$ calls to a linear system solver, and Algorithm \ref{alg:ResSol} calls this algorithm $\log p$ times, we obtain the required runtime.
\end{proof}
We are now ready to prove our main result.
\begin{restatable}{theorem}{CompleteAlgo}\label{thm:CompleteAlgorithm}
Let $p \geq 2$, and $\kappa \geq 1$. Let the initial solution $\xx^{(0)}$ satisfying $\AA\xx^{(0)} = \bb$. Algorithm \ref{alg:IterRef} using Algorithm \ref{alg:ResSol} as a subroutine returns an $\epsilon$-approximate solution $\xx$ to \eqref{eq:Main} in at most $O\left(p^2 m^{\frac{p-2}{3p-2}}\log p \log \left(\frac{\ff(\xx^{(0)})-\ff(\xx^{\star})}{\epsilon }\right)\right)$ calls to a linear system solver.
\end{restatable}
\begin{proof}
Follows from Theorem \ref{thm:IterativeRefinement} and Corollary \ref{cor:ResApprox}.
\end{proof}

\subsection{Complete Algorithm for Pure \texorpdfstring{$\ell_p$}{TEXT} Objectives}

Consider the special case when our problem is only the $\ell_p$-norm, i.e., Problem \eqref{eq:lpObj},
\[
\min_{\AA\xx=\bb}\|\NN\xx\|_p^p.
\]
In Section \ref{sec:homotopy} we described how to find a good starting point for such problems. Combining this algorithm with our algorithm for solving the residual problem we can obtain a complete algorithm for finding a good starting point. Specifically, we prove the following result.
\begin{corollary}\label{cor:homotopyFull}
Algorithm \ref{alg:homotopy} using Algorithm \ref{alg:ResSol} returns $\xx^{(0)}$ such that $\AA\xx^{(0)} = \bb$ and $\|\NN\xx^{(0)}\|_p^p\leq O(m)\|\NN\xx^{\star}\|_p^p$ in $O\left(p^2  m^{\frac{p-2}{3p-2}}\log^2 p \log m\right)$ calls to a linear system solver.
\end{corollary}
\begin{proof}
From Lemma \ref{lem:homotopy} Algorithm \ref{alg:homotopy} finds such a solution in time $O\left(p \log m \right) \sum_{k = 2^i,i = 2}^{i = \lfloor \log p - 1\rfloor} \kappa_k T(k,\kappa_k)$, where $\kappa_k$ and $T(k,\kappa_k)$ denote the approximation and time to solve a $\ell_k$ norm problem. Now consider Algorithm \ref{alg:ResSol} with Algorithm \ref{alg:FasterOracleAlgorithm} as a subroutine. From Corollary \ref{cor:ResApprox}, we can solve any $\ell_k$-norm residual problem to a $O(1)$-approximation in $O\left(k  m^{\frac{k-2}{3k-2}}\log k\right)$ calls to a linear system solver. We thus have $\kappa_k = O(1)$ for all $k$ and $T(k,\kappa_k) = O\left(k  m^{\frac{k-2}{3k-2}}\log k\right) \leq O\left(p  m^{\frac{p-2}{3p-2}}\log p\right)$. Using these values, we obtain a runtime of,
\[
O\left(p \log m \right) \sum_{k = 2^i,i = 2}^{i = \lfloor \log p - 1\rfloor} \kappa_k T(k,\kappa_k) \leq O\left(p \log m \right) \cdot \log p \cdot O\left(p  m^{\frac{p-2}{3p-2}}\log p\right) \leq O\left(p^2  m^{\frac{p-2}{3p-2}}\log^2 p \log m\right).
\]
\end{proof}
The following theorem gives a complete runtime for pure $\ell_p$ objectives.

\begin{corollary}
Let $p \geq 2$, and $\kappa \geq 1$. Let $\xx^{(0)}$ be the solution returned by Algorithm \ref{alg:homotopy}. Algorithm \ref{alg:IterRef} using Algorithm \ref{alg:ResSol} as a subroutine returns $\xx$ such that $\AA\xx = \bb$ and $\|\NN\xx\|_p^p \leq (1 + \epsilon) \|\NN\xx^{\star}\|_p^p$, in at most $O\left(p^2 m^{\frac{p-2}{3p-2}}\log^2 p \log \left(\frac{m}{\epsilon }\right)\right)$ calls to a linear system solver.
\end{corollary}
\begin{proof}
Follows directly from Corollary \ref{cor:homotopyFull} and Theorem \ref{thm:CompleteAlgorithm}.
\end{proof}


%% file: Chapters/pNorm2qNorm.tex

\section{Solving \texorpdfstring{$p$}{TEXT}-norm Problems using \texorpdfstring{$q$}{TEXT}-norm Oracles}
\label{sec:p2q}

In this section, we propose a new technique that allows us to solve $\ell_p$-norm residual problems by instead solving an $\ell_q$-norm residual problem without adding much to the runtime. Such a technique is unknown for pure $\ell_p$ objectives without a large overhead in the runtime. As a consequence we also obtain an algorithm for $\ell_p$-regression with a linear runtime dependence on $p$ instead of the $p^2$ dependence in the algorithms from previous sections. The $p^2$ dependence in algorithms had one $p$ factor resulting from solving the $p$-norm residual problem. At a high level, we show that it is sufficient to solve a $\log m$-norm residual problem when $p$ is large, thus replacing a $p$-factor with $\log m$.  
We prove the following results which are based on the proofs and results of \textcite{AdilS20}.

\begin{restatable}{theorem}{CompleteLp}\label{thm:CompleteLp}
Let $\epsilon >0 $, $2 \leq p\leq poly(m)$ and consider an instance of Problem \eqref{eq:Main}, 
\[
\min_{\AA\xx = \bb} \quad \ff(\xx) = \dd^{\top}\xx + \|\MM\xx\|_2^2 + \|\NN\xx\|_p^p. 
\]
Algorithm \ref{alg:CompleteLp} finds an $\epsilon$-approximate solution to \eqref{eq:Main} in $O\left(p m^{\frac{p-2}{3p-2}}\log p \log m\log \frac{\ff(\xx^{(0)})-\ff(\xx^{\star})}{\epsilon}\right)$ calls to a linear system solver.
\end{restatable}

\begin{restatable}{theorem}{CompletePureLp}\label{thm:CompletePureLp}
Let $\epsilon >0 $, $2 \leq p\leq poly(m)$ and consider a pure $\ell_p$ instance, 
\[
\min_{\AA\xx = \bb} \quad \|\NN\xx\|_p^p. 
\]
Let $\xx^{(0)}$ be the output of Algorithm \ref{alg:homotopy}. Algorithm \ref{alg:CompleteLp} using $\xx^{(0)}$ as a starting solution finds $\xx$ such that $\AA\xx = \bb$ and $\|\NN\xx\|_p^p \leq (1+\epsilon) \min_{\AA\xx = \bb}  \|\NN\xx\|_p^p$ in $O\left(p  m^{\frac{p-2}{3p-2}} \log^2 p \log m\log \frac{m}{\epsilon }\right)$ calls to a linear system solver.
\end{restatable}

\subsection{Relation between Residual Problems for \texorpdfstring{$\ell_p$}{TEXT} and \texorpdfstring{$\ell_q$}{TEXT} Norms}\label{sec:p2qNorms}
In this section we prove how $q$-norm residual problems can be used to solve $p$-norm residual problems. This idea first appeared in the work of \textcite{AdilS20}, where they also apply the results to the maximum flow problem. In this paper, we provide a much simpler proof for the main techncial content and unify the cases of $p<q$ and $p>q$ that were presented separately in previous works. We also unify the case of relating the decision versions of the residual problems (without the linear term) and the entire objective. The results for the maximum flow problem and $\ell_p$-norm flow problem as described in the original paper still follow and we refer the reader to the original paper for these applications. The main result of the section is as follows.

\begin{restatable}{theorem}{p2q}\label{thm:p2q}
Let $p,q\geq 2$ and $\zeta$ be such that $\res_p(\Dopt)\in(\zeta/2,\zeta]$, where $\Dopt$ is the optimum of the $\ell_p$-norm residual problem (Definition \ref{def:residual}). The following $\ell_q$-norm residual problem has optimum at least $\frac{\zeta}{4}$,
\begin{equation}\label{eq:resq}
\max_{\AA\Delta = 0} \gg^{\top} \Delta - \Delta^{\top}\RR\Delta - \frac{1}{4}\zeta^{1-\frac{q}{p}} m^{\min\left\{\frac{q}{p}-1,0\right\}}\|\NN\Delta\|_q^q.
\end{equation}
Let $\beta\geq 1$ and $\Dtil$ denote a feasible solution to the above $\ell_q$-norm residual problem with objective value at least $\frac{\zeta}{16\beta}$. For $\alpha = \frac{1}{256\beta}m^{-\frac{p}{p-1}\abs{\frac{1}{p}-\frac{1}{q}}}$, $\alpha\Dtil$ gives a $O(\beta^2) m^{\frac{p}{p-1}\abs{\frac{1}{p}-\frac{1}{q}}}$-approximate solution to the $\ell_p$-norm residual problem $\res_p$.
\end{restatable}
\begin{proof}
Consider $\Dopt$, the optimum of the $\ell_p$-norm residual problem. Note that $\lambda\Dopt$ is a feasible solution for all $\lambda$ since $\AA(\lambda\Dopt) = 0.$ We know that the objective is optimum for $\lambda = 1$. Thus,
\[
\left[\frac{d}{d\lambda}\res_p(\lambda \Dopt)\right]_{\lambda = 1} = 0,
\]
which gives us,
\[
\gg^{\top} \Dopt - 2\Dopt^{\top}\RR\Dopt - p\|\NN\Dopt\|_p^p = 0.
\]
Rearranging,
\[{}
\Dopt^{\top}\RR\Dopt + (p-1)\|\NN\Dopt\|_p^p = \gg^{\top} \Dopt - \Dopt^{\top}\RR\Dopt - \|\NN\Dopt\|_p^p \leq \zeta.
\]
Since $p \geq 2$, $\|\NN\Dopt\|_p \leq \zeta^{1/p}$ which implies 
\[
 \|\NN\Dopt\|_q \leq \begin{cases}
 \zeta^{1/p} & \text{if, $p\leq q$}\\
 m^{\frac{1}{q}-\frac{1}{p}}\zeta^{1/p} & \text{otherwise.}
 \end{cases}
 \]
We also note that,
\[
\gg^{\top} \Dopt - \Dopt^{\top}\RR\Dopt > \frac{\zeta}{2} + \|\NN\Dopt\|_p^p > \frac{\zeta}{2}.
\]
Combining these bounds, we obtain the optimum of \eqref{eq:resq} is at least,
\[
\gg^{\top} \Dopt - \Dopt^{\top}\RR\Dopt - \frac{1}{4}\zeta^{1-\frac{q}{p}}m^{\min\left\{\frac{q}{p}-1,0\right\}}\|\NN\Dopt\|_q^q > \frac{\zeta}{2} - \frac{1}{4}\zeta^{1-\frac{q}{p}}\zeta^{q/p} > \frac{\zeta}{4}.
\] 
Since the optimum of \eqref{eq:resq} is at least $\zeta/4$, there exists a feasible $\Dtil$ with objective value at least $\zeta/16\beta$. We now prove the second part, that a scaling of $\Dtil$ gives a good approximation to the $\ell_p$-norm residual problem. First, let us assume $|\gg^{\top}\Dtil| \leq \zeta$. Since $\Dtil$ has objective value at least $\zeta/16\beta$,
\[
\Dtil^{\top}\RR\Dtil + \frac{1}{4}\zeta^{1-\frac{q}{p}}m^{\min\left\{\frac{q}{p}-1,0\right\}}\|\NN\Dtil\|_q^q \leq \gg^{\top}\Dtil - \frac{\zeta}{16\beta} \leq \zeta.
\]
Thus, $m^{\min\left\{\frac{1}{p}-\frac{1}{q},0\right\}}\|\NN\Dtil\|_q \leq 4^{\frac{1}{q}} \zeta^{\frac{1}{p}}$, and $\|\NN\Dtil\|_p^p \leq 4^{\frac{p}{q}} \zeta m^{\abs{1 - \frac{p}{q}}}$. Let $\Dbar = \alpha \Dtil$, where $\alpha = \frac{1}{256\beta}m^{-\frac{p}{p-1}\abs{\frac{1}{p}-\frac{1}{q}}}$. We will show that $\alpha\Dbar$ is a good solution to the $\ell_p$-norm residual proble{}m.
\begin{align*}
\res_p(\alpha\Dbar) & = \alpha \left( \gg^{\top} \Dtil - \alpha \Dtil^{\top}\RR\Dtil - \alpha^{p-1} \|\NN\Dtil\|_p^p\right)\\
& \geq \alpha\left(\frac{\zeta}{16\beta} - \frac{1}{256\beta} \zeta - \alpha^{p-1}4^{\frac{p}{q}} \zeta m^{\abs{1 -\frac{p}{q}}} \right)\\
& \geq \alpha\left(\frac{\zeta}{16\beta} - \frac{\zeta}{256\beta}  -  \frac{\zeta}{64\beta}  \right) \\
& \geq \frac{\alpha}{64\beta} \res_p(\Dopt).
\end{align*}
For the case $|\gg^{\top}\Dtil| \geq \zeta$, consider the vector $z\Dtil$ where $z = \frac{\zeta}{2|\gg^{\top}\Dtil|} \leq \frac{1}{2}$. This vector is still feasible for Problem \eqref{eq:resq} and $\gg^{\top}z\Dtil = \frac{\zeta}{2}$ and,
\[
z\gg^{\top} \Dtil - z^2\Dtil^{\top}\RR\Dtil - z^q\frac{1}{4}\zeta^{1-\frac{q}{p}}m^{\min\left\{\frac{q}{p}-1,0\right\}}\|\NN\Dtil\|_q^q \geq \frac{\zeta}{2} - z^2 \zeta\geq \frac{\zeta}{4}. 
\]
We can now repeat the same argument as above.
\end{proof}

\subsection{Faster Algorithm for \texorpdfstring{$\ell_p$}{TEXT}-Regression}

In this section, we will combine the tools developed in previous chapters and combine it with Section \ref{sec:p2qNorms} to obtain an algorithm for Problem \ref{eq:Main} that requires $O\left(p m^{\frac{p-2}{3p-2}}\log p\log m \log \frac{\ff(\xx^{(0)})-\ff(\xx^{\star})}{\epsilon }\right)$ calls to a linear systems solver. For pure $\ell_p$ objectives we can combine our algorithm with the algorithm in Section \ref{sec:homotopy} to obtain a convergence rate of $O\left(p m^{\frac{p-2}{3p-2}}\log^2 p \log m\log \frac{m}{\epsilon }\right)$ linear systems solves.

\begin{algorithm}[H]
\caption{Complete Algorithm with Linear $p$-dependence}
\label{alg:CompleteLp}
 \begin{algorithmic}[1]
 \Procedure{\textsc{$\ell_p$-Solver}}{$\AA, \MM,\NN, \dd,\bb,p,\epsilon$}
\State $\xx \leftarrow \xx^{(0)}$
\State $\nu \leftarrow$ Upper bound on $ \ff(\xx^{(0)})-\ff(\xx^{\star})$\Comment{If $\ff(\xx^{\star})\geq 0$, then $\nu \leftarrow \ff(\xx^{(0)})$}
\While{$\nu >\epsilon$}
\If{$p \geq \log m$}
\State $\Dtil \leftarrow$ $\log m$-ResidualSolver($\xx,\MM,\NN,\AA,\dd,\bb,\nu,p$)
\Else 
\State $\Dtil \leftarrow$ ResidualSolver($\xx,\MM,\NN,\AA,\dd,\bb,\nu,p$)
\EndIf
\If{$\res_p(\Dtil) \geq \frac{\nu}{32p\kappa}$}
\State $\xx \leftarrow \xx - \frac{\Dtil}{p}$ 
\Else
\State $\nu \leftarrow \frac{\nu}{2}$
\EndIf
\EndWhile
\State \Return $\xx$
 \EndProcedure 
 \end{algorithmic}
\end{algorithm}

\begin{algorithm}[H]
\caption{Residual Solver using $\log m$-norm}
\label{alg:ResidualQ}
 \begin{algorithmic}[1]
 \Procedure{\textsc{$\log m$-ResidualSolver}}{$\xx,\MM,\NN,\AA,\dd,\bb,\nu,p$}
\State $\zeta \leftarrow\nu$
\State $\alpha \leftarrow  m^{-\frac{1}{p-1}}$
\State $(\gg,\RR,\NN) \leftarrow \res_p$\Comment{Create residual problem at $\xx$}
\While{$\zeta > \frac{\nu}{32p}$}
\State $\widetilde{\NN}\leftarrow \frac{1}{4^{1/\log m}}\zeta^{\frac{1}{\log m}-\frac{1}{p}} m^{\min\left\{\frac{1}{p}-\frac{1}{\log m},0\right\}}\NN$
\State $\Dtil_{\zeta} \leftarrow$ {\sc MWU-Solver}$\left([\AA,\gg^{\top}], \RR^{1/2},\widetilde{\NN}, [0,\frac{\zeta}{2}]^{\top},\zeta,\log m\right)$\Comment{Algorithm \ref{alg:FasterOracleAlgorithm}}
\State $\zeta \leftarrow \frac{\zeta}{2}$
\EndWhile
\State \Return $\alpha \Dtil \leftarrow \arg\min_{\Dtil_{\zeta}} \ff\left(\xx - \frac{\alpha \Dtil_{\zeta}}{p}\right)$
 \EndProcedure 
 \end{algorithmic}
\end{algorithm}

\begin{lemma}\label{lem:residualQ}
Let $ poly(m) \geq p \geq \log m$. Algorithm \ref{alg:ResidualQ} returns an $O(m^{\frac{1}{p-1}})$-approximate solution to the $\ell_p$-residual problem $\res_p$ at $\xx$ in at most $O\left(m^{\frac{p-2}{3p-2}} \log m \log p\right)$ calls to a linear system solver.
\end{lemma}
\begin{proof}
Let $\nu$ be such that $\ff(\xx^{(t)})-\ff(\xx^{\star}) \in (\nu/2,\nu]$. Refer to Lemma \ref{lem:invariant} to see that this is the case in which we use the solution of the residual problem. Now, from Lemma \ref{lem:RelateResidualOpt} we know that the optimum of the residual problem satisfies $\res_p(\Dopt) \in (\nu/32p,\nu]$. Since we vary $\zeta$ to take all such values in the range $(\nu/16p,\nu]$ for one such $\zeta$ we must have $\res_p(\Dopt) \in (\zeta/2,\zeta].$ For such a $\zeta$, consider the $\log m$-norm residual problem \eqref{eq:resq}.
Using Algorithm \ref{alg:FasterOracleAlgorithm} for this problem, from Theorem \ref{cor:ResidualDecision} we are guaranteed to find a solution $\Dtil$ such that $\Dtil^{\top}\RR\Dtil \leq O(1) \zeta$ and $\|\widetilde{\NN}\Dtil\|_{\log m}^{\log m} \leq O(3^p)\zeta$. Now from Lemma \ref{lem:binary}, we note that $\Dtil$ is an $O(1)$-approximate solution to the $\log m$-residual problem. We now use Theorem \ref{thm:p2q}, which states that $\alpha\Dtil$ is a $O\left(m^{\frac{1}{p-1}}\right)$-approximate solution to the required residual problem $\res_p$.

Since for $p \geq \log m$, Algorithm \ref{alg:FasterOracleAlgorithm} requires $O\left( m^{\frac{\log m-2}{3\log m-2}} \log m\right) \leq O\left( m^{\frac{p-2}{3p-2}} \log m\right)$ calls to a linear system solver, and Algorithm \ref{alg:ResSol} calls this algorithm $\log p$ times, we obtain the required runtime.
\end{proof}

\CompleteLp*
\begin{proof}
We note that Algorithm \ref{alg:CompleteLp} is essentially Algorithm \ref{alg:IterRef} which calls different residual solvers depending on the value of $p$. If $p \leq \log m$, from Theorem \ref{thm:CompleteAlgorithm}, we obtain the required solution in $O\left( m^{\frac{p-2}{3p-2}}\log p \log \frac{\ff(\xx^{(0)})-\ff(\xx^{\star})}{\epsilon }\right)$ calls to a linear system solver. If $p \geq \log m$, from Lemma \ref{lem:residualQ}, we obtain an $O(m^{\frac{1}{p-1}}) \leq O(m^{\frac{1}{\log m}}) \leq O(1)$ approximate solution to the residual problem at any iteration in $O\left(m^{\frac{p-2}{3p-2}} \log m \log p\right)$ calls to a linear system solver. Combining this with Theorem \ref{thm:IterativeRefinement}, we obtain our result.
\end{proof}

\CompletePureLp*
\begin{proof}
From Lemma \ref{lem:homotopy} we can find an $O(m)$-approximation to the above problem in time
\[
O(p\log m)\sum_{k = 2^i,i = 2}^{i = \lfloor \log p - 1\rfloor} \kappa_k T(k,\kappa_k),
\]
where $\kappa$ is the approximation to which we solve the residual problem for the $k$-norm problem and $T(k,\kappa)$ is the time required to do so. If $k \geq \log m$, we use Algorithm \ref{alg:ResidualQ} to solve such residual problems. Thus $\kappa_k = m^{\frac{1}{k-1}} \leq m^{\frac{1}{\log m}} \leq O(1)$ and $T(k,\kappa_k) = O\left( m^{\frac{p-2}{3p-2}}\log p\right)$. If $k \leq \log m$, we can use Algorithm \ref{alg:ResSol} and $\kappa_k = O(1)$, $T(k,\kappa_k) = O\left( m^{\frac{p-2}{3p-2}}\log p\right)$. Thus, the total runtime is $O\left(p m^{\frac{p-2}{3p-2}}\log m \log^2 p\right)$. We now combine this with Theorem \ref{thm:CompleteLp} to obtain the required rates of convergence.
\end{proof}

%% file: Chapters/InverseMaintenance.tex

\section{Speedups for General Matrices via Inverse Maintenance}
\label{sec:Inv}

Inverse maintenance was first introduced by Vaidya in 1990 \cite{vaidya1990solving} for speeding up algorithms for minimum cost and multicommodity flow problems. The key idea is to reuse the inverse of matrices, which is possible due to the controllable rates at which variables are updated in some algorithms. In the work by \textcite{AdilKPS19}, the authors design a new inverse maintenance algorithm for $\ell_p$-regression that can solve $\ell_p$-regression for any $p>2$ almost as fast as linear regression. This section is based on Section 6 of \cite{AdilKPS19} and we give a more fine grained and simplified analysis of the original result. In particular, we simplify the proofs and give the result with explicit dependencies on both matrix dimensions as opposed to just the larger dimension.

  Our inverse maintenance procedure is based on the same
  high-level ideas of combining low-rank updates and matrix
  multiplication as in \textcite{vaidya1990solving} and \textcite{LeeS15}. However,
  recall that the rate of convergence of our algorithm is controlled
  by two potentials which change at different rates based on the two
  different kind of weight update steps in our algorithm. In order to
  handle these updates, our inverse maintenance algorithm uses a new
  fine-grained bucketing scheme, inspired by lazy updates in data
  structures and is different from previous works on inverse
  maintenance which usually update weights based on fixed
  thresholds. Our scheme is also simpler than those used in
  \cite{vaidya1990solving, LeeS15}. We now present our algorithm in
  detail.

Consider the weighted linear system being solved at each iteration of Algorithm \ref{alg:FasterOracleAlgorithm}. Each weighted linear system is of the form,
\[
\min_{\AA\xx = \cc} \xx^{\top}\left(\MM^{\top}\MM + \NN^{\top}\RR\NN\right)\xx 
\]
where $\AA \in \mathbb{R}^{d \times n},\NN \in \mathbb{R}^{m_1 \times n},\MM\in \mathbb{R}^{m_2 \times n}$. From Equation \eqref{eq:optimizer} in Section \ref{chap:MWU}, the solution of the above linear system is given by,
\[
\xx^{\star} = \left(\MM^{\top}\MM+ \NN^{\top}\RR\NN\right)^{-1}\AA^{\top}\left(\AA \left(\MM^{\top}\MM+ \NN^{\top}\RR\NN\right)^{-1} \AA^{\top}\right)^{-1}\cc.
\]
In order to compute the above expression, we require the following products in order. The runtimes are considering the fact $\omega \geq 2$.
\begin{itemize}
  \item $\MM^{\top}\MM$ and $\NN^{\top}\RR\NN$: require time $m_2 n^{\omega-1}$ and $m_1 n^{\omega - 1}$ respectively
  \item $\left(\MM^{\top}\MM+ \NN^{\top}\RR\NN\right)^{-1}$: requires time $n^{\omega}$
  \item $\left(\MM^{\top}\MM+ \NN^{\top}\RR\NN\right)^{-1}\AA^{\top}$ and $\AA \left(\MM^{\top}\MM+ \NN^{\top}\RR\NN\right)^{-1} \AA^{\top}$: require time $n^2d^{\omega - 2}$
  \item $\left(\AA \left(\MM^{\top}\MM+ \NN^{\top}\RR\NN\right)^{-1} \AA^{\top}\right)^{-1}$: requires time $d^{\omega}$
  \item $\left(\MM^{\top}\MM+ \NN^{\top}\RR\NN\right)^{-1}\AA^{\top}\left(\AA \left(\MM^{\top}\MM+ \NN^{\top}\RR\NN\right)^{-1} \AA^{\top}\right)^{-1}$: requires time $nd^{\omega -1 }$
  \end{itemize}

The cost of solving the above problem is dominated by the first step, and we thus require time $O(mn^{\omega-1})$, where $m = \max\{m_1,m_2\}$. This directly gives the runtime of Algorithm \ref{alg:FasterOracleAlgorithm} to be $O\left(pm^{\frac{p-2}{(3p-2)}} mn^{\omega-1}\right)$. In this section, we show that we can implement Algorithm \ref{alg:FasterOracleAlgorithm} in time similar to solving a system of linear equations for all $p \geq 2$. In particular, we prove the following result.

\begin{restatable}{theorem}{InverseMaint}\label{thm:InverseMaint}
If $\AA,\MM,\NN$ are explicitly given, matrices with polynomially bounded condition numbers,
and $p \geq 2$, Algorithm~\ref{alg:FasterOracleAlgorithm} as given in Section~\ref{sec:AKPSOracle} can be implemented to run in total time
\[
O \left(mn^{\omega-1} + p^{3-\omega}n^2 m^{\omega-2} + p^{3-\omega}n^2 m^{\frac{p - \left(10 - 4 \omega \right)}{3p - 2}}\right).
\]
\end{restatable}

\subsection{Inverse Maintenance Algorithm}

We first note that the weights $\ww_e^{(i)}$'s, and thus $\rr_{e}^{(i)}$'s are monotonically increasing. Our algorithm in Section \ref{sec:AKPSOracle} updates both in every iteration. Here, we will instead update these gradually when there is a significant increase in the values. We thus give a lazy update scheme. The update can be done via the following consequence of the Woodbury matrix formula. The main idea is that we initially explicitly compute the inverse of the required matrix, and then when we update the coordinates that have significant increases, but are still within a good factor approximation of the original values, and directly use the current matrix inverse as a preconditioner and solve linear systems faster.

\subsubsection{Low Rank Update}

The following lemma is the same as Lemma 6.2 of \textcite{AdilKPS19}.
\begin{lemma}
\label{lem:LowRankUpdate}
Given matrices $\NN\in \mathbb{R}^{m_1 \times n},\MM\in \mathbb{R}^{m_2 \times n}$, and vectors $\rr$ and $\tilde{\rr}$ that
differ in $k$ entries, as well as the matrix
 $\widehat{\ZZ} = (\MM^{\top}\MM + \NN^{\top} Diag(\rr) \NN)^{-1}$, we can construct $(\MM^{\top}\MM + \NN^{\top} Diag(\tilde{\rr})\NN)^{-1}$
in $O(k^{\omega - 2}  n^2)$ time.
\end{lemma}

\begin{proof}
Let $S$ denote the entries that differ in $\rr$ and $\tilde{\rr}$.
Then we have
\[
\MM^{\top}\MM + \NN^{\top} Diag(\tilde{\rr}) \NN =
\MM^{\top}\MM + \NN^{\top} Diag(\rr) \NN + \NN_{:,S}^{\top}
\left( Diag(\tilde{\rr}_{S}) - Diag(\rr_{S}) \right)
\NN_{S, :}.
\]
This is a low rank perturbation, so by Woodbury matrix identity we get:
\[
\left( \MM^{\top}\MM + \NN^{\top} Diag(\tilde{\rr}) \NN
 \right)^{-1} = \widehat{\ZZ} - \widehat{\ZZ} \NN_{:,S}^{\top} \left( \left( Diag({\tilde{\rr}_{S}}) - Diag(\rr_{S}) \right)^{-1}+ \NN_{S, :} \widehat{\ZZ} \NN_{:, S}^{\top} \right)^{-1}\NN_{S, :} \widehat{\ZZ},
\]
where we use $\widehat{\ZZ}^{\top} = \widehat{\ZZ}$ because
$\MM^{\top}\MM + \NN^{\top} Diag(\rr) \NN$ is a symmetric matrix. To explicitly compute this matrix, we need to:
\begin{enumerate}
\item compute the matrix
$\NN_{S,:} \widehat{\ZZ}$,
\item compute $\NN_{:, S} \widehat{\ZZ} \NN_{:, S}^{\top}$
\item invert the middle term.
\end{enumerate}
This cost is dominated by the first term, which can be viewed
as multiplying $\lceil n / k \rceil$ pairs of $k \times n$
and $n \times k$ matrices.
Each such multiplication takes time $k^{\omega - 1}n$,
for a total cost of $O(k^{\omega - 2} n^2)$.
The other terms all involve matrices with dimension at most $k \times n$,
and are thus lower order terms.
\end{proof}

\subsubsection{Approximation and Fast Linear Systems Solver}
We now define the notion of approximation we use and how to solve linear systems fast given a good preconditioner.

\begin{definition}
\label{def:Approx}
We use $a \approx_{c} b$ for positive numbers $a$ and $b$
iff $c^{-1} a \leq b \leq c \cdot b$, and for vectors
and for vectors $\aa$ and $\bb$ we use $\aa \approx_{c} \bb$
to denote $\aa_{i} \approx_{c} \bb_{i}$ entry-wise.
\end{definition}

In our algorithm, we only update $k$ resistances that have increased by a constant factor. We can therefore use a constant factor preconditioner to solve the new linear system. We will use the following result on solving preconditioned systems of linear equations.

\begin{lemma}
\label{lem:LOLWhatError}
If $\rr$ and $\tilde{\rr}$ are vectors such that
$\rr \approx_{\Otil(1)} \tilde{\rr}$, and we're given the matrix $\widehat{\ZZ}^{-1} = (\MM^{\top}\MM + \NN^{\top} Diag(\rr) \NN)^{-1}$ explicitly, then we can solve a system of linear equations involving $\ZZ = \MM^{\top}\MM + \NN^{\top} Diag(\tilde{\rr}) \NN$ to $1/poly(n)$ accuracy in $\Otil(n^2)$ time.
\end{lemma}
\begin{proof}
Suppose we want to solve the system,
\[
\ZZ\xx = \bb.
\]
We know $\widehat{\ZZ}^{-1}$ and that for some constant $c$, $\frac{1}{c}\II \preceq \widehat{\ZZ}^{-1/2}\ZZ\widehat{\ZZ}^{-1/2} \preceq c \II$. The following iterative method (which is essentially gradient descent),
\[
\xx^{(k+1)}\rightarrow \xx^{(k)} - \hat{\ZZ}^{-1}\left(\ZZ\xx-\bb\right)
\]
 converges to an $\epsilon$-approximate solution in $O\left(c \log \frac{1}{\epsilon}\right)$ iterations. Each iteration can be computed via matrix-vector products. Since matrix vector products for $n \times n$ matrices require at most $O(n^2)$ we get the above lemma for $\epsilon = 1/poly(n)$.
\end{proof}

\subsubsection{Algorithm}

The algorithm is the same as that in Section 6 of \textcite{AdilKPS19}. The algorithm has two parts, an initialization routine $\textsc{InverseInit}$ which is called only at the first iteration, and the inverse maintenance procedure, {\textsc{UpdateInverse}} which is called from Algorithm~\ref{alg:oracle},
\textsc{Oracle}. Algorithm \textsc{Oracle} is called every time the resistances are updated in Algorithm~\ref{alg:FasterOracleAlgorithm}. For this section, we will assume access to all variables from these routines, and maintain the following global variables:
\begin{enumerate}
\item $\rrhat$: resistances from the last time
we updated each entry.
\item $counter(\eta)_{e}$: for each entry, track the number of times that it
changed (relative to $\rrhat$) by a factor of about $2^{-\eta}$
since the previous update.
\item $\widehat{\ZZ}$, an inverse of the matrix given by $\MM^{\top}\MM + \NN^{\top} Diag(\rrhat) \NN$.
\end{enumerate}

\begin{algorithm}[H]
\caption{Inverse Maintenance Initialization}
\label{alg:InverseInit}
 \begin{algorithmic}[1]
\Procedure{InverseInit}{$\MM,\NN,\rr^{(0)}$}
\State Set $\rrhat \leftarrow \rr^{(0)}$.
\State Set $counter(\eta)_e \leftarrow 0$ for all $0 \leq \eta \leq \log(m)$ and
$e$. 
\State Set $\widehat{\ZZ} \leftarrow (\MM^{\top}\MM + \NN^{\top} Diag(\rr) \NN)^{-1}$ by explicitly
inverting the matrix. 
\EndProcedure
\end{algorithmic}
\end{algorithm}

\begin{algorithm}[H]
\caption{Inverse Maintenance Procedure}
\label{alg:InverseMaintenance}
 \begin{algorithmic}[1]
\Procedure{UpdateInverse}{}
\For{all entries $e$}
\State Find the least non-negative integer $\eta$ such that
\[
\frac{1}{2^{\eta}} 
\leq
\frac{\rr^{\left(i\right)}_{e} - \rr^{\left(i - 1\right)}_{e}}{\rrhat_{e}}.
\]
\State Increment $counter(\eta)_{e}$.
\EndFor
\State $E_{changed}
\leftarrow
\union_{\eta: i \pmod {2^{\eta}} \equiv 0} \{ e: counter(\eta)_e \geq 2^{\eta}\}$ 
\label{lem:EmptyBucket}
\State $\tilde{\rr} \leftarrow \rrhat$
\For{all $e \in E_{changed}$}
\State $\tilde{\rr}_{e} \leftarrow \rr_{e} ^{\left( i \right)}$. 
\State Set $counter(\eta)_{e} \leftarrow 0$ for all $\eta$. 
\EndFor
\State $\widehat{\ZZ} \leftarrow \textsc{LowRankUpdate}( \widehat{\ZZ}, \rrhat, \tilde{\rr})$.
\State $\rrhat \leftarrow \tilde{\rr}$.
\EndProcedure 
 \end{algorithmic}
\end{algorithm}

\subsubsection{Analysis}
We first verify that the maintained inverse is always a 
good preconditioner to the actual matrix, $\MM^{\top}\MM + \NN^{\top} Diag(\rr^{(i)}) \NN$.

\begin{lemma}[Lemma 6.5, \textcite{AdilKPS19}]
\label{lem:GoodApprox}
After each call to {\em UpdateInverse}, the vector $\rrhat$
satisfies
\[
\rrhat \approx_{\Otil \left( 1 \right)} \rr^{\left( i \right)}.
\]
\end{lemma}

\begin{proof}
  First, observe that any change in resistance exceeding $1$ is
  reflected immediately. 
  Otherwise,
  every time we update $counter(j)_e$, $\rr_e$ can only increase
  additively by at most
\[
2^{-j + 1} \rrhat_e.
\]
Once $counter(j)_e$ exceeds $2^{j}$, $e$ will be added to 
$E_{changed}$ after at most $2^{j}$ steps. So when we start from $\rrhat_e$,  $e$ is added to 
$E_{changed}$ after  $counter(j)_e \leq 
2^j + 2^j = 2^{j+1}$ iterations. The maximum possible increase in resistance due to the bucket $j$ is,
\[
2^{-j + 1} \rrhat_{e}
\cdot
2^{j + 1}
= 4 \rrhat_{e}.
\] 
Since there are only at most $m^{1/3}$ iterations,
the contributions of buckets with $j > \log{m}$ 
are negligible. Now the change in resistance is influenced by all buckets $j$, 
each contributing at most $4\rrhat_{e}$ increase. The total change is at most $4 \rrhat_{e}  \log m$ since there are at most $\log m$ buckets.
We therefore have
\[
\rrhat_e
\leq
\rr^{\left(i\right)}_e
\leq
5 \rrhat_e \log m .
\]
for every $i$.
\end{proof}

It remains to bound the number and sizes of calls
made to Lemma~\ref{lem:LowRankUpdate}.
For this we define variables $k\left(\eta\right)^{\left(i \right)}$ to denote the number of edges added to $E_{changed}$
at iteration $i$ due to the value of $counter(\eta)_e$.
Note that $k(\eta)^{(i)}$ is non-zero only if
$i \equiv 0 \pmod{2^{\eta}}$, and
\[
\abs{E_{changed}^{\left(i\right)}}
\leq
\sum_{\eta} k\left(\eta\right)^{\left(i \right)}.
\]

We divide our analysis into 2 cases, when the relative change in resistance is at least $1$ and  when the relative change in resistance is at most $1$.  To begin with, let us first look at the following lemma that relates the change in weights to the relative change in resistance. 
\begin{restatable}{lemma}{ResistanceToFlow}
\label{lem:ResistanceChangeToFlowValue}
Consider a primal step from Algorithm~\ref{alg:FasterOracleAlgorithm}. We have
\[
\frac{\rr_e^{\left(i + 1\right)}
  - \rr_e^{\left(i\right)}}{\rr_e^{\left(i \right)}}
\leq  \left( 1 + \alpha \frac{\abs{\NN\Delta}_e}{\zeta^{1/p}} \right)^{p - 2} -1
\]
where $\Delta$ is the solution produced by the oracle Algorithm \ref{alg:oracle}.
\end{restatable}
\begin{proof}
Recall from Algorithm~\ref{alg:oracle} that
\[
\rr_e^{\left(i \right)}
=\left( \ww_e^{\left(i\right)}\right)^{p - 2}.
\]
For a primal step of Algorithm~\ref{alg:FasterOracleAlgorithm},
we have
\[
\ww^{\left( i + 1 \right)}_{e} - \ww^{\left( i \right)}_{e}
=
\frac{\alpha}{\zeta^{1/p}} \abs{\NN\Delta}_e.
\]
Substituting this in gives
\[
\frac{\rr_e^{\left(i + 1\right)} - \rr_e^{\left(i\right)}}{\rr_e^{\left(i \right)}}
=
\frac{\left( \ww_e^{\left(i\right)} + \frac{\alpha}{\zeta^{1/p}} \abs{\NN\Delta}_e\right)^{p - 2}
- \left( \ww_e^{\left(i\right)}\right)^{p - 2}}
{ \left( \ww_e^{\left(i\right)}\right)^{p - 2}}
\leq
\left( 1 + \frac{\frac{\alpha}{\zeta^{1/p}} \abs{\NN\Delta}_e}{\ww_e^{\left( i \right)}} \right)^{p - 2}
-1
\leq
\left( 1 + \frac{\alpha}{\zeta^{1/p}} \abs{\NN\Delta}_e \right)^{p - 2}
-1,
\]
where the last inequality utilizes $\ww_e^{(i)} \geq 1$.
\end{proof}

We now consider the case when the relative change in resistance is at least $1$.

\begin{lemma}
\label{lem:CountHighWidth}
Throughout the course of a run of Algorithm~\ref{alg:FasterOracleAlgorithm},
the number of edges added to $E_{changed}$ due to  relative resistance increase of at least $1$, 
\[
\sum_{1 \leq i \leq T}
k\left( 0 \right)^{\left( i \right)}
\leq
O\left( m^{\frac{p + 2}{3p - 2}} \right).
\]
\end{lemma}

\begin{proof}
From Lemma \ref{lem:ckmst:res-increase}, we know that the change in energy over one iteration is at least,
\[
\sum_e \left(\NN\Delta \right)_e^2 \left(1-\frac{\rr^{(i)}_e}{\rr^{(i+1)}_e}\right).
\]
Over all iterations, the change in energy is at least, 
\[
\sum_i \sum_e \left(\NN\Delta \right)_e^2 \left(1-\frac{\rr^{(i)}_e}{\rr^{(i+1)}_e}\right) 
\]
 which is upper bounded by $O(m^{\frac{p-2}{p}})\zeta^{2/p}$. When iteration $i$ is a width reduction step, the relative resistance change is always at least $1$. In this case $\abs{\NN\Delta} \geq \rho \zeta^{1/p}$. When we have a primal step, Lemma \ref{lem:ResistanceChangeToFlowValue} implies that when the relative change in resistance is at least $1$ then,
\[
\abs{\NN\Delta}_e \geq \Omega(1)\alpha^{-1}\zeta^{1/p}.
\]
Using the bound $\abs{\NN\Delta}_e \geq \Omega(p^{-1})\alpha^{-1}\zeta^{1/p}$ is sufficient since $\rho >\Omega(p^{-1}\alpha^{-1})$ and both kinds of iterations are accounted for. The total change in energy can now be bounded.
\begin{align*}
& p^{-2}\alpha^{-2} \zeta^{2/p} \sum_{i}  \sum_e \mathbb{1}_{\left[\frac{\rr^{(i+1)}_e - \rr^{(i)}_e}{\rr^{(i)}_e} \ge 1 \right]} \le O(m^{\frac{p-2}{p}})\zeta^{2/p}\\
& \Leftrightarrow  p^{-2}\alpha^{-2} \sum_{i}  k(0)^{(i)}  \le O(m^{\frac{p-2}{p}})\\
& \Leftrightarrow  \sum_{i} k \left( 0 \right)^{\left( i \right)} \le O(p^2 m^{(p-2)/p} \alpha^{2}).
\end{align*}
The Lemma follows by substituting $\alpha=  \Theta\left(p^{-1}m^{-\frac{p^2-5p+2}{p(3p-2)}}\right)$ in the above equation.
\end{proof}

\begin{lemma}
\label{lem:CountLowWidth}
Throughout the course of a run of Algorithm~\ref{alg:FasterOracleAlgorithm}, the number of edges added to $E_{changed}$ due to  relative resistance increase between $2^{-\eta}$ and $2^{-\eta+1}$, 
\[
\sum_{1 \leq i \leq T}
k\left( \eta \right)^{\left( i \right)}
\leq
\begin{cases}
0 & \text{if $2^{\eta} \geq T$},\\
O\left(m^{\frac{p + 2}{3p - 2}} 2^{2\eta}\right) & \text{otherwise}.
\end{cases}
\]
\end{lemma}
\begin{proof}
From Lemma \ref{lem:ckmst:res-increase}, the total change in energy is at least, 
\[
\sum_i \sum_e \left(\NN\Delta \right)_e^2 \left(1-\frac{\rr^{(i)}_e}{\rr^{(i+1)}_e}\right) .
\]
We know that $\frac{\rr^{(i+1)}_e - \rr^{(i)}_e}{\rr^{(i)}_e} \geq 2^{-\eta}$. Using Lemma \ref{lem:ResistanceChangeToFlowValue}, we have,
\[
\left(1+\alpha \frac{\abs{\NN\Delta}_e}{\zeta^{1/p}}\right)^{p-2} - 1 \geq 2^{-\eta}. 
\]
We thus obtain,
\[
\left(1+\alpha \frac{\abs{\NN\Delta}_e}{\zeta^{1/p}}\right)^{p-2} - 1\leq 
\begin{cases}
\alpha \frac{\abs{\NN\Delta}_e}{\zeta^{1/p}} & \text{ when $\alpha \abs{\Delta_e} \leq \zeta^{1/p}$ or $p-2 \leq 1$} \\
 \left(2 \alpha \frac{\abs{\NN\Delta}_e}{\zeta^{1/p}}\right)^{p-2} & \text{ otherwise. }
\end{cases}
\]
Now, in the second case, when $\alpha \abs{\NN\Delta_e} \geq \zeta^{1/p}$ and $p-2>1$, 
\[
\left(2 \alpha \frac{\abs{\NN\Delta}_e}{\zeta^{1/p}}\right)^{p-2} \geq 2^{-\eta} \Rightarrow \alpha \abs{\NN\Delta}_e \geq \left(\frac{1}{2^{\eta}}\right)^{1/(p-2)+1}\zeta^{1/p} \geq  2^{-\eta-1}\zeta^{1/p}
\]
Therefore, for both cases we have, 
\[
\alpha \abs{\NN\Delta}_e \geq  \left(2^{-\eta-1}\right)\zeta^{1/p}.
\]
Using the above bound and the fact that the total change in energy is at most $O(m^{\frac{p-2}{p}})\zeta^{2/p}$, gives,
\begin{align*}
&\sum_i \sum_e \left(\NN\Delta \right)_e^2 \left(1-\frac{\rr^{(i)}_e}{\rr^{(i+1)}_e}\right) \leq O(m^{\frac{p-2}{p}})\zeta^{2/p}\\
\Rightarrow & \frac{1}{4}\sum_i \sum_e \left(\alpha^{-1} 2^{-\eta}\zeta^{1/p}\right)^2 \cdot \left(2^{-\eta}\mathbb{1}_{2^{-\eta+1} \geq \frac{\rr^{(i+1)}_e - \rr^{(i)}_e}{\rr^{(i)}_e} \geq 2^{-\eta}}\right)  \leq O(m^{\frac{p-2}{p}})\zeta^{2/p}\\
\Rightarrow &  \alpha^{-2} 2^{-3\eta} \sum_i  2^{\eta}  k\left( \eta \right)^{\left( i \right)}\leq O(m^{\frac{p-2}{p}})\\
\Rightarrow & \sum_i k\left( \eta \right)^{\left( i \right)} \leq O \left( m^{(p-2)/p} \alpha^{2} 2^{2\eta}  \right)
\end{align*}
The Lemma follows substituting $\alpha=  \Theta\left(p^{-1}m^{-\frac{p^2-5p+2}{p(3p-2)}}\right)$ in the above equation.
\end{proof}
We can now use the concavity of $f(z) = z^{\omega - 2}$ to upper bound
the contribution of these terms.
\begin{corollary}
\label{lem:TotalCostWidth}
Let $k(\eta)^{(i)}$ be as defined. Over all iterations we have,
\[
\sum_{i} \left( k\left( 0 \right)^{\left( i \right)} \right) ^{\omega - 2}
\leq
O \left(p^{3-\omega} m^{\frac{p - \left(10 - 4 \omega\right) }{3p - 2}} \right)
\]
and for every $\eta$,
\[
\sum_i^T \left( k\left( \eta \right)^{\left( i \right)} \right) ^{\omega - 2} \leq 
\begin{cases}
0 & \text{if $2^{\eta} \geq T$},\\
O \left(p^{3-\omega} m^{\frac{p - 2 + 4 \left( \omega - 2 \right)}{3p - 2}} \cdot 2^{\eta \left( 3 \omega - 7 \right) } \right) & \text{otherwise}.
\end{cases}
\]
\end{corollary}

\begin{proof}
Due to the concavity of the $\omega - 2 \approx 0.3727 < 1$ power,
this total is maximized when it's equally distributed over all iterations. In the first sum, the number of terms is equal to the number of iterations, i.e.,  
$O(p m^{\frac{p - 2}{3p - 2}})$. In the second sum the number of terms is $O(p m^{\frac{p - 2}{3p - 2}}) 2^{-\eta}$.
Distributing the sum equally over the above numbers give, 
\[
\sum_i^T \left( k\left( 0 \right)^{\left( i \right)} \right) ^{\omega - 2} \leq \left( O \left( p^{-1} m^{\frac{p+2}{3p-2} - \frac{p-2}{3p - 2}} \right) \right)^{\omega - 2}
\cdot O\left( pm^{\frac{p - 2}{3p - 2}} \right)
=
O \left(p^{3-\omega} m^{\frac{p - 2 + 4 \left( \omega - 2 \right)}{3p - 2}} \right)
=
O \left(p^{3-\omega} m^{\frac{p - \left(10 - 4 \omega\right) }{3p - 2}} \right)
\]
and 
\begin{align*}
\sum_i^T \left( k\left( \eta \right)^{\left( i \right)} \right) ^{\omega - 2} &\leq O\left(p m^{\frac{p - 2}{3p - 2}} 2^{-\eta}\right) \cdot \left( p^{-1}\frac{m^{\frac{p + 2}{3p -2}} 2^{2\eta}} {m^{\frac{p - 2}{3p - 2}} 2^{-\eta}} \right)^{\omega - 2}\\
&= O\left(p^{3-\omega}m^{\frac{p - 2 + 4 \left( \omega - 2 \right) }{3p - 2}}2^{-\eta}\cdot 2^{3\eta (\omega - 2 )}\right)\\
&= O \left(p^{3-\omega}m^{\frac{p - 2 + 4 \left( \omega - 2 \right)}{3p - 2}} 2^{\eta (3 \omega - 7)}\right).
\end{align*}
\end{proof}

\subsection{Proof of Theorem \ref{thm:InverseMaint}}

\InverseMaint*
\begin{proof} By Lemma~\ref{lem:GoodApprox}, the $\rrhat$ that the inverse being
  maintained corresponds to always satisfy
  $\rrhat \approx_{\Otil(1)} \rr^{(i)}$.  So by the iterative linear
  systems solver method outlined in Lemma~\ref{lem:LOLWhatError},
  we can implement each call to $\textsc{Oracle}$
  (Algorithm~\ref{alg:oracle})in time $O(n^2)$ in addition to the
  cost of performing inverse maintenance.  This leads to a total cost
  of
\[
\Otil\left(pn^2m^{\frac{p - 2}{3p - 2}}\right).
\]
across the $T = \Theta(pm^{\frac{p - 2}{3p - 2}})$ iterations.

The costs of inverse maintenance is dominated by the calls to
the low-rank update procedure outlined in Lemma~\ref{lem:LowRankUpdate}.
Its total cost is bounded by
\[
O\left( \sum_{i} \abs{E_{changed}^{\left( i \right)}}^{\omega - 2}n^2 \right)
=
O\left( 
n^2
\sum_{i} \left(
   \sum_{\eta} k\left( \eta \right)^{\left( i \right)}
\right) ^{\omega - 2}
\right).
\]
Because there are only $O(\log{m})$ values of $\eta$,
and each $k(\eta)^{(i)}$ is non-negative, we can bound the total cost by:
\[
\Otil\left(n^2\sum_{i} \sum_{\eta}
\left(k\left( \eta \right)^{\left( i \right)}\right)^{\omega - 2}
\right)\leq\Otil\left(p^{3-\omega} n^2\sum_{\eta: 2^{\eta} \leq T} 
m^{\frac{p - 2 + 4 \left( \omega - 2 \right)}{3p - 2}}
\cdot 2^{\eta \left( 3 \omega - 7 \right) }
\right),
\]
where the inequality follows from substituting in the result
of Lemma~\ref{lem:TotalCostWidth}.
Depending on the sign of $3 \omega - 7$, this sum is dominated
either at $\eta = 0$ or $\eta = \log{T}$.
Including both terms then gives
\[
\Otil\left( p^{3-\omega}
n^2\left(m^{\frac{p - 2 + 4 \left( \omega - 2 \right)}{3p - 2}}
+m^{\frac{p - 2 + 4 \left( \omega - 2 \right)
+ \left( p - 2 \right) \left( 3 \omega - 7 \right) }{3p - 2}}\right)
\right),
\]
with the exponent on the trailing term simplifying to $\omega - 2$ to give,
\[
\Otil\left( p^{3-\omega}n^2
\left( m ^{\frac{p - \left(10 - 4 \omega \right)}{3p - 2}}
+m^{\omega-2}\right)\right).
\]
\end{proof}


%% file: Chapters/IRLS.tex

\section{Iteratively Reweighted Least Squares Algorithm}
\label{sec:IRLS}

Iteratively Reweighted Least Squares (IRLS) Algorithms are a family of algorithms for solving $\ell_p$-regression. These algorithms have been studied extensively for about 60 years \cite{lawson1961,rice1964,gorodnitsky1997} and the classical form solves the following version of $\ell_p$-regression,
\begin{equation}\label{eq:IRLS-true}
\min_{\xx}\|\AA\xx-\bb\|_p,
\end{equation}
where $\AA$ is a tall thin matrix and $\bb$ is a vector. The main idea in IRLS algorithms is to solve a weighted least squares problem in every iteration to obtain the next iterate,
\begin{equation}\label{eq:UpdateIRLS}
\xx^{(t+1)} = \arg\min_{\xx} (\AA\xx^{(t)}-\bb)^{\top}\RR(\AA\xx^{(t)}-\bb)
\end{equation}
starting from $\xx^{(0)}$ which is usually $\arg\min_{\xx} \|\AA\xx-\bb\|_2^2$. Here $\RR$ is picked to be $Diag(|\AA\xx^{(t)}-\bb|^{p-2})$ and note that the above equation now becomes a fixed point iterate for the $\ell_p$-regression problem. It is known that the fixed point is unique for $p \in (1,\infty)$.

The basic version of the above IRLS algorithm is guaranteed to converge for $p\in (1.5,3)$, however, even for small $p\approx 3.5$, the algorithm diverges \cite{rios2019}. Over the years there have been several studies on IRLS algorithms and attempts  to show convergence \cite{karlovitz1970,osborne1985}, but none of them show quantitative bounds or require starting solutions close enough to the optimum. Refer to \textcite{burrus2012} for a complete survey on these methods.

In this section we propose an IRLS algorithm and prove that our algorithm is guaranteed to converge geometrically to the optimum. Our algorithm is based on the algorithm of \textcite{AdilPS19} and present some experimental results from experiments performed in the paper that demonstrate our algorithm works very well in practice. We provide a much simpler analysis and integrate the analysis with the framework we have built so far.

We will focus on the following pure $\ell_p$ setting for better readability,
\[
\min_{\AA\xx=\bb} \|\NN\xx\|_p.
\]
We note that our algorithm also works for the setting described in Equation \eqref{eq:IRLS-true}. We will first describe our algorithm in the next section, and then present some experimental results from experiments that were performed in \textcite{AdilPS19}.

\subsection{IRLS Algorithm}
Our IRLS algorithm is based on our overall iterative refinement framework (Algorithm \ref{alg:IterRef}) where we will directly use a weighted least squares problem to solve the residual problem. Consider Algorithm \ref{alg:IRLS} and compare it with Algorithm \ref{alg:IterRef}. We note that it is same overall, except now we have an extra step {\sc LineSearch} and we update the solution (Line \ref{alg:line:updateSol}) at every iteration. These steps do not affect the overall convergence guarantees of the iterative refinement framework in Algorithm \ref{alg:IterRef}, since these are only ensuring that given a solution from {\sc ResidualSolver-IRLS}, we are taking a step that reduces the objective value the most as opposed the fixed update defined in Algorithm \ref{alg:IterRef}. In other words, we are reducing the objective value in each iteration at least as much as in Algorithm \ref{alg:IterRef}. We thus require to prove the guarantees of {\sc ResidualSolver-IRLS} (Algorithm \ref{alg:Res-IRLS}) and combine it with Theorem \ref{thm:IterativeRefinement} to obtain our final convergence guarantees. We will prove the following result on our IRLS algorithm (Algorithm \ref{alg:IRLS}).

\IRLSMain*
The key connection with IRLS algorithms is that we are able to show that it is sufficient to solve a weighted least squares problem to solve the residual problem. The two main differences are, in every iteration we add a small systematic {\it padding} to $\RR$ and, we perform a line search. These tricks are common empirical modifications used to avoid ill conditioning of matrices and for a faster convergence \cite{karlovitz1970,vargas1999}.
\begin{algorithm}
\caption{Iteratively Reweighted Least Squares}
\label{alg:IRLS}
 \begin{algorithmic}[1]
 \Procedure{\textsc{IRLS}}{$\AA, \NN,\bb,p,\epsilon$}
\State $\xx \leftarrow {\arg\min}_{\AA\xx=\bb}\|\NN\xx\|_2^2$
\State $\nu \leftarrow \|\NN\xx\|_p^p$
\While{$\nu >\frac{\epsilon}{2} \|\NN\xx\|_p^p$}
\State $\Dtil, \kappa\leftarrow$ {\sc ResidualSolver-IRLS($\xx,\NN,\AA,\bb,\nu,p$)}
\State $\alpha \leftarrow$ {\sc LineSearch($\NN,\xx,\Dtil$)} \Comment{$\alpha = \arg\min_{\beta}\|\NN(\xx-\beta\Dtil)\|_p^p$}
\State $\xx \leftarrow \xx - \alpha\frac{\Dtil}{p}$ \label{alg:line:updateSol}
\If{$\res_p(\alpha \Dtil) < \frac{\nu}{32p\kappa}$}
\State $\nu \leftarrow \frac{\nu}{2}$
\EndIf
\EndWhile
\State \Return $\xx$
 \EndProcedure 
 \end{algorithmic}
\end{algorithm}
\begin{algorithm}
\caption{Residual Solver for IRLS}
\label{alg:Res-IRLS}
 \begin{algorithmic}[1]
 \Procedure{\textsc{ResidualSolver-IRLS}}{$\xx,\NN,\AA,\bb,\nu,p$}
\State $\gg \leftarrow Diag(|\NN\xx|^{p-2})\NN\xx$
\State $\RR \leftarrow 2Diag(|\NN\xx|^{p-2})$
\State $\ss \leftarrow \nu^{\frac{p-2}{p}}m^{-\frac{p-2}{p}}$
\State $\Dtil \leftarrow {\arg\max}_{\AA\Delta = 0}\gg^{\top}\NN\Delta - \Delta^{\top}\NN^{\top}(R + \ss\II)\NN\Delta$\Comment{Problem \eqref{eq:l2-IRLS}}
\State $k \leftarrow \frac{\|\NN\Dtil\|_p^p}{\Dtil^{\top}\NN^{\top}(\RR+s\II)\NN\Dtil}$
\State $\alpha_0 \leftarrow \min\left\{\frac{1}{2},\frac{1}{2k^{1/(p-1)}}\right\}$
\State \Return $\Dtil, 2^{13}p^2\alpha_0^{-1}$
 \EndProcedure 
 \end{algorithmic}
\end{algorithm}

We will prove the following result about solving the residual problem.
\begin{lemma}\label{lem:residualIRLS}
Let $\xx$ be the current iterate and $\nu$ be such that $\|\NN\xx\|_p^p - OPT \in(\nu/2, \nu]$. Let $\Dtil$ be the solution of \eqref{eq:l2-IRLS}. Then for $\alpha_0$ and $\alpha$ as defined in Algorithm \ref{alg:Res-IRLS} and Algorithm \ref{alg:IRLS} respectively, $\alpha\Dtil$ is a $O\left(p^2 \alpha_0^{-1}\right) = O\left(p^2m^{\frac{p-2}{2(p-1)}}\right)$-approximate solution to the residual problem.
\end{lemma}
We note that Theorem \ref{thm:IRLS-Main} directly follows from Lemma \ref{lem:residualIRLS}, Lemma \ref{lem:invariant} and Theorem \ref{thm:IterativeRefinement}. Therefore, in the next section, we will prove Lemma \ref{lem:residualIRLS}.

\subsubsection{Solving the Residual Problem}

 Recall the residual problem (Definition \ref{def:residual}),
\[
\max_{\AA\Delta = 0} \gg^{\top}\NN\Delta - \Delta^{\top}\NN^{\top}\RR\NN\Delta - \|\NN\Delta\|_p^p,
\]
with $\gg = Diag(|\NN\xx|^{p-2})\NN\xx$ and $\RR = 2Diag(|\NN\xx|^{p-2})$. Let $\nu$ be as in Algorithm \ref{alg:IterRef}, then we will show that the solution of the following weighted least squares problem is a good approximation to the residual problem,
\begin{equation}\label{eq:l2-IRLS}
\max_{\AA\Delta = 0}\gg^{\top}\NN\Delta - \Delta^{\top}\NN^{\top}(\RR + \nu^{\frac{p-2}{p}}m^{-\frac{p-2}{p}}\II)\NN\Delta.
\end{equation}

\subsubsection{Proof of Lemma \ref{lem:residualIRLS}}
\begin{proof}
Since $\|\NN\xx\|_p^p - OPT \in(\nu/2, \nu]$, from Lemma \ref{lem:RelateResidualOpt}, we have the optimum of the residual problem satisfies, $res_p(\Dopt) \in (\nu/32p,\nu]$. 
We will next prove, that the objective of \eqref{eq:l2-IRLS} at the optimum is at most $\nu$ and at least $\frac{\nu}{2^{13} p^2}$.
Before proving the above bound, we will prove how $\alpha\Dtil$ gives the required approximation to the residual problem. We have $\alpha_0 = \min\left\{\frac{1}{2}, \frac{1}{(2k)^{1/p-1}}\right\}$.
\begin{align*}
res_p(\alpha\Dtil)& \geq 16 p \cdot  res_p(\alpha \Dtil/16p)\\
& \geq \|\NN\xx\|_p^p - \|\NN(\xx- \alpha\Dtil)\|_p^p\\
& \geq \|\NN\xx\|_p^p - \|\NN(\xx- \alpha_0\Dtil)\|_p^p\\
& \geq res_p(\alpha_0 \Dtil)\\
& = \alpha_0\left(\gg^{\top}\NN\Dtil - \alpha_0 \Dtil^{\top}\NN^{\top}\RR\NN\Dtil - \alpha_0^{p-1} \|\NN\Dtil\|_p^p \right)\\
& \geq \alpha_0\left(\gg^{\top}\NN\Dtil - \alpha_0 \Dtil^{\top}\NN^{\top}(\RR+s\II)\NN\Dtil - \alpha_0^{p-1} k \Dtil^{\top}\NN^{\top}(\RR+s\II)\NN\Dtil \right)\\
& \geq \alpha_0\left(\gg^{\top}\NN\Dtil - \frac{1}{2} \Dtil^{\top}\NN^{\top}(\RR+s\II)\NN\Dtil - \frac{1}{2} \Dtil^{\top}\NN^{\top}(\RR+s\II)\NN\Dtil \right)\\
& = \alpha_0\left(\gg^{\top}\NN\Dtil - \Dtil^{\top}\NN^{\top}(\RR+s\II)\NN\Dtil \right)\\
& \geq \frac{\alpha_0 \nu}{2^{13} p^2} \geq \frac{\alpha_0}{2^{13} p^2}OPT.
\end{align*}
It remains to prove the bound on the optimal objective of \eqref{eq:l2-IRLS} and bound $\alpha_0$ for which it is sufficient to find an upper bound on $k$,
\[
k = \frac{\|\NN\Dtil\|_p^p}{\Dtil^{\top}\NN^{\top}(\RR+s\II)\NN\Dtil}.
\]
We will first bound $k$. Since, $s\II\preceq \RR + s\II$,
\[
\|\NN\Dtil\|_2^2 \leq \frac{1}{s} \Dtil^{\top}\NN^{\top}(\RR+s\II)\NN\Dtil,
\]
and
\[
\|\NN\Dtil\|_p^p \leq \|\NN\Dtil\|_2^p \leq \frac{1}{s} \|\NN\Dtil\|_2^{p-2}\Dtil^{\top}\NN^{\top}(\RR+s\II)\NN\Dtil.
\]
Therefore it is sufficient to bound $\|\NN\Dtil\|_2$, as
\[
k = \frac{\|\NN\Dtil\|_p^p}{\Dtil^{\top}\NN^{\top}(\RR+s\II)\NN\Dtil} \leq \frac{1}{s} \|\NN\Dtil\|_2^{p-2}.
\]
To bound $\|\NN\Dtil\|_2$, we start by assuming $|\gg^{\top}\NN\Dtil| \leq \nu$. Now, since optimal objective of \eqref{eq:l2-IRLS} is lower bounded by $\frac{\nu}{2^{13} p^2}$,
\[
 \Dtil^{\top}\NN^{\top}(R + \nu^{\frac{p-2}{p}}m^{-\frac{p-2}{p}}\II)\NN\Dtil \leq \gg^{\top}\NN\Dtil - \frac{\nu}{2^{13} p^2} \leq \nu.
\]
We thus have,
\[
\nu^{\frac{p-2}{p}}m^{-\frac{p-2}{p}}\|\NN\Dtil\|_2^2 \leq \nu.
\]
Using this we get,
\[
k \leq \frac{1}{\nu^{\frac{p-2}{p}}m^{-\frac{p-2}{p}}} \frac{\nu^{\frac{p-2}{2}}}{\nu^{\frac{(p-2)^2}{2p}}m^{-\frac{(p-2)^2}{2p}}} =m^{\frac{p-2}{2}}.  
\]
We thus have $\alpha_0$ lower bounded by $m^{-\frac{p-2}{2(p-1)}}$, which gives us our result. It remains to give a lower bound to the optimal objective of \eqref{eq:l2-IRLS}. 

Let $\Dopt$ denote the optimum of the residual problem. We know that $\|\NN\Dopt\|_p^p \leq \nu, \Dopt^{\top}\NN^{\top}\RR\NN\Dopt \leq \nu$ and $\gg^{\top}\NN\Dopt > \nu/32p$. Since $\|\NN\Dopt\|_p^p \leq \nu$ we have $\|\NN\Dopt\|_2^2 \leq m^{(p-2)/p}\nu^{2/p}$. For $a = 1/2^7 p,$ $a\Dopt$ is a feasible solution for \eqref{eq:l2-IRLS}.
\begin{align*}
\gg^{\top}\NN\Dtil - \Dtil^{\top}\NN^{\top}(R + \nu^{\frac{p-2}{p}}m^{-\frac{p-2}{p}}\II)\NN\Dtil &\geq
a \gg^{\top}\NN\Dopt - a^2\Dopt^{\top}\NN^{\top}(R + \nu^{\frac{p-2}{p}}m^{-\frac{p-2}{p}}\II)\NN\Dopt\\
& \geq a \left(\frac{\nu}{32p} - a \Dopt^{\top}\NN^{\top}\RR\NN\Dopt - a \nu^{\frac{p-2}{p}}m^{-\frac{p-2}{p}} \|\NN\Dopt\|_2^2\right)\\
& \geq  a \left(\frac{\nu}{32p} - a \nu - a \nu^{\frac{p-2}{p}}m^{-\frac{p-2}{p}} m^{(p-2)/p}\nu^{2/p}\right)\\
& = a \left(\frac{\nu}{32p} - a \nu - a \nu\right)\\
& = a \frac{\nu}{2^6 p} = \frac{\nu}{2^{13} p^2}
\end{align*}
Thus, the optimal objective of \eqref{eq:l2-IRLS} is lower bounded by $\frac{\nu}{2^{13} p^2}$.
\end{proof}

\subsection{Experiments}
 In this section, we include the experimental results from \textcite{AdilPS19} which are based on Algorithm $p$-IRLS described in the paper. We would like to mention that $p$-IRLS is similar in spirit to Algorithm \ref{alg:IRLS} and thus we expect a similar performance by an implementation of Algorithm \ref{alg:IRLS}. Algorithm $p$-IRLS is described for setting \eqref{eq:IRLS-true} and is available at https://github.com/fast-algos/pIRLS \cite{AdilPS2019Code}. We now give a brief summary of the experiments.

 \subsubsection{Experiments on p-IRLS}

\begin{figure}
\subfloat[Size of $\AA$ fixed to $1000 \times 850$.]{  \includegraphics[width =0.3\textwidth]{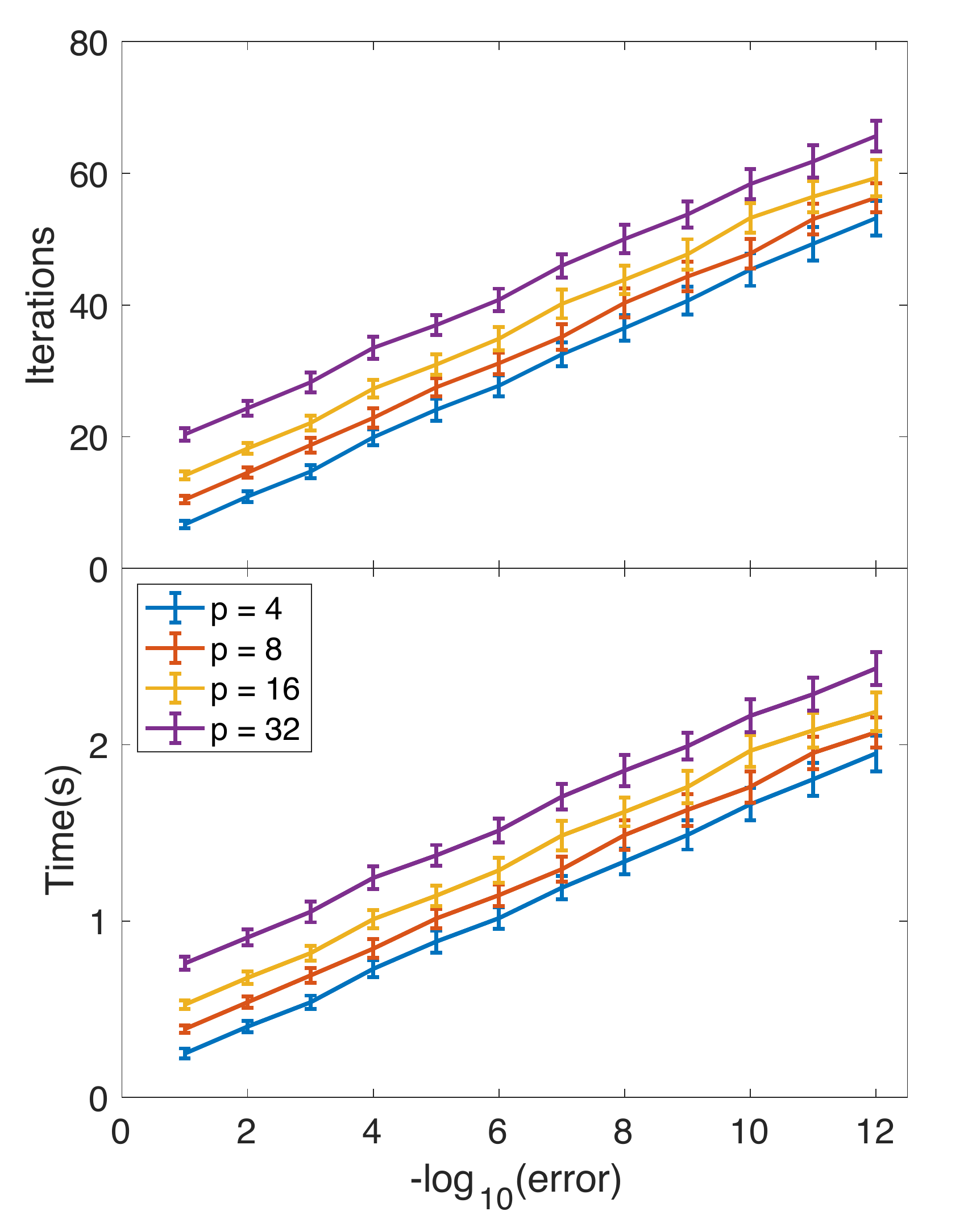}}
\hfill
\subfloat[Sizes of $\AA$: $(50+100(k-1))\times 100k$. Error $\epsilon = 10^{-8}.$]{  \includegraphics[width =0.3\textwidth]{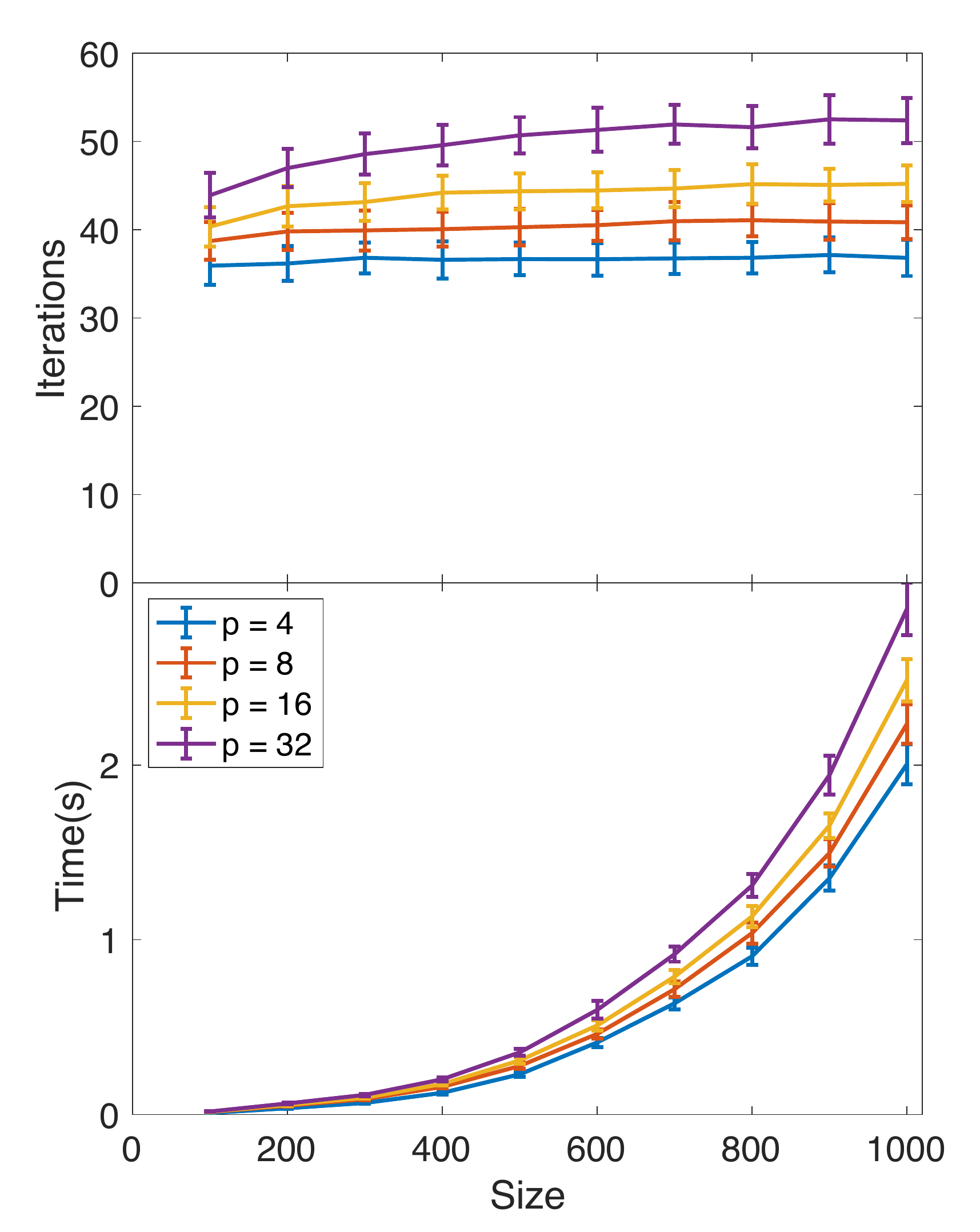}}
\hfill
 \subfloat[Size of $\AA$ is fixed to $1000 \times 850$. Error $\epsilon = 10^{-8}.$]{  \includegraphics[width =0.3\textwidth]{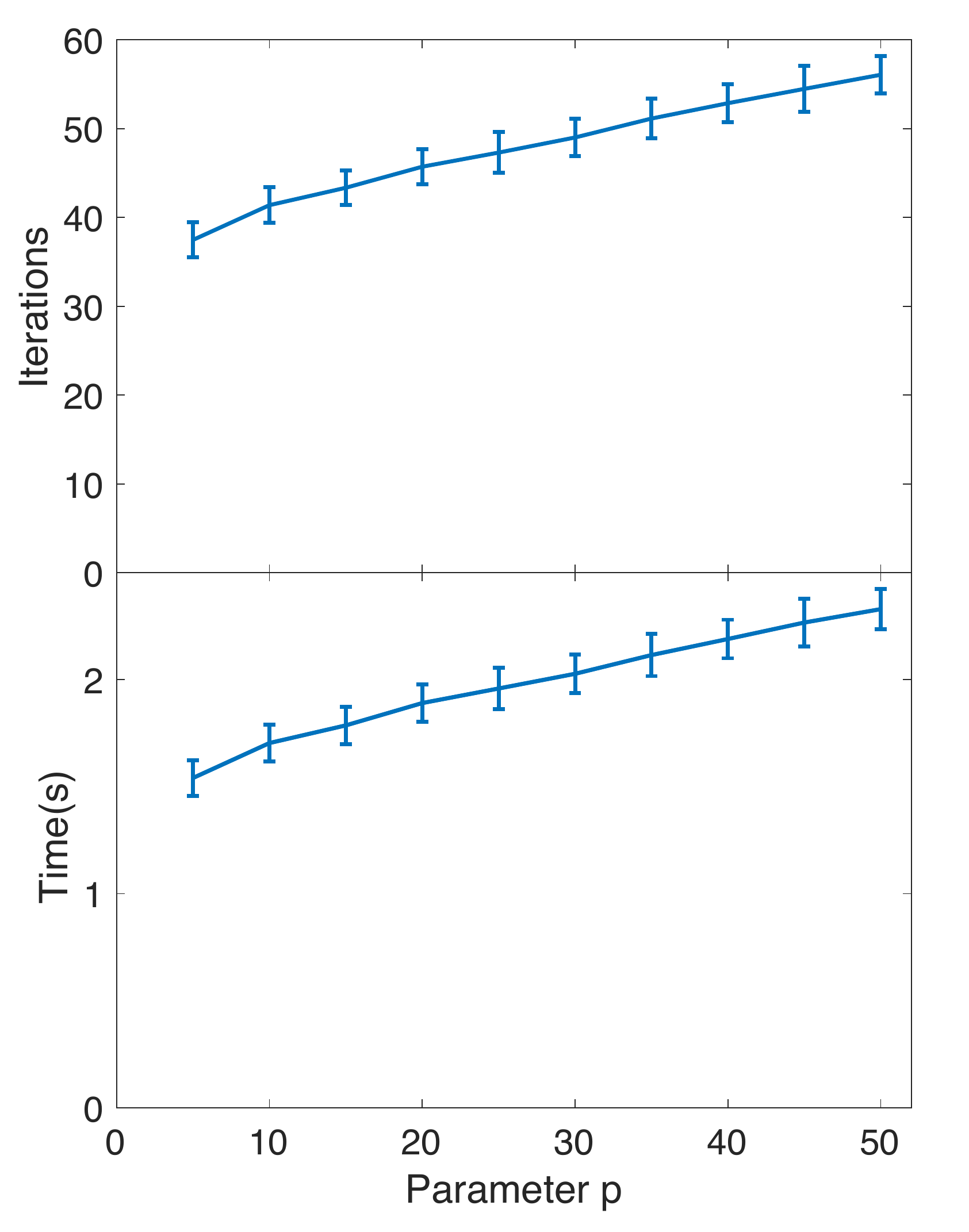}}
   \caption{Random Matrix instances. Comparing the number of iterations and time taken by our algorithm with the parameters. Averaged over 100 random samples for $\AA$ and $\bb$. Linear solver used : backslash.}
   \label{fig:Matrices}
\end{figure}

\begin{figure}
  \centering
    \includegraphics[width=0.3\textwidth]{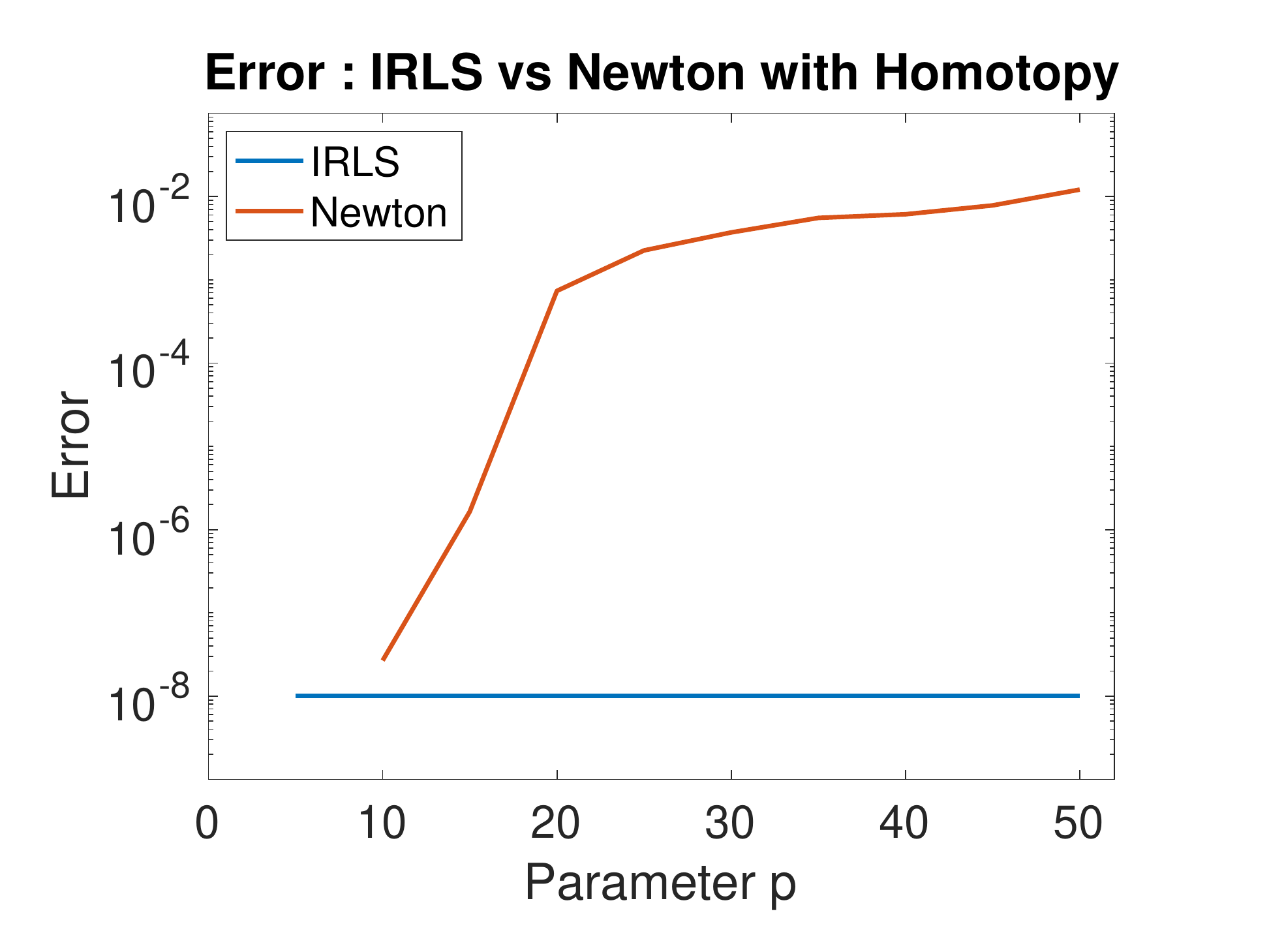}
\caption{Averaged over 100 random samples. Graph: $1000$ nodes ($5000$-$6000$ edges). Solver: PCG with Cholesky preconditioner.}
\label{fig:IRLSvsNt}
\end{figure}

\begin{figure}
\subfloat[Size of graph fixed to $1000$ nodes (around $5000$-$6000$ edges).]{  \includegraphics[width =0.3\textwidth]{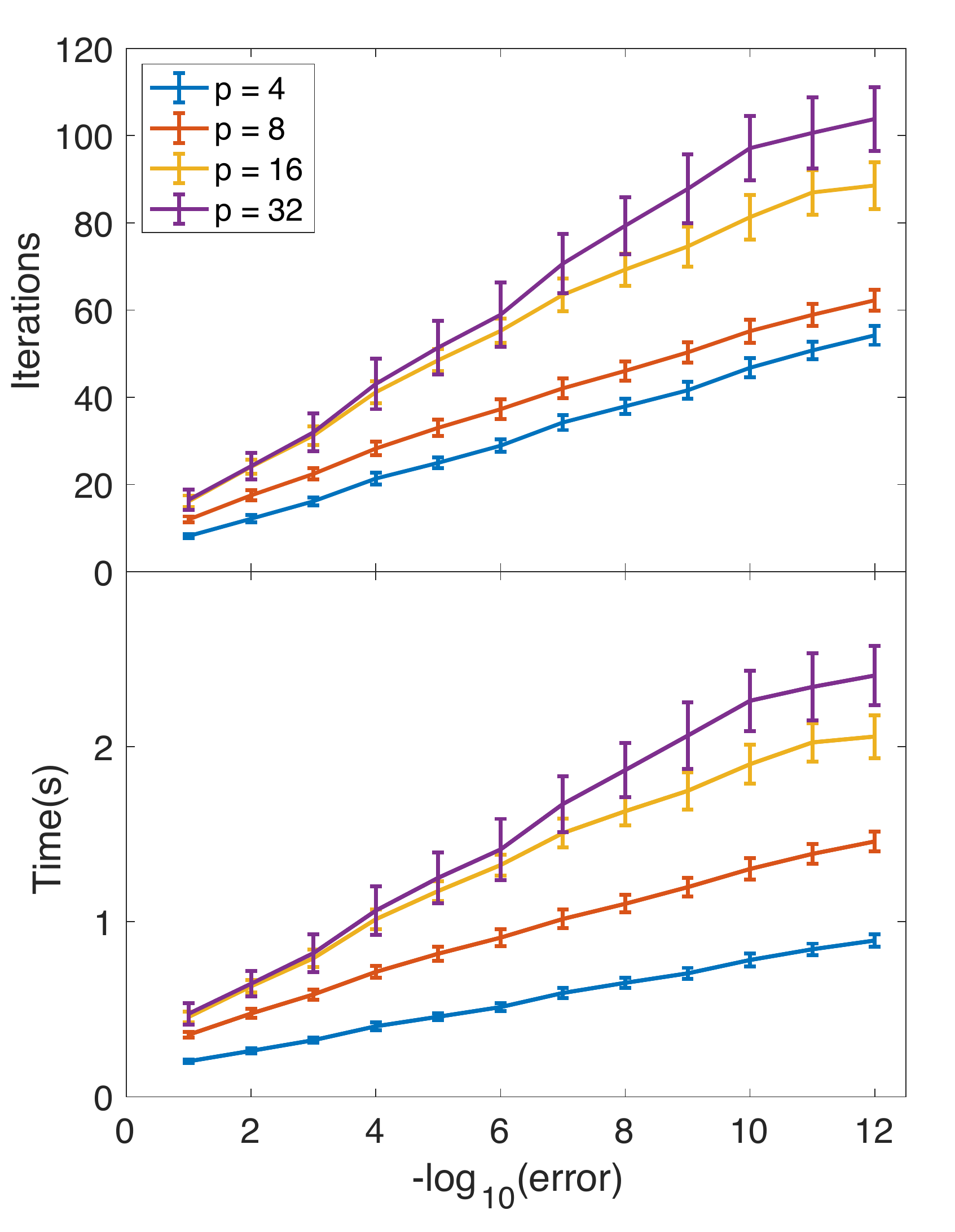}}
\hfill
\subfloat[Number of nodes: $100k$. Error $\epsilon = 10^{-8}.$]{  \includegraphics[width =0.3\textwidth]{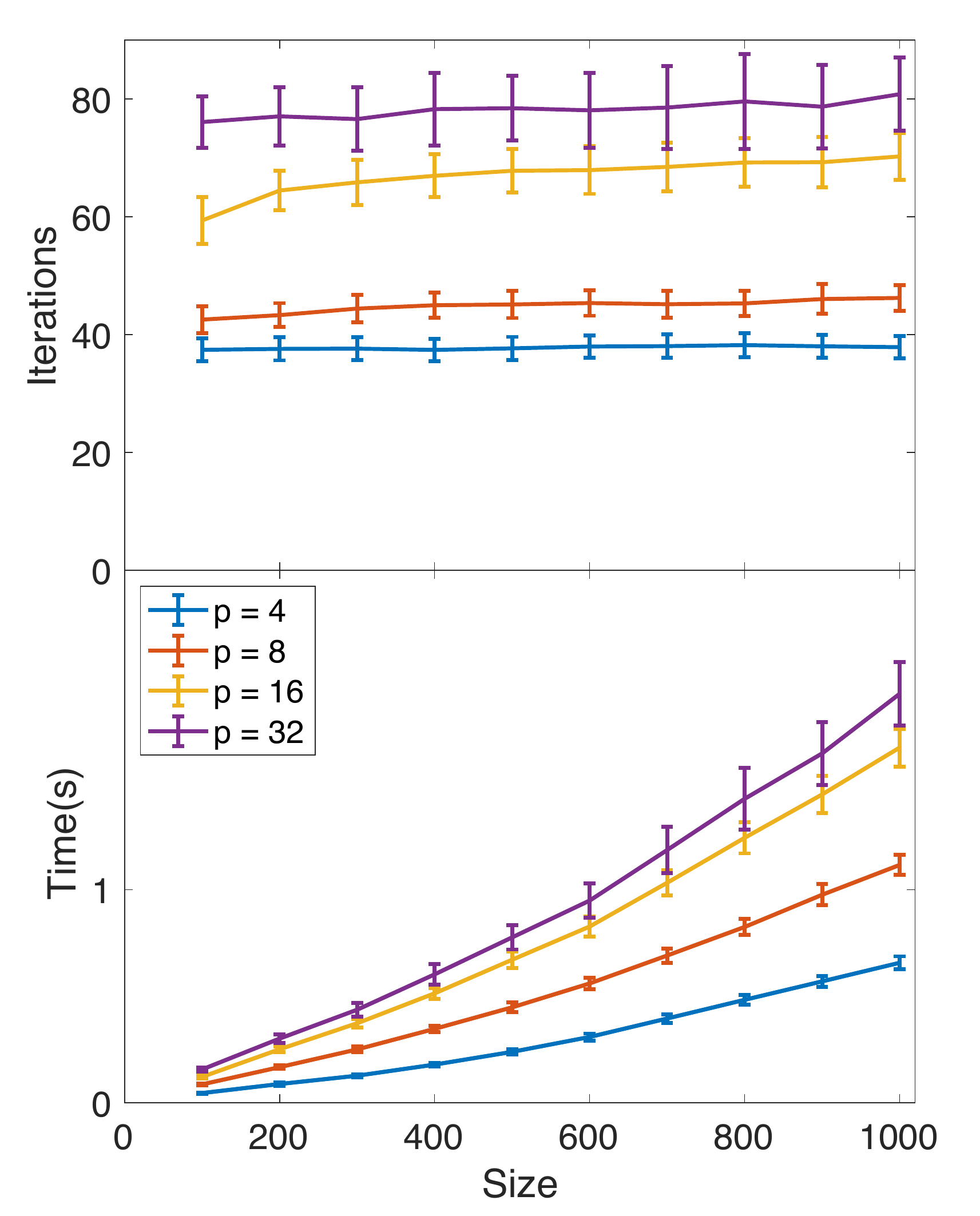}}
\hfill
 \subfloat[Size of graph fixed to $1000$ nodes (around $5000$-$6000$ edges). Error $\epsilon = 10^{-8}.$]{  \includegraphics[width =0.3\textwidth]{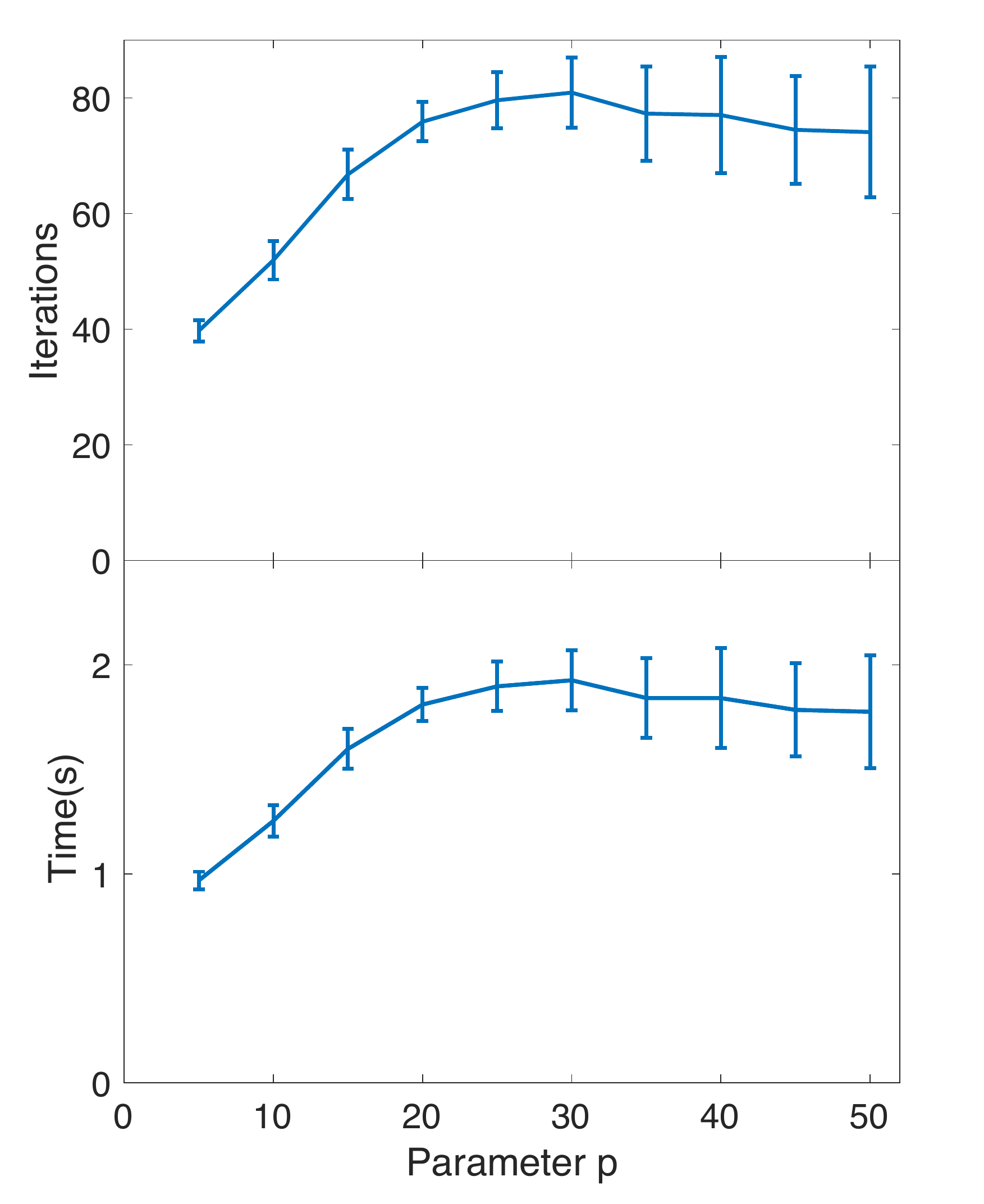}}
   \caption{Graph Instances. Comparing the number of iterations and time taken by our algorithm with the parameters. Averaged over 100 graph samples. Linear solver used : backslash.}
   \label{fig:Graphs}
\end{figure}

\begin{figure}
\centering
   \subfloat[Fixed $p = 8$. Size of matrices: $100k \times (50+100(k-1))$.] {\includegraphics[width=0.23\textwidth]{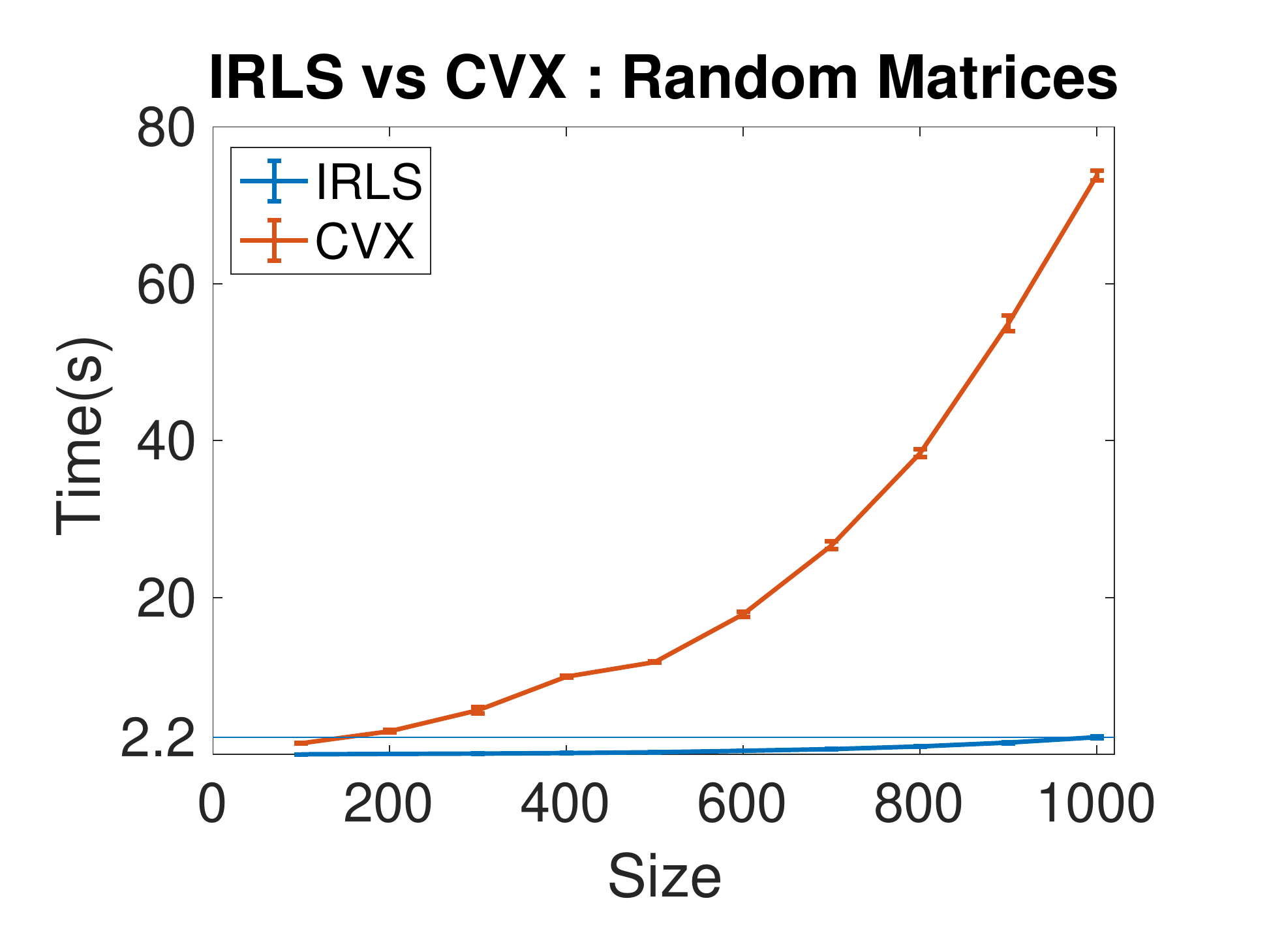}}
  \hfill
  \subfloat[Size of matrices fixed to $500 \times 450$.]{\includegraphics[width=0.23\textwidth]{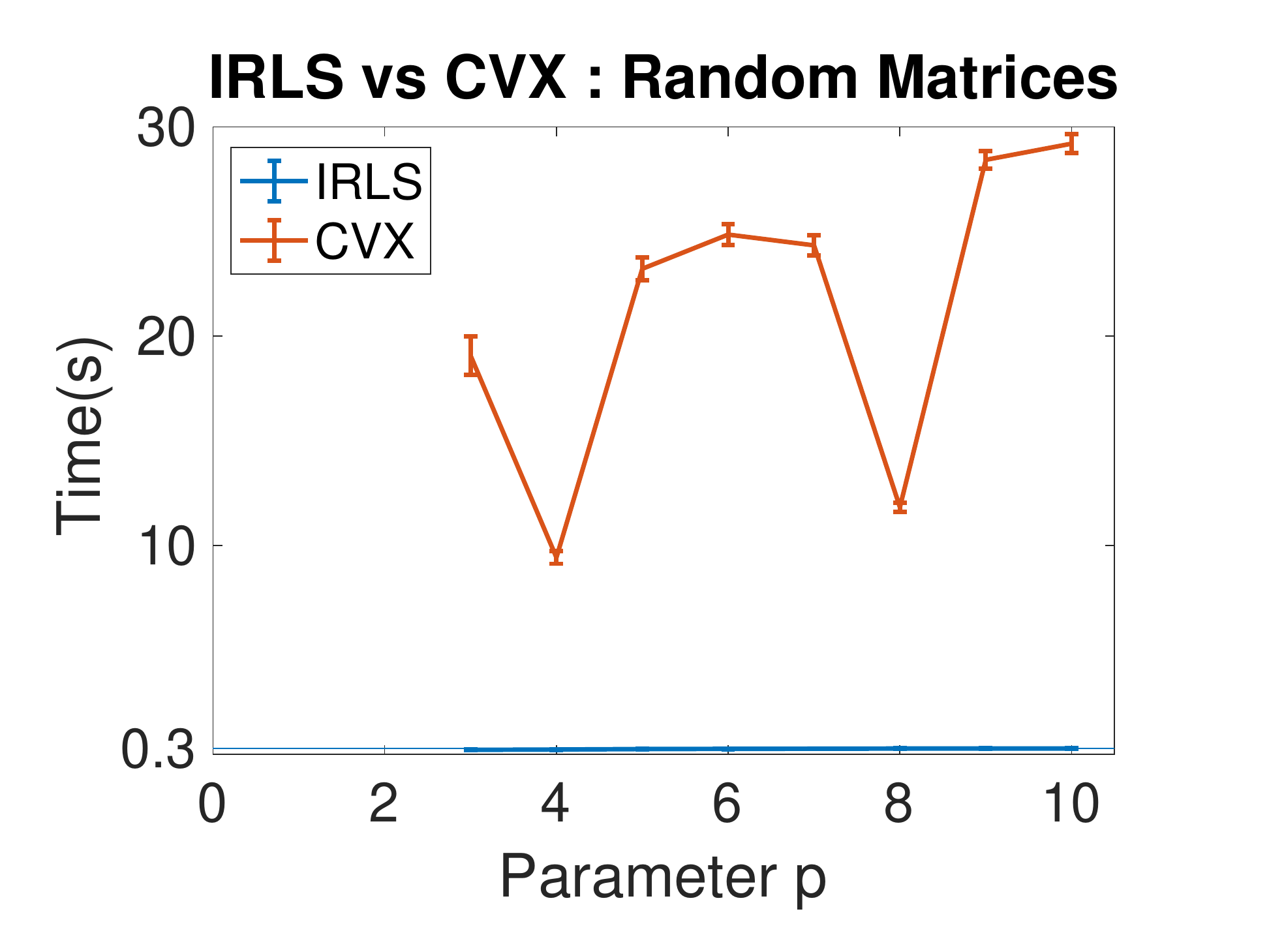}}
  \hfill
  \centering
    \subfloat[Fixed $p = 8$. The number of nodes : $50k, k = 1,2,...,10$.]{\includegraphics[width=0.23\textwidth]{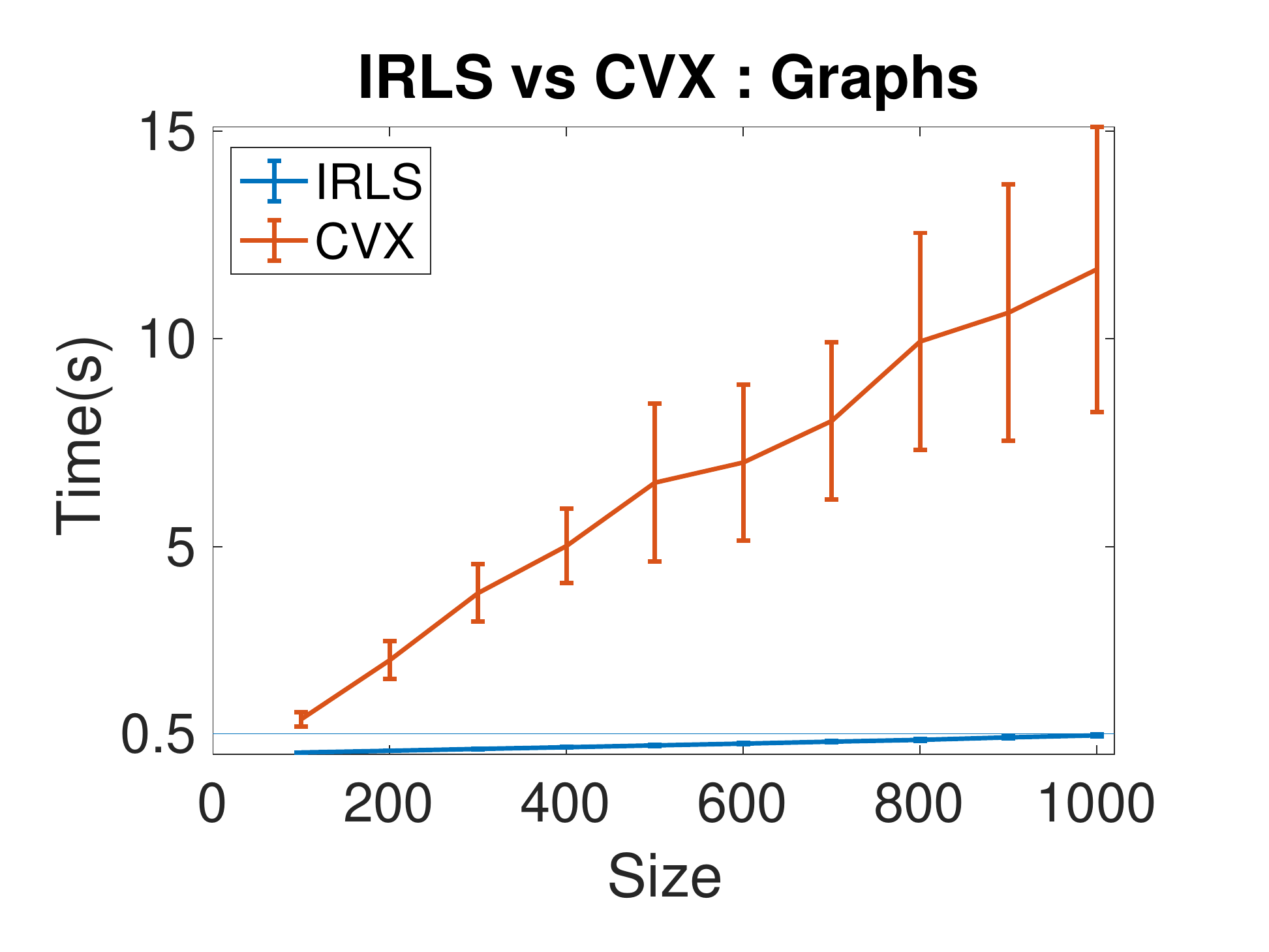}{}}
  \hfill
    \subfloat[Size of graphs fixed to $400$ nodes ( around $2000$ edges). ]{\includegraphics[width=0.23\textwidth]{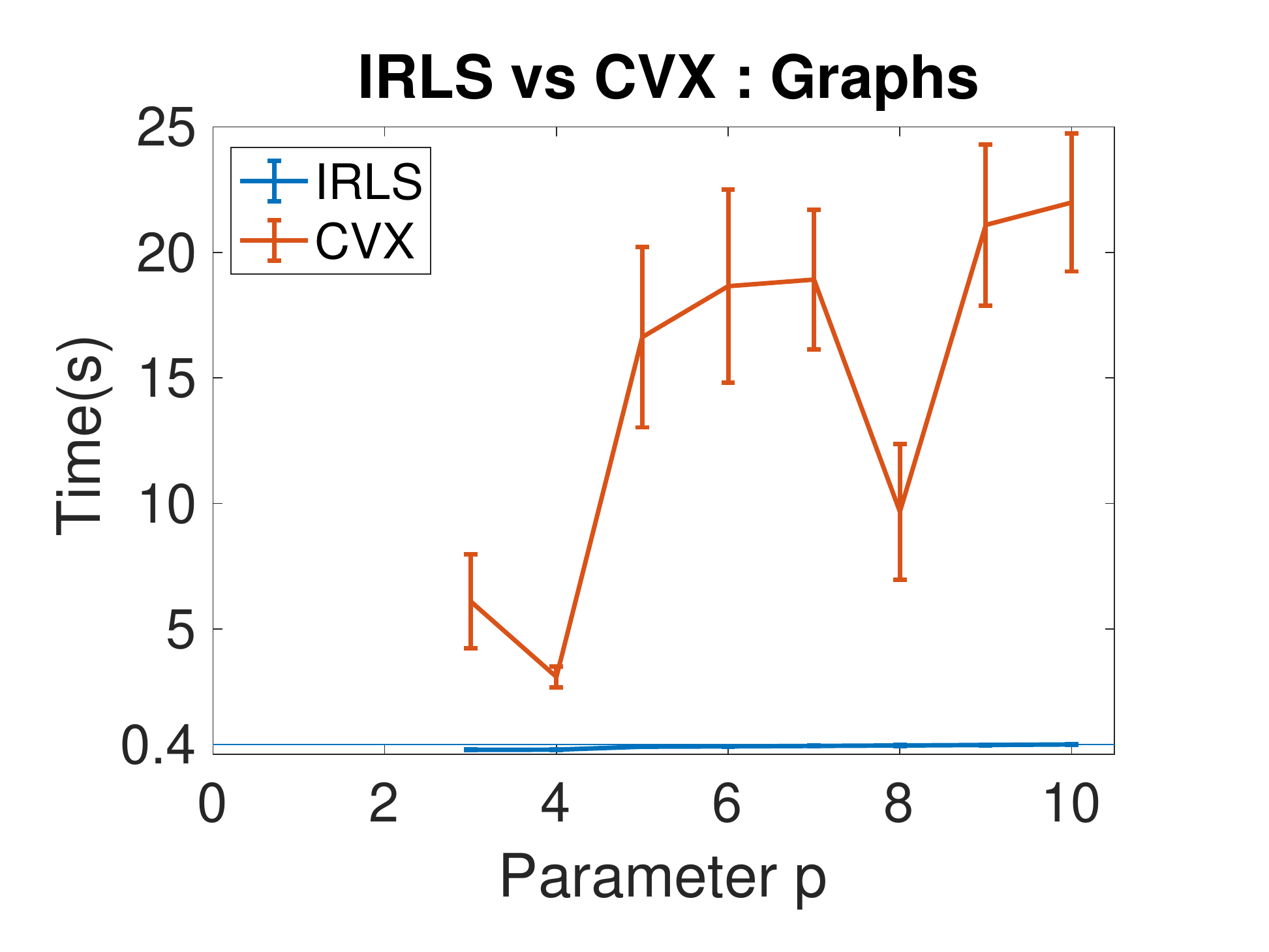}}
 \caption{Averaged over 100 samples. Precision set to $\epsilon = 10^{-8}$.CVX solver used : SDPT3 for Matrices and Sedumi for Graphs.}
 \label{fig:CVXvsIRLS}
\end{figure}

All implementations were done on on MATLAB 2018b on a Desktop ubuntu machine with an Intel Core $i5$-$4570$ CPU @ $3.20GHz \times 4$ processor and 4GB RAM. The two kinds of instances considered are {\it Random Matrices} and {\it Graph instances} for the problem $\min_{x} \norm{\AA\xx-\bb}_p$.

\begin{tight_enumerate}
\item {\bf Random Matrices:} Matrices $\AA$ and $\bb$ are generated randomly i.e., every entry of the matrix is chosen uniformly at random between $0$ and $1$.
\item {\bf Graphs:} Instances are generated as in \textcite{rios2019}.  Vertices are uniform random vectors in $[0,1]^{10}$ and edges are created by connecting the $10$ nearest neighbors. The weight of every edge is determined by a Gaussian function (Eq 3.1,\cite{rios2019}). Around 10 vertices have labels chosen uniformly at random between $0$ and $1$. The problem is to minimize the $\ell_p$ laplacian. Appendix \ref{chap:laplacian} contains details on how to formulate this problem into our standard form. These instances were generated using the code by \textcite{Flores19}.
\end{tight_enumerate}

The performance of $p$-IRLS is compared against  Matlab/CVX solver \cite{cvx,gb08} and the IRLS/homotopy based
implementation from~\textcite{rios2019}. More details on the experiments are in \textcite{AdilPS19} and the plots and specific details of the implementation are included in Figures \ref{fig:Matrices},\ref{fig:IRLSvsNt},\ref{fig:Graphs} and, \ref{fig:CVXvsIRLS}.


%% file: Chapters/Appendix-l2Prob.tex

\section{Solving \texorpdfstring{$\ell_2$}{TEXT} Problems under Subspace Constraints}\label{chap:l2solve}
We will show how to solve general problems of the following form using a linear system solver.
\begin{align*}
\min_{\xx} &\quad \norm{\AA\xx-\bb}_2^2\\
&\CC\xx = \dd.
\end{align*}
We first write the Lagrangian of the problem,
\[
L(\xx,\vvs) =  \min_{\xx} \max_{\vvs} \quad(\AA\xx-\bb)^{\top}(\AA\xx-\bb) + \vvs^{\top}(\dd-\CC\xx)
\]
Using Lagrangian duality and noting that strong duality holds, we can write the above as,
\begin{align*}
L(\xx,\vvs) = & \min_{\xx} \max_{\vvs} \quad(\AA\xx-\bb)^{\top}(\AA\xx-\bb) + \vvs^{\top}(\dd-\CC\xx)\\
= & \max_{\vvs}\min_{\xx} \quad(\AA\xx-\bb)^{\top}(\AA\xx-\bb) + \vvs^{\top}(\dd-\CC\xx).
\end{align*}
We first find $\xx^{\star}$ that minimizes the above objective by setting the gradient with respect to $\xx$ to $0$. We thus have,
\[
\xx^{\star} = (\AA^{\top}\AA)^{-1}\left(\frac{2\AA^{\top}\bb+\CC^{\top}\vvs}{2}\right).
\]
Using this value of $\xx$ we arrive at the following dual program.
\[
L(\vvs) = \max_{\vvs} \quad -\frac{1}{4}\vvs^{\top}\CC(\AA^{\top}\AA)^{-1}\CC^{\top}\vvs -\bb^{\top}\AA(\AA^{\top}\AA)^{-1}\AA^{\top}\bb - \vvs^{\top}\CC(\AA^{\top}\AA)^{-1}\AA^{\top}\bb +\bb^{\top}\bb+\vvs^{\top}\dd,
\]
which is optimized at,
\[
\vvs^{\star}= 2\left(\CC(\AA^{\top}\AA)^{-1}\CC^{\top}\right)^{-1} \left(\dd - \CC(\AA^{\top}\AA)^{-1}\AA^{\top}\bb\right).
\] 
Strong duality also implies that $L(\xx,\vvs^{\star})$ is optimized at $\xx^{\star}$, which gives us,
\[
\xx^{\star} =  (\AA^{\top}\AA)^{-1}\left(\AA^{\top}\bb+\CC^{\top}\left(\CC(\AA^{\top}\AA)^{-1}\CC^{\top}\right)^{-1} \left(\dd - \CC(\AA^{\top}\AA)^{-1}\AA^{\top}\bb\right) \right).
\]

We now note that we can compute $\xx^{\star}$ by solving the following linear systems in order:

\begin{enumerate}
	\item Find inverse of $\AA^{\top}\AA$
	\item $\left(\CC(\AA^{\top}\AA)^{-1}\CC^{\top}\right) x = \left(\dd - \CC(\AA^{\top}\AA)^{-1}\AA^{\top}\bb\right)$
\end{enumerate}

\section{Converting \texorpdfstring{$\ell_p$}{TEXT}-Laplacian Minimization to Regression Form}
\label{chap:laplacian}
Define the following terms:
\begin{tight_itemize}
\item $n$ denote the number of vertices.
\item $l$ denote the number of labels.
\item $\BB$ denote the edge-vertex adjacency matrix.
\item $\gg$ denote the vector of labels for the $l$ labelled vertices.
\item $\WW$ denote the diagonal matrix with weights of the edges.
\end{tight_itemize}
Set $\AA = \WW^{1/p}\BB$ and $\bb = - \BB[:,n:n+l]\gg$. Now $\norm{\AA\xx-\bb}_p^p$ is equal to the $\ell_p$ laplacian and we can use our IRLS algorithm from Chapter 7 to find the $\xx$ that minimizes this.

%% file: Chapters/Appendix-EnergyLemma.tex

\section{Increasing Resistances}
\label{sec:r-inc}

We first prove the following lemma that shows how much $\Psi$ changes with a change in resistance.

\begin{restatable}{lemma}{ckmstResIncrease}
  \label{lem:ckmst:res-increase}
Let $\Dtil = \arg\min_{\AA\Delta = c} \Delta^{\top}\MM^{\top}\MM\Delta + \sum_e \rr_e (\NN\Delta)_e^2$. Then one has for any $\rr'$ and $\rr$ such that $\rr' \geq \rr$,
\[
{\Psi({\rr'})} \geq {\Psi\left({\rr}\right)} + \sum_{e}\left(1-\frac{\rr_e}{\rr'_e}\right) \rr_e (\NN\Dtil)_e^2.
\]
\end{restatable}

\begin{proof} 
For this proof, we use $\RR = Diag(\rr)$.
\[
\Psi(\rr) = \min_{\AA\xx = \cc} \xx^{\top}\MM^{\top}\MM\xx + \xx^{\top}\NN^{\top}\RR\NN\xx.
\]
Constructing the Lagrangian and noting that strong duality holds,
\begin{align*}
\Psi(\rr) &= \min_{\xx}\max_{\yy} \quad \xx^{\top}\MM^{\top}\MM\xx + \xx^{\top}\NN^{\top}\RR\NN\xx + 2\yy^{\top}(\cc-\AA\xx)\\
& =\max_{\yy} \min_{\xx} \quad \xx^{\top}\MM^{\top}\MM\xx + \xx^{\top}\NN^{\top}\RR\NN\xx + 2\yy^{\top}(\cc-\AA\xx).
\end{align*}
Optimality conditions with respect to $\xx$ give us,
\[
2\MM^{\top}\MM\xx^{\star} + 2\NN^{\top}\RR\NN\xx^{\star} = 2\AA^{\top}\yy.
\]
Substituting this in $\Psi$ gives us,
\[
\Psi(\rr) = \max_{\yy}\quad 2\yy^{\top}\cc  - \yy^{\top}\AA \left(\MM^{\top}\MM+ \NN^{\top}\RR\NN\right)^{-1} \AA^{\top}\yy.
\]
Optimality conditions with respect to $\yy$ now give us,
\[
2\cc  =  2 \AA \left(\MM^{\top}\MM+ \NN^{\top}\RR\NN\right)^{-1} \AA^{\top} \yy^{\star},
\]
which upon re-substitution gives,
\[
\Psi(\rr) =  \cc^{\top}\left(\AA \left(\MM^{\top}\MM+ \NN^{\top}\RR\NN\right)^{-1} \AA^{\top}\right)^{-1} \cc.
\]
We also note that 
\begin{equation}\label{eq:optimizer}
\xx^{\star} = \left(\MM^{\top}\MM+ \NN^{\top}\RR\NN\right)^{-1}\AA^{\top}\left(\AA \left(\MM^{\top}\MM+ \NN^{\top}\RR\NN\right)^{-1} \AA^{\top}\right)^{-1}\cc.
\end{equation}
We now want to see what happens when we change $\rr$. Let $\RR$ denote the diagonal matrix with entries $\rr$ and let $\RR' = \RR+\SS$, where $\SS$ is the diagonal matrix with the changes in the resistances. We will use the following version of the Sherman-Morrison-Woodbury formula multiple times,
\[
(\XX + \UU\CC\VV)^{-1} = \XX^{-1} - \XX^{-1}\UU(\CC^{-1} + \VV\XX^{-1}\UU)^{-1}\VV\XX^{-1}.
\]
We begin by applying the above formula for $\XX = \MM^{\top}\MM+ \NN^{\top}\RR\NN$, $\CC = \II$, $\UU = \NN^{\top}\SS^{1/2}$ and $\VV = \SS^{1/2}\NN$. We thus get,
\begin{multline}
\left(\MM^{\top}\MM+ \NN^{\top}\RR'\NN\right)^{-1} = \left(\MM^{\top}\MM+ \NN^{\top}\RR\NN\right)^{-1} -  \left(\MM^{\top}\MM+ \NN^{\top}\RR\NN\right)^{-1}\NN^{\top}\SS^{1/2} \\
\left(\II + \SS^{1/2}\NN \left(\MM^{\top}\MM+ \NN^{\top}\RR\NN\right)^{-1}\NN^{\top}\SS^{1/2}\right)^{-1}\SS^{1/2}\NN \left(\MM^{\top}\MM+ \NN^{\top}\RR\NN\right)^{-1}.
\end{multline}
We next claim that, 
\[
\II + \SS^{1/2}\NN \left(\MM^{\top}\MM+ \NN^{\top}\RR\NN\right)^{-1}\NN^{\top}\SS^{1/2} \preceq \II + \SS^{1/2}\RR^{-1}\SS^{1/2} ,
\]
which gives us,
\begin{multline}
\left(\MM^{\top}\MM+ \NN^{\top}\RR'\NN\right)^{-1} \preceq \left(\MM^{\top}\MM+ \NN^{\top}\RR\NN\right)^{-1} -  \\ \left(\MM^{\top}\MM+ \NN^{\top}\RR\NN\right)^{-1}\NN^{\top}\SS^{1/2}(\II + \SS^{1/2}\RR^{-1}\SS^{1/2})^{-1}\SS^{1/2}\NN \left(\MM^{\top}\MM+ \NN^{\top}\RR\NN\right)^{-1}.
\end{multline}
This further implies,
\begin{multline}
\AA\left(\MM^{\top}\MM+ \NN^{\top}\RR'\NN\right)^{-1}\AA^{\top} \preceq \AA\left(\MM^{\top}\MM+ \NN^{\top}\RR\NN\right)^{-1} \AA^{\top} -  \\ \AA\left(\MM^{\top}\MM+ \NN^{\top}\RR\NN\right)^{-1}\NN^{\top}\SS^{1/2}(\II + \SS^{1/2}\RR^{-1}\SS^{1/2})^{-1}\SS^{1/2}\NN \left(\MM^{\top}\MM+ \NN^{\top}\RR\NN\right)^{-1}\AA^{\top}.
\end{multline}
We apply the Sherman-Morrison formula again for, $\XX =\AA\left( \MM^{\top}\MM+ \NN^{\top}\RR\NN\right)^{-1}\AA^{\top}$, $\CC = -(\II + \SS^{1/2}\RR^{-1}\SS^{1/2})^{-1}$, $\UU = \AA\left(\MM^{\top}\MM+ \NN^{\top}\RR\NN\right)^{-1}\NN^{\top}\SS^{1/2}$ and $\VV = \SS^{1/2}\NN \left(\MM^{\top}\MM+ \NN^{\top}\RR\NN\right)^{-1}\AA^{\top}$. Let us look at the term $\CC^{-1} + \VV\XX^{-1}\UU$.
\[
-\left(\CC^{-1} + \VV\XX^{-1}\UU\right)^{-1} =    \left(\II + \SS^{1/2}\RR^{-1}\SS^{1/2} -  \VV\XX^{-1}\UU\right)^{-1}  \succeq   (\II + \SS^{1/2}\RR^{-1}\SS^{1/2})^{-1}.
\]
Using this, we get,
\[
\left(\AA\left(\MM^{\top}\MM+ \NN^{\top}\RR'\NN\right)^{-1}\AA^{\top}\right)^{-1} \succeq \XX^{-1} + \XX^{-1}\UU(\II + \SS^{1/2}\RR^{-1}\SS^{1/2})^{-1}\VV\XX^{-1},
\]
which on multiplying by $\cc^{\top}$ and $\cc$ gives,
\[
\Psi(\rr') \geq \Psi(\rr) +  \cc^{\top} \XX^{-1}\UU(\II + \SS^{1/2}\RR^{-1}\SS^{1/2})^{-1}\VV\XX^{-1} \cc.
\]
We note from Equation \eqref{eq:optimizer} that $\xx^{\star} = \left(\MM^{\top}\MM+ \NN^{\top}\RR\NN\right)^{-1}\AA^{\top} \XX^{-1}\cc$.
We thus have,
\begin{align*}
\Psi(\rr') &\geq \Psi(\rr) +  \left(\xx^{\star}\right)^{\top} \NN^{\top}\SS^{1/2} (\II + \SS^{1/2}\RR^{-1}\SS^{1/2})^{-1}\SS^{1/2}\NN\xx^{\star}\\
& = \Psi(\rr) + \sum_e \left(\frac{\rr'_e -\rr_e}{\rr'_e}\right)\rr_e (\NN\xx^{\star})^2_e.
\end{align*}
\end{proof}